\documentclass[11pt]{article}
\usepackage{geometry} 

\usepackage{amsmath,amsthm,amsfonts,dsfont,mathtools}
\usepackage{amssymb}
\usepackage{mathrsfs}
\usepackage{booktabs}
\usepackage{hyperref}
\usepackage{adjustbox}
\usepackage{graphicx}

\usepackage[dvipsnames]{xcolor}

\usepackage{colortbl}
\usepackage{braket}
\usepackage[ruled]{algorithm2e}

\usepackage{wrapfig}
\usepackage{makecell}
\usepackage{setspace}
\usepackage{multirow}
\usepackage[numbers,sort]{natbib}

\usepackage{enumitem}
\setitemize{noitemsep,topsep=3pt,parsep=3pt,partopsep=3pt}

\definecolor{darkgreen}{rgb}{0,0.5,0}
\usepackage{hyperref}
\hypersetup{
    unicode=false,          
    colorlinks=true,        
    linkcolor=red,          
    citecolor=darkgreen,        
    filecolor=magenta,      
    urlcolor=cyan           
}
\usepackage[capitalize, nameinlink]{cleveref}

\newcommand{\A}{\mathcal{A}}
\newcommand{\VV}{\mathcal{V}}

\newcommand{\SSS}{\mathcal{S}}
\newcommand{\DDD}{\mathcal{D}}
\newcommand{\EE}{\mathcal{E}}
\newcommand{\CC}{\mathcal{C}}

\newcommand{\M}{\mathcal{M}}

\newcommand{\LCL}{\mathsf{LCL}}
\newcommand{\LOCAL}{\mathsf{LOCAL}}
\newcommand{\CONGEST}{\mathsf{CONGEST}}

\renewcommand{\P}{\mathsf{P}}

\newcommand{\EXP}{\mathsf{EXP}}

\DeclareMathOperator{\poly}{poly}

\newcommand{\trim}{\mathsf{trim}}
\newcommand{\flexSCC}{\mathsf{flexible}\text{-}\mathsf{SCC}}
\newcommand{\flexibility}{\mathsf{flexibility}}

\newcommand{\indeg}{\deg_{\mathsf{in}}}
\newcommand{\outdeg}{\deg_{\mathsf{out}}}

\newcommand{\PathPi}{\Pi^{\mathrm{path}}}
\newcommand{\PathCC}{\mathcal{C}^{\mathrm{path}}}

\newcommand{\VC}[1]{{V}_{#1}^\mathsf{C}}
\newcommand{\VR}[1]{{V}_{#1}^\mathsf{R}}
\newcommand{\SigmaC}[1]{{\Sigma}_{#1}^\mathsf{C}}
\newcommand{\SigmaR}[1]{{\Sigma}_{#1}^\mathsf{R}}

\newcommand{\sR}{\mathsf{R}}
\newcommand{\sC}{\mathsf{C}}

\newcommand{\rake}{{\sf rake}}
\newcommand{\compress}{{\sf compress}}

\crefname{theorem}{Theorem}{Theorems}
\Crefname{lemma}{Lemma}{Lemmas}
\Crefname{figure}{Figure}{Figures}
\Crefname{claim}{Claim}{Claims}
\Crefname{observation}{Observation}{Observations}

\newtheorem{theorem}{Theorem}[section]
\newtheorem{lemma}{Lemma}[section]

\newtheorem{definition}{Definition}[section]

\definecolor{myblue}{HTML}{0088cc}
\definecolor{myorange}{HTML}{f26924}
 
\graphicspath{{figs/}}

\title{Efficient Classification of Locally Checkable \\ Problems in Regular Trees}

\begin{document}

\date{}
\author{
Alkida Balliu
\\
\normalsize Gran Sasso Science Institute \\
\normalsize \texttt{alkida.balliu@gssi.it} \\
\and
Sebastian Brandt
\\
\normalsize CISPA Helmholtz Center for Information Security \\
\normalsize \texttt{brandt@cispa.de} \\
\and
Yi-Jun Chang
\\
\normalsize National University of Singapore \\
\normalsize \texttt{cyijun@nus.edu.sg} \\
\and
Dennis Olivetti
\\
\normalsize Gran Sasso Science Institute \\
\normalsize \texttt{dennis.olivetti@gssi.it} \\
\and
Jan Studen\'y
\\
\normalsize Aalto University \\
\normalsize \texttt{jan.studeny@aalto.fi} \\
\and
Jukka Suomela
\\
\normalsize Aalto University \\
\normalsize \texttt{jukka.suomela@aalto.fi} \\
}

\maketitle

\begin{abstract}
We give practical, efficient algorithms that automatically determine the asymptotic distributed round complexity of a given locally checkable graph problem in the $[\Theta(\log n), \Theta(n)]$ region, in two settings. We present one algorithm for unrooted regular trees and another algorithm for rooted regular trees. The algorithms take the description of a locally checkable labeling problem as input, and the running time is polynomial in the size of the problem description. The algorithms decide if the problem is solvable in $O(\log n)$ rounds. If not, it is known that the complexity has to be $\Theta(n^{1/k})$ for some $k = 1, 2, \dotsc$, and in this case the algorithms also output the right value of the exponent $k$.

In rooted trees in the $O(\log n)$ case we can then further determine the exact complexity class by using algorithms from prior work; for unrooted trees the more fine-grained classification in the $O(\log n)$ region remains an open question.
\end{abstract}

\section{Introduction}

We give practical, efficient algorithms that automatically determine the asymptotic distributed round complexity of a given \emph{locally checkable} graph problem in \emph{rooted or unrooted regular trees} in the $[\Theta(\log n), \Theta(n)]$ region, for both $\LOCAL$  and $\CONGEST$ models, see \cref{sec:prelim} for the precise definitions. In these cases, the distributed round complexity of any locally checkable problem is known to fall in one of the classes shown in \cref{fig:contrib} \cite{CP19timeHierarchy, CKP19exponential, BBOS18almostGlobal, Chang20, BHOS19HomogeneousLCL, brandt21trees, smallmessages}. Our algorithms are able to distinguish between all higher complexity classes from $\Theta(\log n)$ to $\Theta(n)$.

\subsection{State of the art}

Since 2016, there has been a large body of work studying the possible complexities of $\LCL$ problems. After an impressive sequence of works, the complexity landscape of $\LCL$ problems on bounded-degree general graphs, trees, and paths is now  well-understood.  For example, it is known that there are no $\LCL$s with deterministic complexity between $\omega(\log^\ast n)$ and $o(\log n)$. The proofs of some of the complexity gaps implies that the design of asymptotically optimal distributed algorithms can be \emph{automated} in certain settings, leading to a series of research studying the computational complexity of automated design of asymptotically optimal distributed algorithms. See \cref{sec:relatedwork} for more details.

The most recent paper~\cite{balliu21rooted-trees} in this line of research presented an algorithm that takes as input the description of an $\LCL$ problem defined in \emph{rooted regular trees} and classifies the problem into one of the four  complexity classes $O(1)$, $\Theta(\log^* n)$, $\Theta(\log n)$, and $n^{\Theta(1)}$.  The classification applies to both the  $\LOCAL$  and $\CONGEST$ models of distributed computing, both for randomized and deterministic algorithms.

To illustrate the setting of locally checkable problems in rooted regular trees, consider, for example, the following problem, which is meaningful for rooted binary trees:
\begin{quote}
Each node is labeled with 1 or 2. If the label of an internal node is 1, exactly one of its two children must have label 1, and if the label of an internal node is 2, both of its children must have label 1.
\end{quote}
We can represent it in a concise manner as a problem $\CC = \{1 : 12, 2 : 11\}$, where $a : bc$ indicates that a node of label $a$ can have its two children labeled with $b$ and $c$, in some order. We can take such a description, feed it to the algorithm from \cite{balliu21rooted-trees}, and it will output that this problem requires $\Theta(\log n)$ rounds in order to be solved in a rooted tree with $n$ nodes.

\begin{figure}
{\footnotesize
\raggedright
\newcommand{\mybrace}[4]{$\color{#1}\underbrace{\quad}_{\mathclap{\substack{\mathstrut\smash{\text{#2}}\\\mathstrut\smash{\text{#3}}\\\mathstrut\smash{\text{#4}}}}}$}
\newcommand{\priorE}{\mybrace{myblue}{prior}{work}{($\EXP$)}}
\newcommand{\priorP}{\mybrace{myblue}{prior}{work}{($\P$)}}
\newcommand{\newP}{\mybrace{myorange}{this}{work}{($\P$)}}
\newcommand{\open}{\mybrace{black!60!white}{unknown}{}{}}
\newcommand{\remarkline}[1]{\multicolumn{16}{@{}l@{}}{#1}\\[3ex]}
\begin{tabular}{@{}l@{}c@{ }c@{ }c@{ }c@{ }c@{ }c@{ }c@{ }c@{ }c@{ }c@{ }c@{ }c@{ }c@{ }c@{ }c@{}}
\remarkline{(a) Rooted regular trees in deterministic and randomized $\CONGEST$ and $\LOCAL$:}
\hspace*{2ex}
& $\Theta(1)$
& \priorE
& $\Theta(\log^* n)$
& \priorE
& $\Theta(\log n)$
& \priorP
& \ldots
& \newP
& $\Theta(n^{1/3})$
& \newP
& $\Theta(n^{1/2})$
& \newP
& $\Theta(n)$
\end{tabular}\\[4ex]
\begin{tabular}{@{}l@{}c@{ }c@{ }c@{ }c@{ }c@{ }c@{ }c@{ }c@{ }c@{ }c@{ }c@{ }c@{ }c@{ }c@{ }c@{}}
\remarkline{(b) Unrooted regular trees in deterministic $\CONGEST$ and $\LOCAL$:}
\hspace*{2ex}
& $\Theta(1)$
& \open
& $\Theta(\log^* n)$
& \open
& $\Theta(\log n)$
& \newP
& \ldots
& \newP
& $\Theta(n^{1/3})$
& \newP
& $\Theta(n^{1/2})$
& \newP
& $\Theta(n)$
\end{tabular}\\[4ex]
\begin{tabular}{@{}l@{}c@{ }c@{ }c@{ }c@{ }c@{ }c@{ }c@{ }c@{ }c@{ }c@{ }c@{ }c@{ }c@{ }c@{ }c@{ }c@{ }c@{}}
\remarkline{(c) Unrooted regular trees in randomized $\CONGEST$ and $\LOCAL$:}
\hspace*{2ex}
& $\Theta(1)$
& \open
& $\Theta(\log^* n)$
& \open
& $\Theta(\log \log n)$
& \open
& $\Theta(\log n)$
& \newP
& \ldots
& \newP
& $\Theta(n^{1/3})$
& \newP
& $\Theta(n^{1/2})$
& \newP
& $\Theta(n)$
\end{tabular}
\caption{The most efficient algorithms for the classification of distributed round complexities. In the figure we show all possible complexity classes. Each gap between two classes corresponds to a natural decision problem: given a locally checkable problem, determine on which side of the gap its complexity is. For each gap we indicate whether a practical algorithm was provided already by prior work \cite{balliu21rooted-trees}, whether it is first presented in this work, or whether the existence of such a routine is still an open question. The figure also indicates whether the algorithms are in $\P$ (polynomial time in the size of the problem description) or in $\EXP$ (exponential time in the size of the problem description).}\label{fig:contrib}
}
\end{figure}

\subsection{What was missing}

What the prior algorithm from \cite{balliu21rooted-trees} can do  is classifying a given problem into one of the four main complexity classes $O(1)$, $\Theta(\log^* n)$, $\Theta(\log n)$, and $n^{\Theta(1)}$. However, if the complexity is $n^{\Theta(1)}$, we do not learn whether its complexity is, say, $\Theta(n)$ or $\Theta(\sqrt{n})$ or maybe $\Theta(n^{1/10})$. There are locally checkable problems of complexity $\Theta(n^{1/k})$ for every $k = 1, 2, \dotsc$, and there have not been any \emph{practical} algorithm that would determine the value of the exponent $k$ for any given problem.

Furthermore, the algorithm from \cite{balliu21rooted-trees} is only applicable in rooted regular trees, while the case of unrooted trees is perhaps even more interesting.

It has been known that the problem of distinguishing between e.g.\ $\Theta(n)$ and $\Theta(\sqrt{n})$ is \emph{in principle} decidable, due to the algorithm of~\cite{Chang20}. This algorithm is, however, best seen as a theoretical construction. To the best of our knowledge, nobody has implemented it, there are no plans of implementing it, and it seems unlikely that one could classify any nontrivial problem with it using any real-world computer, due to its doubly exponential time complexity. This is the missing piece that we provide in this work.

\subsection{Contributions and motivations}

We present polynomial-time algorithms that determine not only whether the round complexity of a given $\LCL$ problem is $\Theta(n^{1/k})$ for some $k$, but they also determine the exact value of $k$. We give one algorithm for the case of unrooted trees (\cref{sec:unrooted}) and one algorithm for the case of rooted trees (\cref{sec:rooted}). 

Our algorithms not only determine the asymptotic round complexity, but they also output a description of a distributed algorithm attaining this complexity. If the given $\LCL$ problem $\Pi$ has optimal complexity $\Theta(n^{1/k})$, then our algorithms will output a description of a deterministic distributed algorithm that solves $\Pi$ in $O(n^{1/k})$ rounds in the $\CONGEST$ model. Similarly, if the given $\LCL$ problem $\Pi$ has optimal complexity $O(\log n)$, then our algorithms will output a description of a deterministic distributed algorithm that solves $\Pi$ in $O(\log n)$ rounds in the $\CONGEST$ model.

 We have implemented both algorithms for the case of $3$-regular trees, the proof-of-concept implementations are freely available online,\footnote{\url{https://github.com/jendas1/poly-classifier}} and they work fast also in practice.

From a practical point of view,
together with prior work from \cite{balliu21rooted-trees}, there is now a practical algorithm that is able to \emph{completely determine} the complexity of any $\LCL$  problem in rooted regular trees.\footnote{Even though some algorithms in \cite{balliu21rooted-trees} are exponential in the size of the description of the problem, they are nevertheless very efficient in practice. In fact, the authors of \cite{balliu21rooted-trees} have implemented them for the case of binary rooted trees and they are indeed very fast in practice \cite{AnonymousRepo}.
} In the case of unrooted regular trees deciding between the lower complexity classes below $o(\log n)$ remains an open question.

From a theoretical point of view, this work \emph{significantly expands} the class of $\LCL$ problems whose  optimal complexity is known to be decidable in polynomial time. See \cref{fig:contrib} for a summary of the current state of the art on the classification of $\LCL$ complexities for regular trees, showing where the new algorithms are applicable and where the state of the art is given by existing results.

We note that the problem of determining the optimal complexity of an $\LCL$ problem is computationally hard in general: It is undecidable in general~\cite{NaorStockmeyer95}, EXPTIME-hard even for bounded-degree trees~\cite{Chang20}, and PSPACE-hard even for paths and cycles with input labels~\cite{balliu19lcl-decidability}. Hence, in order to understand whether polynomial-time algorithms are even possible, we must restrict our consideration to restricted cases, such as $\LCL$s with no inputs defined on regular trees. In fact, it is known that it is possible to use $\LCL$s with no inputs defined on non-regular trees to encode $\LCL$s with inputs, and hence, by allowing inputs, or constraints that depend on the degree of the nodes, we would make decidability at least PSPACE-hard. 

\subparagraph{Motivations}
Studying $\LCL$s is interesting because, on the one hand, this class of problems is large enough to contain a significant fraction of problems that are commonly studied in the context of the $\LOCAL$ model (e.g., $(\Delta+1)$-coloring, $(2\Delta-1)$-edge coloring, $\Delta$-coloring, weak $2$-coloring, maximal matching, maximal independent set, sinkless orientation, many other orientation problems, edge splitting problems, locally maximal cut, defective colorings, \ldots), but, on the other hand, it is restricted enough so that we can prove interesting results about them, such as decidability and complexity gaps. Moreover, techniques used to prove results on $\LCL$s have been already shown to be extremely useful outside the $\LCL$ context: for example, all recent results about lower bounds for locally checkable problems in the unbounded degree case---e.g., for MIS, maximal matching, ruling sets, and other fundamental problems---use techniques that originally were introduced in the context of $\LCL$s~\cite{BBHORS19MMlowerBound, trulytight, BBOrs, BBKO21, BBKU22}.

In this work, we restrict our attention to the case of regular trees. The study of $\LCL$s on trees is related with our understanding of graph problems in the general setting. Actually, for many problems of interest, unrooted regular trees are hard instances, and hence understanding the complexity of $\LCL$s on trees could help us in understanding the complexity of problems in general unbounded-degree graphs. In fact, a relatively new and promising technique called \emph{round elimination} has been used to prove tight lower bounds for interesting graph problems such as maximal matchings, maximal independent sets, and ruling sets, even if, for now, we are only able to apply this technique for proving lower bounds on trees \cite{Brandt19RE,Olivetti2019REtor, BBHORS19MMlowerBound, BBOrs, binary_lcls, trulytight, BBKO21, BBKU22}.

As for the more restrictive setting of \emph{regular} trees, we would like to point out that many natural $\LCL$ problems have the same optimal complexity in both bounded-degree trees and regular trees. This includes, for example, the  $k$-coloring problem. For any tree $T$ whose maximum degree is at most $\Delta$, we may consider the $\Delta$-regular tree $T^\ast$ which is the result of appending degree-$1$ nodes to all nodes $v$ in $T$ with $1 < \deg(v) < \Delta$ to increase the degree of $v$ to $\Delta$. We may locally simulate $T^\ast$ in the network $T$. As any proper $k$-coloring of $T^\ast$ restricting to $T$ is also a proper $k$-coloring, this reduces the $k$-coloring problem on bounded-degree trees to the same problem on regular trees, showing that the $k$-coloring problem has the same optimal complexity in both graph classes. 
More generally, if an $\LCL$ problem  $\Pi$ has the property that removing degree-1 nodes preserves the correctness of a solution, then $\Pi$ has the same optimal complexity in both bounded-degree trees and regular trees, so our results in this work also apply to these $\LCL$s on bounded-degree trees.

\section{Related work}\label{sec:relatedwork}
Locally Checkable Labeling problems have been introduced by Naor and Stockmeyer \cite{NaorStockmeyer95}, but the class of locally checkable problems has been studied in the distributed setting even before (e.g., in the context of self-stabilisation \cite{AfekKY97}). For many locally checkable problems, researchers have been trying to understand the exact time complexity, and while in many cases upper bounds have been known since the 80s, matching lower bound have been discovered only recently. Examples of this line of research relate to the problems of colorings, matchings, and independent sets, see e.g. \cite{ColeVishkin86, Linial92, Luby1986, panconesi01simple, fraigniaud16local, fischer17improved, BBHORS19MMlowerBound, RG20NetDecomposition, MausT20, BBKU22, Ghaffari-Kuhn}.

In parallel, there have been many works that tried to understand these problems from a \emph{complexity theory} point of view, trying to develop general techniques for classifying problems, understanding which complexities can actually exist, and developing generic algorithmic techniques to solve whole classes of problems at once. In particular, a broad class\footnote{For example, our definition of $\LCL$ does not allow an infinite number of labels, so it does not capture some locally checkable problems such as fractional matching.} of locally checkable problems, called \emph{Locally Checkable Labelings} ($\LCL$s), has been studied in the $\LOCAL$ model of distributed computing, which will be formally defined later.

\subparagraph{Paths and cycles} The first graph topologies on which promising results have been proved are paths and cycles. In these graphs, we now know that there are problems with the following \emph{three} possible time complexities:
\begin{itemize}
    \item $O(1)$: this class contains, among others, trivial problems, e.g.~problems that require every node to output the same label.
    \item $\Theta(\log^* n)$: this class contains, for example, the $3$-coloring problem \cite{ColeVishkin86, Linial92, Naor1991}. 
    \item $\Theta(n)$: this class contains hard problems, for example the problem of consistently orient the edges of a cycle, or the $2$-coloring problem.
\end{itemize}
For $\LCL$s in paths and cycles, we know that there are no other possible complexities, that is, there are \emph{gaps} between the above classes. In other words, there are no $\LCL$s with a time complexity that lies between $\omega(1)$ and $o(\log^* n)$ \cite{NaorStockmeyer95}, and no $\LCL$s with a time complexity that lies between $\omega(\log^* n)$ and $o(n)$ \cite{CKP19exponential}. These results hold also for \emph{randomized} algorithms, and they are constructive: if for example we find a way to design an $O(\log n)$-rounds randomized algorithm for a problem, then we can automatically convert it into an $O(\log^* n)$-round deterministic algorithm.

Moreover, in paths and cycles, given an $\LCL$ problem, we can \emph{decide} its time complexity. In particular, it turns out that for problems with no inputs defined on directed cycles, deciding the complexity of an $\LCL$ is as easy as drawing a diagram and staring at it for few seconds~\cite{Brandt2017}. This result has later been extended to undirected cycles with no inputs~\cite{lcls_on_paths_and_cycles}. Unfortunately, as soon as we consider $\LCL$s where the constraints of the problem may depend on the given inputs, decidability becomes much harder, and it is now known to be PSPACE-hard~\cite{balliu19lcl-decidability}, even for paths and cycles.

\subparagraph{Trees}
Another class of graphs that has been studied quite a lot is the one containing \emph{trees}. While there are still problems with complexities $O(1)$, $\Theta(\log^* n)$, and $\Theta(n)$, there are also additional complexity classes, and sometimes here randomness can help. For example, there are problems that require $\Theta(\log n)$ rounds for both deterministic and randomized algorithms, while there are problems, like \emph{sinkless orientation}, that require $\Theta(\log n)$ rounds for deterministic algorithms and $\Theta(\log \log n)$ rounds for randomized ones \cite{BFHKLRSU16,CKP19exponential,ghaffari17distributed}. Moreover, there are problems with complexity $\Theta(n^{1/k})$, for any natural number $k \ge 1$ \cite{CP19timeHierarchy}. It is known that these are the only possible time complexities in trees \cite{CP19timeHierarchy, CKP19exponential, BBOS18almostGlobal, Chang20, BHOS19HomogeneousLCL, brandt21trees}. In \cite{smallmessages}, it has been shown that the same results hold also in a more restrictive model of distributed computing, called $\CONGEST$ model, and that for any given problem, its complexities in the $\LOCAL$ and in the $\CONGEST$ model, on trees, are actually the same. 

Concerning decidability, the picture is not as clear as in the case of paths and cycles. As discussed in the introduction, it is decidable, \emph{in theory}, if a problem requires $n^{\Omega(1)}$ rounds, and, in that case, it is also decidable to determine the exact exponent \cite{CP19timeHierarchy,Chang20}, but the algorithm is very far from being practical, and in this work we address exactly this issue. Moreover, for \emph{lower} complexities, the problem is still open.
Different works tried to tackle this issue by considering restricted cases. In \cite{binary_lcls}, authors showed that it is indeed possible to achieve decidability in some cases, that is, when problems are restricted to the case of unrooted regular trees, where leaves are unconstrained, and the problem uses only \emph{two} labels.
Then, promising results have been achieved in \cite{balliu21rooted-trees}, where it has been shown that, if we consider \emph{rooted} trees, then we can decide the complexity of $\LCL$s even for $n^{o(1)}$ complexities. Unfortunately, it is very unclear if such techniques can be used to solve the problem in the general case.
In fact, we still do not know if it is decidable whether a problem can be solved in $O(1)$ rounds or it requires $\Omega(\log^* n)$ rounds, and it is not known if it is decidable whether a problem can be solved in $O(\log^* n)$ rounds or it requires $\Omega(\log n)$ for deterministic algorithms and $\Omega(\log \log n)$ for randomized ones. These two questions are very important, and understanding them may also help in understanding problems that are not restricted to regular trees of bounded degree. This is because, as already mentioned before, for many problems it happens that unrooted regular trees are hard instances, and studying the complexity of problems in these instances may give insights for understanding problems in the general setting.

\subparagraph{General graphs}
In general graphs, many more $\LCL$ complexities are possible. For example, there is a gap similar to the one between $\omega(1)$ and $o(\log^* n)$ of trees, but now it holds only up to $o(\log \log^* n)$, and we know that there are problems in the region between $\Omega(\log \log^* n)$ and $o(\log^* n)$. In fact, for any rational $\alpha \ge 1$, it is possible to construct problems with complexity $\Theta(\log^\alpha \log^* n)$ \cite{BHKLOS18lclComplexity}.
A similar statement holds for complexities between $\Omega(\log n)$ and $O(n)$ \cite{BHKLOS18lclComplexity,BBOS18almostGlobal}. 

There are still complexity regions in which we do not know if there are problems or not. For example, while it is known that any problem that has randomized complexity $o(\log n)$ can be sped up to $O(T_{\mathrm{LLL}})$ \cite{CP19timeHierarchy}, where $T_{\mathrm{LLL}}$ is the distributed complexity of the constructive version of the Lov\'asz $\LOCAL$  Lemma, the exact value of $T_{\mathrm{LLL}}$ is unknown, and we only known that it lies between $\Omega(\log \log n)$ and $O(\poly \log \log n)$ \cite{BFHKLRSU16,CPS17DistrLLL, FischerGhaffari17LLL, RG20NetDecomposition}. Another problem that falls in this region is the $\Delta$-coloring problem, for which we still do not know the exact complexity. 

Another open question regards the role of randomness. In general graphs, we know that randomness can also help outside the $O(\log n)$ region \cite{BBOS20paddedLCL}, but we still do not know exactly when it can help and how much.

In general graphs, unfortunately, determining the complexity of a given $\LCL$ problem is undecidable. In fact, we know that this question is undecidable even on grids \cite{NaorStockmeyer95}.

\section{Preliminaries}\label{sec:prelim}

\subparagraph{Graphs}
Let $G = (V,E)$ be a graph. We denote with $n = |V|$ the number of nodes of $G$, with $\Delta$ the maximum degree of $G$, and with $\deg(v)$, for $v \in V$, the degree of $v$. If $G$ is a directed graph, we denote with $\indeg(v)$ and $\outdeg(v)$, the indegree and the outdegree of $v$, respectively. The radius-$r$ neighborhood of a node $v$ is defined to be the subgraph of $G$ induced by the nodes at distance at most $r$ from $v$.

\subparagraph{Model of computing}
In the $\LOCAL$ model of distributed computing, the network is represented with a graph $G = (V,E)$, where the nodes correspond to computational entities, and the edges correspond to communication links. In this model, the computational power of the nodes is unrestricted, and nodes can send arbitrarily large messages to each other.

This model is synchronous, and computation proceeds in rounds. Nodes all start the computation at the same time, and at the beginning they know $n$ (the total number of nodes), $\Delta$ (the maximum degree of the graph), and a unique ID in $\{1,\ldots,n^c\}$, for some constant $c \ge 1$, assigned to them. Then, the computation proceeds in rounds, and at each round nodes can send (possibly different) messages to each neighbor, receive messages, and perform some $\LOCAL$  computation. 

At the end of the computation, each node must produce its own part of the solution. For example, in the case of the $(\Delta+1)$-coloring problem, each node must output its own color, that must be different from the ones of its neighbors. The time complexity is measured as the worst case number of rounds required to terminate, and it is typically expressed as a function of $n$, $\Delta$, and $c$.

\section{Technical overview}\label{sec:overview}

Our new results build on several techniques developed in previous works~\cite{balliu21rooted-trees,lcls_on_paths_and_cycles} designing polynomial-time algorithms that determine the distributed complexity of $\LCL$ problems. In this section, we first give a brief overview of these techniques, then we discuss how in this paper we build upon them and obtain our new results. The aim of this section is to present the intuition behind the results. To keep the discussion at a high level, the presentation here will be a bit imprecise, see \cref{sec:unrooted,sec:rooted} for the precise statements of our results.

\subsection{The high-level framework}
Existing algorithms for deciding the complexity of a given $\LCL$ problem are often based on the following approach.
\begin{enumerate}
    \item Define some combinatorial property $P$ of $\LCL$ problems.
    \item Show that computing $P(\Pi)$ for a given problem $\Pi$ can be done efficiently.
    \item Show that $\Pi$ is in a certain complexity class if and only if $P(\Pi)$ holds.
\end{enumerate}

As discussed in~\cite{Brandt2017,lcls_on_paths_and_cycles}, any $\LCL$  $\Pi$ on directed paths can be viewed as a regular language. Taking the corresponding non-deterministic automaton, we obtain a directed graph $G(\Pi)$ that represents $\Pi$ on directed paths. 

For example, the maximal independent set problem can be described as the automaton with states $V=\{00,01,10\}$ and transitions $E = \{00 \rightarrow 01, 01 \rightarrow 10, 10 \rightarrow 00, 10 \rightarrow 01\}$. Each state corresponds to a possible labeling of the two endpoints $u$ and $v$ of a directed edge $u \rightarrow v$.  Each transition describes a valid configuration of two neighboring directed edges $u \rightarrow v$ and $v \rightarrow w$.

It has been shown~\cite{balliu21rooted-trees,Brandt2017,brandt2021local,lcls_on_paths_and_cycles} that in several cases the distributed complexity of an $\LCL$ can be characterized by simple graph properties of  $G(\Pi)$, even if the underlying graph class is much more complicated than directed paths. The precise definition of $G(\Pi)$ will depend on the choice of the $\LCL$ formalism. 

\subsection{Paths and cycles}
It was shown in~\cite{Brandt2017,lcls_on_paths_and_cycles} that the distributed complexity and solvability of $\Pi$ on paths and cycles can be characterized by simple graph properties of $G(\Pi)$. In particular, $\Pi$ on directed cycles is solvable in $O(\log^* n)$ rounds if and only if $G(\Pi)$ contains a node $v$ that is \emph{path-flexible}, in the sense that there exists a number $K$ such that, in $G(\Pi)$, there is a length-$k$ returning walk for $v$, for each $k \geq K$. 
If such a path-flexible node $v$ exists in $G(\Pi)$, then $\Pi$ on directed cycles can be solved in $O(\log^\ast n)$ rounds in the following manner.
\begin{enumerate}
    \item 
In $O(\log^\ast n)$ rounds, compute an independent set $I$ such that the distance between the nodes in $I$ is at least $K$ and at most $2K$.
\item  Fix the labels for the nodes in $I$ according to the path-flexible node $v$ in $G(\Pi)$.
\item  By the path-flexibility of $v$, this partial labeling can be completed into a correct complete labeling. 
\end{enumerate}

For example, in the automaton for maximal independent set described above, the state $01$ is flexible, as for each $k \geq 5$, there is a length-$k$ walk starting and ending at $01$, so a maximal independent set can be found in $O(\log^\ast n)$ rounds on directed cycles via the above algorithm.

The above characterization can be generalized to both paths and cycles, undirected and directed, after some minor modifications, see~\cite{lcls_on_paths_and_cycles} for the details. For further examples of representing $\LCL$s as automata and how the round complexity of an $\LCL$ can be inferred from basic properties of its associated automaton, see~\cite[Fig.~3]{Brandt2017} and~\cite[Fig.~1 and 3]{lcls_on_paths_and_cycles}.

\subsection{\boldmath The  \texorpdfstring{$O(\log n)$}{O(log n)} complexity class in regular trees}
Subsequently, it was shown in~\cite{balliu21rooted-trees,brandt2021local} that the class of $O(\log n)$-round solvable $\LCL$ problems on rooted and unrooted regular trees can be characterized in a similar way, based on the notion of path-flexibility in the directed graph $G(\Pi)$. To keep the discussion at a high level, we do not discuss the difference between rooted and unrooted trees here. Roughly speaking, $\Pi$ can be solved in $O(\log n)$ rounds on rooted or unrooted regular trees if and only if there exists a subset of labels $S$ such that, if we restrict $\Pi$ to $S$, then its corresponding directed graph is strongly connected and contains a path-flexible node. Such a set $S$ of labels is also called a \emph{certificate} for $O(\log n)$-round solvability.\footnote{Although the certificate described in~\cite{balliu21rooted-trees} also includes the steps in the construction of $S$, the set $S$ alone suffices to certify that $\Pi$ can be solved in $O(\log n)$ rounds, as the $O(\log n)$-round algorithm described in~\cite{balliu21rooted-trees} uses only $S$.}

A key property of such a directed graph is that there exists a number $K$ such that, for each pair of nodes $(u,v)$, and for each integer $k \geq K$, there is a length-$k$ walk from $u$ to $v$ (here we allow the possibility of $u=v$). The property can be described in the following more intuitive manner. For any path of length at least $K$, regardless of how we fix the labels of its two endpoints using $S$, it is always possible to complete the partial labeling into a correct labeling~w.r.t.~$\Pi$ of the entire path using only labels in $S$.

The intuition behind such a characterization is the fact~\cite{CP19timeHierarchy} that all $\LCL$s solvable in  $O(\log n)$ rounds on bounded-degree trees can be solved in a canonical way based on \emph{rake-and-compress decompositions}.  Roughly speaking, a rake-and-compress process is a procedure that decomposes a tree by iteratively removing degree-$1$ nodes (rake) and removing degree-$2$ nodes (compress). This process partitions the set of nodes into several parts:
\[V = \VR{1} \cup \VC{1} \cup \VR{2} \cup \VC{2}   \cup \cdots \cup \VR{L},\] 
where $\VR{i}$ is the set of nodes removed by the rake operation in the $i$th iteration and $\VC{i}$ is the set of nodes removed by the compress operation in the $i$th iteration. It can be shown that $L = O(\log n)$~\cite{Miller1985}.

There are several variants of a rake-and-compress process. Here the considered variant is such that, in the compress operation, a degree-$2$ node $v$ is removed if $v$ belongs to a path whose length is at least $\ell$, so we may assume that the connected components in the subgraph induced by $\VC{i}$ are paths with length at least $\ell$.  

Let $\Pi$ be any $\LCL$ problem satisfying the combinatorial characterization for $O(\log n)$-round solvability discussed above, and let the set of labels $S$ be a certificate for $O(\log n)$-round solvability.
By setting $\ell = K$ in the property of the combinatorial characterization, we may obtain an $O(\log n)$-round algorithm solving the given $\LCL$ problem $\Pi$ using only the labels in $S$. The high-level idea is that we can label the tree in an order that is the reverse of the one of the rake-and-compress procedure: $\VR{L},  \ldots, \VC{2}, \VR{2}, \VC{1}, \VR{1}$, as we observe that the property of the combinatorial characterization discussed above ensures that any correct labeling of $\VR{L} \cup \cdots \cup \VR{i}$  can be extended to a correct labeling of $\VR{L} \cup \cdots \cup \VR{i} \cup \VC{i-1}$ and similarly any  correct labeling of $\VR{L} \cup \cdots \cup \VC{i}$  can be extended to a correct labeling of $\VR{L} \cup \cdots \cup \VC{i} \cup \VR{i}$.

The requirement that $\Pi$ is an $\LCL$ problem defined on \emph{regular} trees is \emph{critical} in the above approach, as this requirement ensures that for each non-leaf node, the set of constraints is the same, so we do not need to worry about the possibility for different nodes in the tree to have different sets of constraints in  $\Pi$. Indeed, if we allow nodes of different degrees to have different sets of constraints, then the problem of determining the distributed complexity of an $\LCL$ in bounded-degree trees becomes EXPTIME-hard~\cite{Chang20}.

\subsection{The polynomial complexity region in regular trees}
In this work, we will extend the above approach to cover all complexity classes in the $[\Theta(\log n), \Theta(n)]$ region. By~\cite{BBOS18almostGlobal,Chang20,CP19timeHierarchy}, we know that the possible complexity classes in this region are $\Theta(\log n)$ and $\Theta(n^{1/k})$ for all positive integers $k$. Similar to the complexity class $O(\log n)$, any $\LCL$ problem $\Pi$ solvable in $O(n^{1/k})$ rounds can be solved in a canonical way  in  $O(n^{1/k})$ rounds using a variant of rake-and-compress decomposition~\cite{Chang20}.

Specifically, $\Pi$ is  $O(n^{1/k})$-round solvable if and only if it can be solved in a canonical way  using a rake-and-compress decomposition, where in each iteration, we perform $\gamma = O(n^{1/k})$ rake operations and one compress operation. Similar to the case of complexity class $O(\log n)$, in the compress operation, a degree-$2$ node $v$ is removed if  $v$ belongs to a path whose length is at least $\ell$, where $\ell = O(1)$ is some sufficiently large number depending only on the $\LCL$ problem $\Pi$.
It can be shown~\cite{Chang20} that by selecting $\gamma = O(n^{1/k})$ to be large enough, the number of layers $L$ in the decomposition $V = \VR{1} \cup \VC{1} \cup \VR{2} \cup \VC{2}   \cup \cdots \cup \VR{L}$ is $k$, and such a decomposition can be computed in $O(n^{1/k})$ rounds.  

To derive a certificate for $O(n^{1/k})$-round solvability based on the result of~\cite{Chang20}, we will need to take into consideration the following properties about the variant of the rake-and-compress decomposition described above. 
\begin{itemize}
    \item The number of layers $L = k$ is now a \emph{finite} number independent of the size of the graph~$n$. For technical reasons, this means that the certificate for $O(n^{1/k})$-round solvability cannot be based on a single set of labels $S$, as the certificate for $O(\log n)$-round solvability~\cite{balliu21rooted-trees,brandt2021local}. We need to consider the possibility that different sets of labels are used for different layers in the design of the  certificate for $O(n^{1/k})$-round solvability.
    \item The number of rake operations for a layer can be unbounded as $n$ goes to infinity. That is, $\VR{i}$  is no longer an independent set, and each connected component in the subgraph induced by $\VR{i}$ can be a very large tree.  
\end{itemize}

\subparagraph{The certificate}
Our certificate for $O(n^{1/k})$-round solvability will be based on the notion of a \emph{good} sequence of sets of labels. The definition of a good sequence relies on two functions on a set of labels: $\trim$ and $\flexSCC$. As we will later see, these two functions correspond to rake and compress, respectively. Given an $\LCL$ problem $\Pi$ and a set of labels $S$,   $\trim(S)$ and $\flexSCC$ are defined as follows.
\begin{itemize}
    \item $\trim(S)$ is the subset of $S$ resulting from removing all labels $\sigma \in S$ meeting the following conditions: There exists some number $i$ such that if the root of the complete regular tree $T$ of height $i$ is labeled by $\sigma$, then we are not able to complete the labeling of $T$  using only labels in $S$ such that the overall labeling is correct w.r.t.~$\Pi$.
    \item $\flexSCC(S)$ is a collection of disjoint subsets of $S$ defined as follows. Consider the directed graph representing the $\LCL$ problem $\Pi$ restricted to $S$. Let $\flexSCC(S)$ be the set of strongly connected components that have a path-flexible node. The intuition behind this definition is similar to the intuition behind the certificate for $O(\log n)$-round solvability.
\end{itemize}

We briefly explain the connection between $\trim$ and rake. Suppose we want to find a correct labeling of a regular tree $T$ using only the labels in $S$. If a label $\sigma$ is in $\trim(S)$, then $\sigma$ can only be used in places that are sufficiently close to a leaf. To put it another way, if we do a large number of rakes to $T$, then the labels in $\trim(S)$ can only be used to label the nodes that removed due to a rake operation. 

The connection between $\flexSCC$ to compress is due to the fact that the nodes removed due to a compress operation form long paths, and we know that in order to label long paths efficiently in $O(\log^\ast n)$ rounds, it is necessary to use labels corresponding to path-flexible nodes, due to the existing automata-theoretic characterization~\cite{Brandt2017,lcls_on_paths_and_cycles} of round complexity of $\LCL$s on paths and cycles.

We say that a sequence  $(\SigmaR{1}, \SigmaC{1}, \SigmaR{2}, \SigmaC{2}, \ldots, \SigmaR{k})$
is \emph{good} if it satisfies the following rules, where $\Sigma$ is the set of all labels of $\Pi$.
\begin{align*}
    \SigmaR{i} & =
    \begin{cases}
    \trim(\Sigma) & \text{if $i=1$},\\
    \trim(\SigmaC{i-1}) &  \text{if $i>1$}.\\
    \end{cases}\\
    \SigmaC{i} &\in \flexSCC(\SigmaR{i}).\\
    \SigmaR{k} &\neq \emptyset.
\end{align*}
The only nondeterminism in the above rules is the choice of $\SigmaC{i} \in \flexSCC(\SigmaR{i})$ for each $i$. We will show that such a sequence exists if and only if the underlying $\LCL$ problem can be solved in $O(n^{1/k})$ rounds. Intuitively, $\SigmaR{i}$ represents the set of labels that are eligible to label the nodes in $\VR{i}$, and similarly $\SigmaC{i}$ represents the set of labels that are eligible to label the nodes in $\VC{i}$.

\subparagraph{The classification}
The notion of a good sequence allows us to classify the complexity classes in the region $[\Theta(\log n), \Theta(n)]$.
Specifically, we define the \emph{depth} $d_\Pi$ of an $\LCL$ problem $\Pi$  as the largest $k$ such that a good sequence $(\SigmaR{1}, \SigmaC{1}, \SigmaR{2}, \SigmaC{2}, \ldots, \SigmaR{k})$ exists. 
If there is no good sequence, then we set $d_\Pi = 0$. 
If there is a good sequence $(\SigmaR{1}, \SigmaC{1}, \SigmaR{2}, \SigmaC{2}, \ldots, \SigmaR{k})$ for each positive integer $k$, then we set $d_\Pi = \infty$.
We will show that $d_\Pi$ characterizes the distributed complexity of $\Pi$ in the following manner.

\begin{itemize}
    \item If $d_\Pi = 0$, then $\Pi$ is unsolvable in the sense that there exists a regular tree such that there is no correct solution of $\Pi$ on this rooted tree. This follows from the definition of $\trim$ and the observation that  $d_\Pi = 0$ if $\trim(\Sigma) = \emptyset$.
    \item If $d_\Pi = k$ is a positive integer, then the distributed complexity of $\Pi$ is $\Theta(n^{1/k})$.
    \item If $d_\Pi = \infty$, then $\Pi$ can be solved in $O(\log n)$ rounds. If we can have a good sequence that is arbitrarily long, then there must be a \emph{fixed point} $S$ in the sequence such that $\trim(S) = S$ and $\flexSCC(S) = \{S\}$, because $\SigmaR{1} \supseteq \SigmaC{1} \supseteq \cdots \supseteq \SigmaR{k}$. We will show that the fixed point $S$ qualifies to be a certificate for $O(\log n)$-round solvability. 
\end{itemize}

The fixed point phenomenon explains why the  notion of good sequence was not 
needed in~\cite{balliu21rooted-trees,brandt2021local}, as the existence of a fixed point for the case  $\Pi$ is $O(\log n)$-round solvable implies that we may apply the same strategy according to the fixed point to label each layer of the rake-and-compress decomposition to solve $\Pi$ in $O(\log n)$ rounds.

\subparagraph{The proof ideas}

To show the correctness and efficiency of our characterization, we need to do the following.
\begin{description}
    \item[Upper bound:] Given a good sequence $(\SigmaR{1}, \SigmaC{1}, \SigmaR{2}, \SigmaC{2}, \ldots, \SigmaR{k})$, show that there exists an $O(n^{1/k})$-round algorithm solving $\Pi$. Therefore, $d_\Pi = k$ implies $O(n^{1/k})$-round solvability.
    \item[Lower bound:] Given an $o(n^{1/k})$-round algorithm solving $\Pi$, show that a  good sequence $(\SigmaR{1}, \SigmaC{1}, \SigmaR{2}, \SigmaC{2}, \ldots, \SigmaR{k+1})$ exists.  Therefore, $d_\Pi = k$ implies $\Omega(n^{1/k})$-round solvability.
    \item[Efficiency:] Design a polynomial-time algorithm that computes $d_\Pi$ for any given  description of an $\LCL$ problem $\Pi$.
\end{description}

The upper bound proof is relatively simple. Similar to the certificate $O(\log n)$-round solvability, we just need to show that $\Pi$ can be solved in $O(n^{1/k})$ rounds using rake-and-compress decompositions given that a good sequence $(\SigmaR{1}, \SigmaC{1}, \SigmaR{2}, \SigmaC{2}, \ldots, \SigmaR{k})$ exists. 

The lower bound proof is much more complicated. Given an algorithm $\mathcal{A}$ solving $\Pi$ in $t = o(n^{1/k})$ rounds, we will consider a tree $G$ that is a result of a hierarchical combination of complete trees and paths of length greater than $t$. Intuitively, $G$ is chosen to be the fullest possible tree that can be partitioned into  $V = \VR{1} \cup \VC{1} \cup \VR{2} \cup \VC{2}   \cup \cdots \cup \VR{k+1}$ with a rake-and-compress decomposition of~\cite{Chang20} with $L = k+1$ layers. We will prove by induction that if we take $\SigmaR{i}$ to be the set of possible output labels of  $\mathcal{A}$  for $\VR{i} \cup \VC{i}   \cup \cdots \cup \VR{k+1}$  and take $\SigmaC{i}$ to be the set of possible output labels of  $\mathcal{A}$  for $\VC{i} \cup \VR{i+1}   \cup \cdots \cup \VR{k+1}$, then $(\SigmaR{1}, \SigmaC{1}, \SigmaR{2}, \SigmaC{2}, \ldots, \SigmaR{k+1})$ must be a good sequence. In particular, the non-emptiness of $\SigmaR{k+1}$ follows from the correctness of $\mathcal{A}$.

To design a polynomial-time algorithm computing $d_\Pi$, we recall that the only nondeterminism in the  rules for a good sequence is the choice of $\SigmaC{i} \in \flexSCC(\SigmaR{i})$, so we will just do a brute-force search for all possibilities. Although this seems very inefficient, we recall that $\flexSCC(\SigmaR{i})$ is a collection of disjoint subsets of $\SigmaR{i}$, so the summation of the size of all sets of labels considered in each level is at most  the total number of labels $|\Sigma|$ in $\Pi$. The number of levels we need to explore is also bounded, as $\SigmaR{1} \supseteq \SigmaC{1} \supseteq \cdots \supseteq \SigmaR{k}$. If $k$ exceeds $|\Sigma|$, then we know that there is a fixed point $\SigmaR{i}$ such that   $\SigmaR{i} = \SigmaC{i} = \SigmaR{i+1} = \SigmaC{i+1} = \cdots$, so $d_\Pi = \infty$.

\subparagraph{The differences between rooted and unrooted trees}
The high-level proof strategy presented in this technical overview applies to both rooted and unrooted regular trees, showing that these two graph classes behave very similarly in the complexity region $[\Theta(\log n), \Theta(n)]$. 
There are still some technical differences between rooted and unrooted trees. 

\begin{itemize}
    \item The formalisms for representing $\LCL$ problems are different for rooted and unrooted trees. In the case of rooted trees, the problem can refer to orientations. For example, what is permitted for a parent can be different from what is permitted for a child. Instead of specifying node and edge configurations, we follow \cite{balliu21rooted-trees} and specify what are permitted multisets of child labels for each node label.
    \item  For the upper bound, we need to generalize the rake-and-compress decomposition of~\cite{Chang20}  so that it is applicable in rooted trees.
    \item For the lower bound, the lower bound graph for unrooted trees  does not work for the rooted trees. Roughly speaking, this is because the presence of edge orientation increases the symmetry breaking capability of nodes, so some indistinguishability arguments in the lower bound proof for unrooted trees do not work for rooted trees. Therefore, we will need to consider a different approach for crafting the lower bound graph for rooted trees.
\end{itemize}

\section{Unrooted trees}\label{sec:unrooted}

In this section, we give a polynomial-time-computable characterization of $\LCL$ problems for regular unrooted trees with complexity $O(\log n)$ or $\Theta(n^{1/k})$  for any positive integer $k$.

\subsection{Locally checkable labeling for unrooted trees} 

A $\Delta$-regular tree is a tree where the degree of each node is either $1$ or $\Delta$.
An $\LCL$ problem for $\Delta$-regular unrooted trees is defined as follows.

\begin{definition}[$\LCL$ problems for regular unrooted trees]\label{def-lcl-unrooted}
For unrooted trees, an $\LCL$ problem $\Pi=(\Delta, \Sigma,\VV,\EE)$ is defined by the following  components.
\begin{itemize}
    \item $\Delta$ is a positive integer specifying the maximum degree. 
    \item $\Sigma$ is a finite set of labels.
    \item $\VV$ is a set of size-$\Delta$ multisets of labels in $\Sigma$ specifying the node constraint.
    \item $\EE$ is a set of size-$2$ multisets of labels in $\Sigma$ specifying the edge constraint.
\end{itemize}
\end{definition}

We call a size-$\Delta$ multiset $C$ of labels in $\Sigma$ a \emph{node configuration}.
A node configuration  $C$  is correct with respect to $\Pi=(\Delta, \Sigma,\VV,\EE)$  if $C \in \VV$.
We call a size-$2$ multiset $D$ of labels in $\Sigma$ an \emph{edge configuration}.
An edge  configuration  $D$  is correct with respect to $\Pi=(\Delta, \Sigma,\VV,\EE)$  if $D \in \EE$.
We define the correctness criteria for a  labeling of  a $\Delta$-regular tree  in \cref{def-correctness-unrooted}.

\begin{definition}[Correctness criteria]\label{def-correctness-unrooted}
Let $G=(V,E)$ be a tree whose maximum degree is at most $\Delta$. For each edge $e=\{u,v\}$ in the tree, there are two half-edges  $(u,e)$ and  $(v,e)$. A solution of $\Pi=(\Delta,\Sigma,\VV,\EE)$ on $G$ is a labeling that assigns a label in $\Sigma$ to each half-edge in $G$.
\begin{itemize}
    \item For each node $v\in V$ with $\deg(v) = \Delta$ its node configuration $C$ is  the multiset of $\Delta$ half-edge labels of $(v,e_1)$, $(v, e_2)$, $\ldots$, $(v, e_{\Delta})$, where $e_1, e_2, \ldots, e_\Delta$ are the $\Delta$ edges incident to $v$. We say that the labeling is locally-consistent on $v$ if $C \in \VV$.
    \item For each edge $e= \{u,v\}\in E$, its edge configuration $D$ is the  multiset of two half-edge labels of $(u,e)$ and $(v,e)$. We say that the labeling is locally-consistent on $e$ if $D \in \EE$.
\end{itemize}
 The labeling  is a correct solution if it is locally-consistent on all $v\in V$ with $\deg(v) = \Delta$ and all $e \in E$.
\end{definition}

In other words, a labeling of $G=(V,E)$ is correct if the edge configuration for each $e\in E$ is correct  and  the node configuration for each $v\in V$ with $\deg(v) = \Delta$ is correct.  All nodes whose degree is not $\Delta$ are unconstrained.

 Although  $\Pi=(\Delta,\Sigma,\VV,\EE)$ is defined for $\Delta$-regular unrooted trees, \cref{def-correctness-unrooted} applies to all trees whose maximum degree is at most $\Delta$.
 We emphasize that all nodes $v$ whose degree is not $\Delta$ are \emph{unconstrained} in that there is no requirement about the node configuration of $v$.
 Nevertheless, we may focus on $\Delta$-regular unrooted trees without loss of generality. The reason is that for any unrooted tree $G$ whose maximum degree is at most $\Delta$, we may consider the unrooted tree $G^\ast$ which is the result of appending degree-$1$ nodes to all nodes $v$ in $G$ with $1 < \deg(v) < \Delta$ to increase the degree of $v$ to $\Delta$. This only blows up the number of nodes by at most a $\Delta$ factor. 
 We claim that the asymptotic optimal round complexity of $\Pi$ is the same in both $G$ and $G^\ast$.
 Any correct solution of $\Pi$ on $G^\ast$ restricted to $G$ is a correct solution of  $\Pi$ on $G$, as all nodes whose degree is not $\Delta$ are unconstrained. Therefore, if we have an algorithm for $\Pi$ in $\Delta$-regular unrooted trees, then the same algorithm also allows us to solve $\Pi$ in unrooted trees with maximum degree $\Delta$ in the same asymptotic round complexity.

\begin{definition}[Complete trees of height $i$]\label{def-complete-trees} We define the rooted trees $T_i$ and $T_i^\ast$ recursively as follows.
 \begin{itemize}
     \item $T_0$ is the trivial tree with only one node.
     \item $T_i$ is the result of appending $\Delta-1$ trees $T_{i-1}$ to the root $r$.
     \item $T_i^\ast$ is the result of appending $\Delta$ trees $T_{i-1}$ to the root $r$.
 \end{itemize}
\end{definition} 

Observe that $T_i^\ast$ is the unique maximum-size tree of maximum degree $\Delta$ and height $i$. All nodes within distance $i-1$ to the root $r$ in  $T_i^\ast$ have degree $\Delta$. All nodes whose distance to $r$ is exactly $i$ are degree-$1$ nodes. Although  $T_i$ and $T_i^\ast$ are defined as rooted trees, they can also be viewed as unrooted trees.

\begin{definition}[Trimming] 
Given an $\LCL$ problem $\Pi=(\Delta,\Sigma,\VV,\EE)$  and a subset $\SSS \subseteq \VV$ of node configurations, we define $\trim(\SSS)$ as the set of all node configurations $C \in \SSS$ such that for each $i \geq 1$ it is possible to find a correct labeling of $T_i^\ast$ such that the node configuration of the root is $C$ and the node configurations of the remaining degree-$\Delta$ nodes are in $\SSS$.
\end{definition}

In the definition, note that if for some $i \geq 1$ it is not possible to find such a labeling of $T_i^\ast$, then it is also not possible for any larger $i$. The reason is that if such a labeling for larger $i$ exists, then by taking subgraph, we obtain such a labeling for of $T_i^\ast$. Here we use the fact that nodes  by taking subgraph, and using the fact that all nodes whose degree is not $\Delta$ are unconstrained. 

Intuitively, $\trim(\SSS)$ is the subset of $\SSS$ resulting from removing all node configurations in $\SSS$ that are not usable in a correct labeling of a sufficiently large $\Delta$-regular tree using only node configurations in $\SSS$.

In fact, given any tree $G$ of maximum degree $\Delta$ and a node $v$ of degree $\Delta$ in $G$, after labeling the half-edges surrounding $v$ using a node configuration in $\trim(\SSS)$, it is always possible to extend this labeling to a complete correct labeling of $G$ using only node configurations  in $\trim(\SSS)$. Such a labeling extension is possible due to \cref{lem-trimming}.

\begin{lemma}[Property of trimming]\label{lem-trimming}
Let  $\SSS \subseteq \VV$ such that  $\trim(\SSS) \neq \emptyset$. 
For each node configuration $C \in \trim(\SSS)$ and each label $\sigma \in C$, there exist a node configuration $C' \in \trim(\SSS)$ and a label $\sigma' \in C'$ such that the multiset $\{\sigma, \sigma'\}$ is in $\EE$.
\end{lemma}
\begin{proof}
Assuming that such $C'$ and $\sigma'$ do not exist,  we derive a contradiction as follows.
We pick $s$ to be the smallest number such that there is no correct labeling of $T_s^\ast$ where the node configuration of the root $r$ is in  $\SSS \setminus \trim(\SSS)$ and the node configuration of each remaining degree-$\Delta$ node of $T_s^\ast$ is in $\SSS$.   Such a number $s$ exists due to the definition of $\trim$.

Now consider a correct labeling of $T_{s+1}^\ast$ where the node configuration of the root $r$ is $C$ and the node configuration of each remaining degree-$\Delta$ node is in $\SSS$. Such a correct labeling exists due to the fact that $C \in \trim(\SSS)$.  Our assumption on the non-existence of $C'$ and $\sigma'$ implies that the node configuration $\tilde{C}$ of one child $w$ of the root $r$  of $T_{s+1}^\ast$ must be in $\SSS \setminus \trim(\SSS)$. However, the radius-$s$ neighborhood  of $w$ in  $T_{s+1}^\ast$ is isomorphic to $T_{s}^\ast$ rooted at $w$. Since the node configuration of $w$ is in  $\SSS \setminus \trim(\SSS)$, our choice of $s$ implies that the labeling of the radius-$s$ neighborhood  of $w$ cannot be correct, which is a contradiction.
\end{proof}

\subparagraph{Path-form of an LCL problem}
Given an $\LCL$ problem $\Pi=(\Delta,\Sigma,\VV,\EE)$ and a subset $\SSS \subseteq \VV$ of node configurations, we define 
\[\DDD_\SSS = \text{the set of all size-$2$ multisets $D$ such that  $D$ is a  sub-multiset of $C$ for some $C \in \SSS$}.\]

To understand the intuition behind the definition $\DDD_\SSS$, define the length-$k$ hairy path $H_k$ as the result obtained by starting from a length-$k$ path $P = (v_1, v_2, \ldots, v_{k+1})$ and then adding degree-$1$ nodes to make $\deg(v_i) = \Delta$ for all $1 \leq i \leq k+1$. If our task is to label hairy paths using node configurations in $\SSS$, then this task is identical to labeling paths using node configurations in $\DDD_\SSS$.
In other words, the $\LCL$ problem $(\Delta, \Sigma, \SSS, \EE)$ on hairy paths is equivalent to the $\LCL$ problem $(2, \Sigma, \DDD_\SSS, \EE)$ on paths. Hence $(2, \Sigma, \DDD_\SSS, \EE)$ is the \emph{path-form} of $(\Delta, \Sigma, \SSS, \EE)$.

\subparagraph{Automaton for the path-form of an LCL problem}
Given a set $\DDD$ of size-$2$ multisets whose elements are in $\Sigma$, we define the directed graph $\M_{\DDD}$ as follows. The node set $V(\M_{\DDD})$  of $\M_{\DDD}$ is the set of all pairs $(a,b) \in \Sigma^2$  such that the multiset $\{a,b\}$ is in $\DDD$. 
The edge set $E(\M_{\DDD})$ of $\M_{\DDD}$ is defined as follows.
For any two pairs $(a,b)\in V(\M_{\DDD})$  and $(c,d)\in V(\M_{\DDD})$, we add a directed edge $(a,b) \rightarrow (c,d)$ if the multiset $\{b,c\}$ is an edge configuration in $\EE$. Note that $\M_{\DDD}$ could contain self-loops.

The motivation for considering $\M_{\DDD}$ is that it can be seen as an automaton recognizing the correct solutions for the $\LCL$ problem $(2, \Sigma, \DDD, \EE)$ on paths, as each length-$k$  walk $(a_1, b_1) \rightarrow (a_2, b_2) \rightarrow \cdots \rightarrow (a_{k+1}, b_{k+1})$ of $\M_{\DDD}$ corresponds to a correct labeling of a length-$k$ path $(v_1, v_2, \ldots, v_{k+1})$ where the labeling of half-edge $(v_i, \{v_{i-1}, v_{i}\})$ is $a_i$ and the labeling of half-edge $(v_i, \{v_i, v_{i+1}\})$ is $b_i$.

\subparagraph{Path-flexibility}   With respect to the directed graph $\M_{\DDD}$,  we say that $(a,b) \in V(\M_{\DDD})$ is path-flexible if   there exists an integer $K$ such that for each integer $k \geq K$,  there exist length-$k$ walks  $(a,b) \leadsto (a,b)$,  $(a,b) \leadsto (b,a)$,  $(b,a) \leadsto (a,b)$, and  $(b,a) \leadsto (b,a)$ in  $\M_{\DDD}$. Throughout this paper, we write $u \leadsto v$ to denote a walk starting from $u$ and ending at $v$.

It is clear that  $(a,b)$ is path-flexible if and only if  $(b,a)$ is path-flexible. Hence we may extend the notion of path-flexibility from $V(\M_{\DDD})$ to $\DDD$.   That is, we say that a size-$2$ multiset $\{a,b\} \in \DDD$ is path-flexible if $(a,b)$ is path-flexible.

The following lemma is useful in lower bound proofs. For any  $\{a,b\} \in \DDD$ that is not path-flexible, the following lemma shows that there are infinitely many path lengths $k$ such that there is no length-$k$ $s \leadsto t$  walk for some $s \in \{(a,b),(b,a)\}$ and $t \in \{(a,b),(b,a)\}$. As we will later see, this inflexibility in the possible path lengths implies lower bounds for distributed algorithms that may use the configuration $\{a,b\}$.

\begin{lemma}[Property of path-inflexibility]\label{lem-inflex}
Suppose that the size-2 multiset $\{a,b\} \in \DDD$ is not path-flexible. Then one of the following holds.
\begin{itemize}
    \item There is no $s \leadsto t$  walk for at least one choice of $s \in \{(a,b),(b,a)\}$ and $t \in \{(a,b),(b,a)\}$.
    \item There is an integer $2 \leq x \leq |\Sigma|^2$ such that for any positive integer $k$ that is not an integer multiple of $x$, there are no length-$k$ walks $(a,b) \leadsto (a,b)$  and $(b,a) \leadsto (b,a)$ in $\M_{\DDD}$.
\end{itemize}
\end{lemma}
\begin{proof}
Suppose that $\{a,b\} \in \DDD$ is not path-flexible. 
We assume that there are $s \leadsto t$  walks for all choices of   $s \in \{(a,b),(b,a)\}$ and  $t \in \{(a,b),(b,a)\}$. To prove this lemma, it suffices to show that there is an integer $2 \leq x \leq |\Sigma|^2$ such that for any positive integer $k$ that is not an integer multiple of $x$, there are no length-$k$ walks $(a,b) \leadsto (a,b)$ and $(b,a) \leadsto (b,a)$.

First of all, we claim that for any integer $K$ there is an integer $k \geq K$ such that there is no length-$k$ walk $(a,b) \leadsto (a,b)$. If this claim does not hold, then there is an integer $K$ such that there is a length-$k$ walk $(a,b) \leadsto (a,b)$ for each $k \geq K$.  Combining these walks with existing walks  $(a,b) \leadsto (b,a)$ and $(b,a) \leadsto (a,b)$, we infer that there exists an integer $K'$ such that for each integer $k \geq K'$, there exist length-$k$ walks  $(a,b) \leadsto (a,b)$,  $(a,b) \leadsto (b,a)$,  $(b,a) \leadsto (a,b)$, and  $(b,a) \leadsto (b,a)$ in  $\M_{\DDD}$, contradicting the assumption that $\{a,b\} \in \DDD$ is not path-flexible.

Let $U$ be the set of integers $k$ such that there is a length-$k$ walk $(a,b) \leadsto (a,b)$. Note that by taking reversal, the existence of a length-$k$ walk $(a,b) \leadsto (a,b)$ implies the existence of a length-$k$ walk $(b,a) \leadsto (b,a)$, and vice versa. 
 Our assumption on the existence of a walk $(a,b) \leadsto (a,b)$ implies  $U \neq \emptyset$.
We choose $x=\gcd(U)$ to be the greatest common divisor of $U$, so that  for any integer $k$ that is not an integer multiple of $x$,  there are no length-$k$ walks $(a,b) \leadsto (a,b)$  and $(b,a) \leadsto (b,a)$ in $\M_{\DDD}$.
We must have $x \geq 2$ because there cannot be two co-prime numbers in $U$, since otherwise there exists an integer $K$ such that $U$ includes all integers that are at least $K$, contradicting the claim proved above. Specifically, if the two co-prime numbers are $k_1$ and $k_2$, then we may set $K = g(k_1, k_2) + 1 = k_1 k_2 - k_1 - k_2 + 1$, where $g(k_1, k_2)$ is the Frobenius number of the set $\{k_1, k_2\}$~\cite{10.2307/2369536}.
We also have $x \leq  |\Sigma|^2$, since the smallest number in $U$ is at most the number of nodes in $\M_\DDD$, which is  upper bounded by $|\Sigma|^2$.
\end{proof}

For the special case of $|\Sigma| = 1$ and $\DDD \neq \emptyset$, we must have $a = b$ in \cref{lem-inflex}. Since there is no integer $x$ satisfying $2 \leq x \leq |\Sigma|^2$ when $|\Sigma| = 1$, 
\cref{lem-inflex} implies that if $\{a,a\}$ is not path-flexible, then  there is no walk $(a,a) \leadsto (a,a)$, where $\{a,a\}$ is the unique element in $\DDD$.

\subparagraph{Path-flexible  strongly connected components} Since each  $\{a,b\} \in \DDD$ corresponds to two nodes $(a,b)$ and $(b,a)$ in  $\M_{\DDD}$, we will consider a different notion of a strongly connected component. In \cref{def-scc}, we do not require the elements $a$, $b$, $c$, and $d$ to be distinct. For example, we may have $\{a,b\} = \{c,d\}$ or $a = b$.

\begin{definition}[Strongly connected components]\label{def-scc}
Let $\DDD$ be a set of size-$2$ multisets of elements in $\Sigma$. For each $\{a,b\} \in \DDD$ and  $\{c,d\} \in \DDD$, we write $\{a,b\} \sim \{c,d\}$ if  there is a walk  $s \leadsto t$ in $\M_{\DDD}$ for each choice of $s \in \{(a,b), (b,a)\}$ and  $t \in \{(c,d), (d,c)\}$.

Let $\DDD^\sim$ be the set of all $\{a,b\} \in \DDD$ such that $\{a,b\} \sim \{a,b\}$. Then we define the strongly connected components of $\DDD$ as the equivalence classes of $\sim$ over $\DDD^\sim$.
\end{definition}

By taking reversal, the existence of an  $(a,b) \leadsto (c,d)$ walk implies the existence of a $(d,c) \leadsto (b,a)$ walk. Therefore, if  there is a walk  $s \leadsto t$ in $\M_{\DDD}$ for each choice of $s \in \{(a,b), (b,a)\}$ and  $t \in \{(c,d), (d,c)\}$, then there is also a walk  $t \leadsto s$ in $\M_{\DDD}$ for each choice of $s \in \{(a,b), (b,a)\}$ and  $t \in \{(c,d), (d,c)\}$. Hence the relation  $\sim$ in \cref{def-scc} is symmetric over $\DDD$.
It is clear from the definition of $\sim$ in \cref{def-scc} that it is transitive over $\DDD$ and it is reflexive over $\DDD^\sim$, so $\sim$ is indeed an equivalence relation over $\DDD^\sim$.

For any strongly connected component $\DDD'$ of $\DDD$, it is clear that either  all $\{a,b\} \in \DDD'$ are path-flexible or  all $\{a,b\} \in \DDD'$ are not path-flexible. We say that a strongly connected component $\DDD'$ is path-flexible  if  all $\{a,b\} \in \DDD'$ are path-flexible. We define
$\flexibility(\DDD')$ as the minimum number $K$ such that  for each integer $k \geq K$ there is an $(a,b) \leadsto (c,d)$ walk of length $k$ for all choices of $a$, $b$, $c$, and $d$ such that $\{a,b\} \in \DDD'$ and  $\{c,d\} \in \DDD'$.
It is clear that such a number $K$ exists given that $\DDD'$ is a path-flexible strongly connected component.
 We define
\[
\flexSCC(\DDD) = \text{\parbox{0.5\textwidth}{the set of all subsets of $\DDD$ that are a path-flexible strongly connected component of $\DDD$.}}
\]

Clearly, elements in $\flexSCC(\DDD)$ are disjoint subsets of $\DDD$. It is possible that $\flexSCC(\DDD)$ is an empty set, and this happens when all nodes in the directed graph $\M_{\DDD}$ are not path-flexible.

\subparagraph{Restriction of a set of node configurations}
Given an $\LCL$ problem $\Pi=(\Delta,\Sigma,\VV,\EE)$, a subset $\SSS \subseteq \VV$   of node configurations, and a set $\DDD$ of size-$2$ multisets whose elements are in $\Sigma$, we define the restriction of $\SSS$ to $\DDD$ as follows.
\[
\SSS\upharpoonright_\DDD = \{ C \in \SSS \ | \ \text{all size-$2$ sub-multisets of $C$ are in $\DDD$}\}.
\]

\cref{lem-flex-scc} shows that if we label the two endpoints of a sufficiently long path using node configurations in $\SSS\upharpoonright_{\DDD^\ast}$, where $\DDD^\ast \in \flexSCC(\DDD_\SSS)$, then it is always possible to complete the labeling of the path using only node configurations in $\SSS$ in such a way that the entire labeling is correct.
Specifically, consider a path $P=(v_1, v_2, \ldots, v_{d+1})$ of length  $d \geq \flexibility({\DDD^\ast})$. Assume that the node configuration of $v_1$ is already fixed to be $C \in \SSS\upharpoonright_{\DDD^\ast}$ where the half-edge $(v_1, \{v_1, v_2\})$ is labeled by $\beta \in C$ and the node configuration of $v_{d+1}$ is already fixed to be $C' \in \SSS\upharpoonright_{\DDD^\ast}$ where the half-edge $(v_{d+1}, \{v_{d}, v_{d+1}\})$ is labeled by $\alpha' \in C'$. \cref{lem-flex-scc} shows that it is possible to complete the labeling of $P$ using only node configurations in $\SSS$, as we may label $v_i$ using the node configuration $C_i$ where the two half-edges $(v_i, \{v_{i-1}, v_i\})$ and $(v_i, \{v_{i}, v_{i+1}\})$ are labeled by $\alpha_i$ and $\beta_i$, for each  $2 \leq i \leq d$.

\begin{lemma}[Property of path-flexible strongly connected components]\label{lem-flex-scc}
Let $\SSS \subseteq \VV$ be a set of node configurations, and let $\DDD^\ast \in \flexSCC(\DDD_\SSS)$.
For any choices of $C \in \SSS\upharpoonright_{\DDD^\ast}$, $C' \in \SSS\upharpoonright_{\DDD^\ast}$,  size-2 sub-multisets $\{\alpha, \beta\} \subseteq C$, $\{\alpha', \beta'\} \subseteq C'$, and a number $d \geq \flexibility({\DDD^\ast})$, there exists a sequence
\[\alpha_1, C_1, \beta_1, \alpha_2, C_2, \beta_2, \ldots, \alpha_{d+1}, C_{d+1}, \beta_{d+1}\]
satisfying the following conditions.
\begin{itemize}
    \item First endpoint: $\alpha_1 = \alpha$, $\beta_1 = \beta$, and $C_1 = C$.
    \item Last endpoint: $\alpha_{d+1} = \alpha'$, $\beta_{d+1} = \beta'$, and $C_{d+1} = C'$.
    \item Node configurations: for $1 \leq i \leq d+1$, $\{\alpha_i, \beta_i\}$ is a size-2 sub-multiset of $C_i$, and $C_i \in \SSS$.
    \item Edge configurations: for $1 \leq i \leq d$, $\{\beta_{i}, \alpha_{i+1}\} \in \EE$.
\end{itemize}
\end{lemma}
\begin{proof}
By the path-flexibility of $\DDD^\ast$, there exists a length-$d$ walk $(\alpha,\beta) \leadsto (\alpha',\beta')$ in $\M_{\DDD_\SSS}$. We fix \[(\alpha_1, \beta_1) \rightarrow (\alpha_2, \beta_2) \rightarrow \cdots \rightarrow (\alpha_{d+1}, \beta_{d+1})\] to be any such walk. This implies that $\{\beta_{i}, \alpha_{i+1}\} \in \EE$ for each $1 \leq i \leq d$. Since $\{\alpha_i, \beta_i\}$ is a size-2 multiset of $\DDD_\SSS$, there exists a choice of $C_i \in \SSS$ for each $2 \leq i \leq d$ such that  $\{\alpha_i, \beta_i\}$ is a sub-multiset of $C_i$. 
\end{proof}

\subparagraph{Good sequences} Given an $\LCL$ problem $\Pi=(\Delta, \Sigma,\VV,\EE)$ on $\Delta$-regular trees, we say that a sequence
\[(\VV_1, \DDD_1, \VV_2, \DDD_2, \ldots, \VV_k)\]
is \emph{good} if it satisfies the following requirements.
\begin{itemize}
    \item $\VV_1 = \trim(\VV)$. That is, we start the sequence from the result of trimming the set $\VV$ of all node configurations in the given $\LCL$ problem $\Pi=(\Delta, \Sigma,\VV,\EE)$.
    \item  For each $1 \leq i \leq k-1$, $\DDD_i \in \flexSCC(\DDD_{\VV_{i}})$.  That is, $\DDD_i$ is a path-flexible strongly connected component of the automaton associated with the path-form of the $\LCL$ problem $(\Delta, \Sigma,\VV_i,\EE)$, which is $\Pi$ restricted to the set of node configurations $\VV_i$.
    \item  For each $2 \leq i \leq k$, $\VV_i = \trim(\VV_{i-1}\upharpoonright_{\DDD_{i-1}})$. That is, $\VV_i$ is the result of taking the restriction of the set of node configurations $\VV_{i-1}$ to $\DDD_{i-1}$ and then performing a trimming.
    \item $\VV_k \neq \emptyset$. That is, we require that the last set of node configurations is non-empty.
\end{itemize}

It is straightforward to see that $\VV_1 \supseteq \VV_2 \supseteq \cdots \supseteq \VV_k$ since $\VV_i = \trim(\VV_{i-1}\upharpoonright_{\DDD_{i-1}})$ is always a subset of $\VV_{i-1}$. Similarly, we also have $\DDD_1 \supseteq \DDD_2 \supseteq \cdots \supseteq \DDD_{k-1}$, as $\DDD_i \in \flexSCC(\DDD_{\VV_{i}})$ is a subset of $\DDD_{\VV_{i}}$ and $\DDD_{\VV_{i}}$ is a subset of $\DDD_{i-1}$ due to the definition $\VV_i = \trim(\VV_{i-1}\upharpoonright_{\DDD_{i-1}})$.

\subparagraph{Depth of an LCL problem}
We define the depth $d_\Pi$ of an $\LCL$ problem $\Pi=(\Delta, \Sigma,\VV,\EE)$ on $\Delta$-regular trees as follows.
If there is no good sequence, then we set $d_\Pi = 0$.
If there is a good sequence $(\VV_1, \DDD_1, \VV_2, \DDD_2, \ldots, \VV_k)$ for each positive integer $k$, then we set $d_\Pi = \infty$.
Otherwise, we set  $d_\Pi$ as  the largest integer $k$ such that there is a good sequence $(\VV_1, \DDD_1, \VV_2, \DDD_2, \ldots, \VV_k)$.
We  prove the following results.

\begin{theorem}[Characterization of complexity classes]\label{thm-unrooted-characterization}
Let $\Pi=(\Delta, \Sigma,\VV,\EE)$ be an $\LCL$ problem on $\Delta$-regular trees. We have the following.
\begin{itemize}
    \item If $d_\Pi = 0$, then $\Pi$ is unsolvable in the sense that there exists a tree of maximum degree $\Delta$ such that there is no correct solution of $\Pi$ on this tree.
    \item If $d_\Pi = k$ is a positive integer, then the optimal round complexity of $\Pi$ is $\Theta(n^{1/k})$.
    \item If $d_\Pi = \infty$, then $\Pi$ can be solved in $O(\log n)$ rounds.
\end{itemize}
\end{theorem}

In \cref{thm-unrooted-characterization}, all the upper bounds hold in the $\CONGEST$ model, and all the lower bounds hold in the $\LOCAL$ model. For example, if $d_\Pi = 5$, then $\Pi$ can be solved in $O(n^{1/5})$ rounds in the $\CONGEST$ model, and there is a matching lower bound $\Omega(n^{1/5})$ in the $\LOCAL$ model.

We note that there are several natural definitions of unsolvability of an $\LCL$ w.r.t.~a given graph class~\cite{lcls_on_paths_and_cycles} that are different from the one in \cref{thm-unrooted-characterization}. 

\begin{theorem}[Complexity of the characterization]\label{thm-unrooted-poly-time}
There is a polynomial-time algorithm $\A$ that computes $d_\Pi$  for any given  $\LCL$ problem $\Pi=(\Delta, \Sigma,\VV,\EE)$ on $\Delta$-regular trees. If $d_\Pi = k$ is a positive integer, then  $\A$ also outputs a description of an  $O(n^{1/k})$-round algorithm for $\Pi$. If $d_\Pi = \infty$, then $\A$ also outputs a description of an  $O(\log n)$-round algorithm for $\Pi$.
\end{theorem}

The distributed algorithms returned by the polynomial-time algorithm $\A$ in \cref{thm-unrooted-poly-time} are also in the $\CONGEST$ model.

\subsection{Upper bounds}\label{sect-upper}

In this section, we prove the upper bound part of \cref{thm-unrooted-characterization}. If a good sequence $(\VV_1$, $\DDD_1$, $\VV_2$, $\DDD_2$, $\ldots$, $\VV_k)$ exists for some positive integer $k$, we show that the $\LCL$ problem $\Pi=(\Delta, \Sigma,\VV,\EE)$ can be solved in $O(n^{1/k})$ rounds. If a good sequence $(\VV_1$, $\DDD_1$, $\VV_2$, $\DDD_2$, $\ldots$, $\VV_k)$ exists for all positive integers $k$, then we show that $\Pi=(\Delta, \Sigma,\VV,\EE)$ can be solved in $O(\log n)$ rounds. All these algorithms do not require sending large messages and can be implemented in the $\CONGEST$ model.

\subparagraph{Rake-and-compress decompositions} 
Roughly speaking, a rake-and-compress process is a procedure that decomposes a tree by iteratively removing degree-$1$ nodes (rake) and removing degree-$2$ nodes (compress). There are several tree decompositions resulting from variants of a rake-and-compress process. Here we  use a variant of decomposition considered in~\cite{Chang20} that is parameterized by three positive integers $\gamma$, $\ell$, and $L$.
A $(\gamma, \ell, L)$ decomposition of a tree $G=(V,E)$ is a partition of the node set 
\[V = \VR{1} \cup \VC{1} \cup \VR{2} \cup \VC{2}   \cup \cdots \cup \VR{L}\] 
satisfying the following requirements.

\subparagraph{Requirements for $\VR{i}$}
For each connected component $S$  of the subgraph  of $G$ induced by $\VR{i}$, it is required that there is a root $z \in S$  meeting the following conditions.
\begin{itemize}
\item $z$ has at most one neighbor in $\VC{i} \cup \VR{i+1} \cup \cdots \cup  \VR{L}$.
\item All nodes in $S \setminus \{z\}$ have no neighbor in $\VC{i} \cup \VR{i+1} \cup \cdots \cup  \VR{L}$.
\item All nodes in $S \setminus \{z\}$ are within distance $\gamma-1$ to $z$.
\end{itemize}  

Intuitively, each connected component $S$ of the subgraph  of $G$ induced by $\VR{i}$ is a rooted tree of height at most $\gamma-1$. 
The root can have at most one neighbor residing in the higher layers of the decomposition. The remaining nodes in $S$ cannot have neighbors in the higher layers of the decomposition.
For the special case of $\gamma = 1$, the set $\VR{i}$ is an independent set.

\subparagraph{Requirements for $\VC{i}$}
For each connected component $S$  of the subgraph  of $G$ induced by $\VC{i}$,  $S$ is a path $(v_1, v_2, \ldots, v_s)$ of $s \in [\ell, 2\ell]$ nodes  meeting the following conditions.
\begin{itemize}
\item There exist two nodes $u$ and $w$ in $\VR{i+1} \cup \VC{i+1} \cup \cdots \cup  \VR{L}$  such that $u$ is adjacent to $v_1$ and $w$ is adjacent to $v_s$.
\item For each $1 \leq j  \leq s$, $v_j$ has no neighbor in $\bigl(\VR{i+1} \cup \VC{i+1} \cup \cdots \cup  \VR{L}\bigr) \setminus \{u,w\}$.
\end{itemize} 

Intuitively, each connected component $S$  of the subgraph  of $G$ induced by $\VC{i}$ is a path. Only the endpoints of the paths can have neighbors residing in the higher layers of the decomposition. The remaining nodes in $S$ cannot have neighbors in the higher layers of the decomposition.

\subparagraph{Existing algorithms for rake-and-compress decompositions}
For any $k = O(1)$ and $\ell = O(1)$, there is an $O(n^{1/k})$-round algorithm computing a $(\gamma, \ell, L)$ decomposition of a tree $G=(V,E)$ with $\gamma =O(n^{1/k})$ and $L = k$~\cite{Chang20}. For any  $\ell = O(1)$, there is an $O(\log n)$-round algorithm computing a $(\gamma, \ell, L)$ decomposition of a tree $G=(V,E)$ with $\gamma = 1$ and $L = O(\log n)$~\cite{CP19timeHierarchy}. These algorithms are deterministic  and can be implemented in the $\CONGEST$ model. We will employ these algorithms as subroutines to prove the upper bound part of  \cref{thm-unrooted-characterization}. 

\begin{lemma}[Solving $\Pi$ using rake-and-compress decompositions]\label{lem-rake-and-compress-use}
 Suppose we are given an $\LCL$ problem $\Pi=(\Delta, \Sigma,\VV,\EE)$ on $\Delta$-regular trees that admits a good sequence
\[(\VV_1, \DDD_1, \VV_2, \DDD_2, \ldots, \VV_k).\]
Suppose we are given a $(\gamma, \ell, L)$ decomposition of an $n$-node tree $G=(V,E)$ of maximum degree at most $\Delta$
\[V = \VR{1} \cup \VC{1} \cup \VR{2} \cup \VC{2}   \cup \cdots \cup  \VR{L}\] 
with $L = k$ and $\ell = \max\{1, \flexibility(\DDD_1)-1, \ldots, \flexibility(\DDD_{k-1})-1\}$. 
Then a correct solution of  $\Pi$ on $G$ can be computed in $O((\gamma+\ell)L)$ rounds in the $\CONGEST$ model.
\end{lemma} 
\begin{proof}
We present an $O((\gamma+\ell)L)$-round $\CONGEST$ algorithm finding a correct solution of  $\Pi$ on $G$.  
The algorithm labels the half-edges surrounding the nodes of the graph in the order $\VR{L}, \VC{L-1}, \ldots, \VR{1}$. Our algorithm has the property that it only uses the node configurations in $\VV_i$ to label the half-edges surrounding the nodes in $\VR{i}$ and $\VC{i}$. 

Not all nodes $v$ in $G$ have $\deg(v) = \Delta$. 
In general, for each node $v$ in $G$, we say that the node configuration of $v$ is in $\VV_i$ if the  multiset of the $\deg(v)$ half-edge labels surrounding $v$ is a sub-multiset of some $C \in \VV_i$. 

\subparagraph{Labeling $\VR{i}$} By induction hypothesis, assume the algorithm has finished labeling the half-edges surrounding the nodes in $\VC{i} \cup \VR{i+1} \cup \cdots \cup  \VR{L}$ in such a way that their node configurations are in $\VV_i$, as we recall that $\VV_i \supseteq \VV_{i+1} \supseteq \cdots \supseteq \VV_k$.
The algorithm then labels each connected component $S$  of the subgraph  of $G$ induced by $\VR{i}$, in parallel and using $O(\gamma)$ rounds in the $\CONGEST$ model, as follows.

The set $S$ has the property that there is at most one node $z \in S$ that may have a neighbor in  $\VC{i} \cup \VR{i+1} \cup \cdots \cup  \VR{L}$, and the number of neighbors of  $z$ in $\VC{i} \cup \VR{i+1} \cup \cdots \cup  \VR{L}$ is at most one. We claim that it is always possible to complete the labeling of half-edges surrounding the nodes in $S$ using only node configurations in $\VV_i$.
To see that this is possible, it suffices to consider the following situation. Given that the existing half-edge labels surrounding a node $v$ form a node configuration in $\VV_i$, consider a neighbor $u$ of $v$, and we want to label the half-edges surrounding $u$ in such a way that the edge configuration of $e=\{u,v\}$ is in $\EE$ and the node configuration of $u$ is in  $\VV_i$. This is always doable due to \cref{lem-trimming}, as we recall   from the definition of $\VV_i$ that $\VV_i=\trim(\SSS)$ for some set $\SSS \subseteq \VV$. 
The round complexity of labeling $S$ is $O(\gamma)$ because $S$ is a tree rooted at $z$ of depth at most $\gamma -1$.

\subparagraph{Labeling $\VC{i}$} Similarly, by induction hypothesis, assume the algorithm has already finished labeling  the half-edges surrounding  the nodes in $\VR{i+1} \cup \VC{i+1} \cup \cdots \cup  \VR{L}$ in such a way that their node configurations are in $\VV_{i+1}$, as we recall that $\VV_{i+1} \supseteq \VV_{i+2} \supseteq \cdots \supseteq \VV_k$.
The algorithm then labels each connected component $S$  of the subgraph  of $G$ induced by $\VC{i}$, in parallel and using $O(\ell)$ rounds in the $\CONGEST$ model, as follows.

The set $S$ has the property that there are exactly two nodes $u$ and $w$ in $\VR{i+1} \cup \VC{i+1} \cup \cdots \cup  \VR{L}$ adjacent to $S$, and the subgraph induced by  $S \cup \{u,w\}$ is a path $(u, v_1, v_2, \ldots, v_s, w)$, with $s \in [\ell, 2\ell]$. Hence the length of this path is $s+1 \geq \ell+1 \geq \flexibility(\DDD_i)$ by our choice of $\ell$.  

Recall that $\DDD_i \in \flexSCC(\DDD_{\VV_i})$ and the node configurations of $u$ and $w$ are in $\VV_{i+1}  = \trim(\VV_{i}\upharpoonright_{\DDD_{i}}) \subseteq \VV_{i}\upharpoonright_{\DDD_{i}}$, so \cref{lem-flex-scc}  ensures that we can label the half-edges surrounding the nodes $v_1, v_2, \ldots, v_s$ using only node configurations in $\VV_i$ in such a way that the edge configurations of all edges in the path  $(u, v_1, v_2, \ldots, v_s, w)$ are in $\EE$. 
The round complexity of labeling $S$ is $O(\ell)$ because $S$ is a path of at most $2 \ell$ nodes.

\subparagraph{Summary} The number of rounds spent on labeling each part $\VR{i}$ is $O(\gamma)$, and the number of rounds spent on labeling each part $\VC{i}$ is $O(\ell)$, so the overall round complexity for solving $\Pi$ given a  $(\gamma, \ell, L)$ decomposition is $O((\gamma+\ell)L)$ rounds in the $\CONGEST$ model.  
\end{proof}

Combining \cref{lem-rake-and-compress-use} with existing algorithms for computing $(\gamma, \ell, L)$ decompositions, we obtain the following results.

\begin{lemma}[Upper bound for the case $d_\Pi = k$]\label{lem-upper-finite}
If $d_\Pi = k$ for some positive integer $k$, then $\Pi$ can be solved in $O(n^{1/k})$ rounds in the $\CONGEST$ model.
\end{lemma}
\begin{proof}
In this case, a good sequence $(\VV_1$, $\DDD_1$, $\VV_2$, $\DDD_2$, $\ldots$, $\VV_k)$ exists. 
As a  $(\gamma, \ell, L)$ decomposition with $\gamma = O(n^{1/k})$, $\ell= O(1)$, and $L = k$  can be computed in $O(n^{1/k})$ rounds~\cite{Chang20}, $\Pi$ can be solved in $O(n^{1/k}) + O((\gamma+\ell)L) = O(n^{1/k})$ rounds  using the algorithm of \cref{lem-rake-and-compress-use}. Here both $k$ and $\ell$ are $O(1)$, as they are independent of the number of nodes~$n$.
\end{proof}

\begin{lemma}[Upper bound for the case $d_\Pi = \infty$]\label{lem-upper-infinite}
If $d_\Pi = \infty$, then $\Pi$ can be solved in $O(\log n)$ rounds in the $\CONGEST$ model.
\end{lemma}
\begin{proof}
In this case, a good sequence $(\VV_1$, $\DDD_1$, $\VV_2$, $\DDD_2$, $\ldots$, $\VV_k)$ exists for all positive integers $k$. As a  $(\gamma, \ell, L)$ decomposition with $\gamma = 1$, $\ell= O(1)$, and $L = O(\log n)$ can be computed in $O(\log n)$ rounds~\cite{CP19timeHierarchy}, by choosing a good sequence $(\VV_1$, $\DDD_1$, $\VV_2$, $\DDD_2$, $\ldots$, $\VV_k)$ with $k = L$, $\Pi$ can be solved in  $O(\log n)+ O((\gamma+\ell)L) = O(\log n)$ rounds using the algorithm of \cref{lem-rake-and-compress-use}. Similarly, here $\ell = O(1)$, as it is independent of the number of nodes $n$.
\end{proof}

\subsection{Lower bounds}\label{sect-lower}

In this section, we prove the lower bound part of \cref{thm-unrooted-characterization}. In our lower bound proofs, we pick $\gamma$ to be the smallest integer satisfying the following requirements. For each subset $\SSS \subseteq \VV$ and each $C \in \SSS \setminus \trim(\SSS)$, there exists no correct labeling of $T_\gamma^\ast$ where the node configuration of the root $r$ is $C$ and the node configurations of the remaining degree-$\Delta$ nodes are in $\SSS$. Such a number $\gamma$ exists due to the definition of $\trim$.

\begin{lemma}[Unsolvability for the case $d_\Pi = 0$]\label{lem-lower-0}
If $d_\Pi = 0$, then $\Pi$ is unsolvable in the sense that there exists a tree $G$ of maximum degree $\Delta$ such that there is no correct solution of $\Pi$ on $G$.
\end{lemma}
\begin{proof}
We take $G = T_\gamma^\ast$. Since $d_\Pi = 0$, we have $\trim(\VV) = \emptyset$. Our choice of $\gamma$ implies that there is no correct solution of $\Pi$ on $G = T_\gamma^\ast$.
\end{proof}

For the rest of this section, we focus on the case that $d_\Pi = k$ is a positive integer. We will prove that solving $\Pi$ requires $\Omega(n^{1/k})$ rounds in the $\LOCAL$ model. The exact choice of   $s = \Theta(t)$ in  \cref{def-lb-graphs}  is to be determined later. In our lower bound proof, we will assume that there exists an algorithm $\mathcal{A}$ violating the lower bound the parameter, and then we will derive a contradiction. As we will later see, the parameter $t$  in \cref{def-lb-graphs}  will corresponds to the time complexity of $\mathcal{A}$.

\begin{definition}[Lower bound graphs]\label{def-lb-graphs}  Let $t$ be any positive integer and choose $s = \Theta(t)$  to be a sufficiently large integer. 

\begin{itemize}
    \item  $G_{\sR,1}$ is the rooted tree $T_\gamma$, and $G_{\sR,1}^\ast$ is the rooted tree $T_\gamma^\ast$. All nodes in $G_{\sR,1}$ and $G_{\sR,1}^\ast$ are said to be in layer $(\sR,1)$.
    \item  For each integer $i \geq 1$,  $G_{\sC,i}$ is the result of the following construction. Start with an $s$-node path $(v_1, v_2, \ldots, v_s)$ and let $v_1$ be the root. For each $1 \leq i < s$, append $\Delta-2$ copies of $G_{\sR,i}$  to $v_i$. For $i = s$, append $\Delta-1$ copies of  $G_{\sR,i}$ to $v_s$. The nodes $v_1, v_2, \ldots, v_s$ are said to be in layer $(\sC,i)$.
    \item For each integer $i \geq 2$,  $G_{\sR,i}$  is the result of the following construction. Start with a rooted tree $T_\gamma$.  Append $\Delta-1$ copies of $G_{\sC,i-1}$ to each leaf in $T_\gamma$.
    All nodes in  $T_\gamma$ are said to be in layer $(\sR,i)$.    The rooted tree  $G_{\sR,i}^\ast$ is defined analogously by replacing $T_\gamma$ with $T_\gamma^\ast$ in the construction.
\end{itemize}
\end{definition}

Although the trees $G_{\sR,i}$, $G_{\sR,i}^\ast$, and $G_{\sC,i}$  are rooted in their definitions, we may also treat them as unrooted trees. Throughout our lower bound proof in \cref{sect-lower}, we fix our main lower bound graph $G=(V,E)$ to be the tree $G_{\sR,k+1}^\ast$, see \cref{fig:lb-graph} for an example.

\begin{figure}
	\centering
	\includegraphics[width=\textwidth]{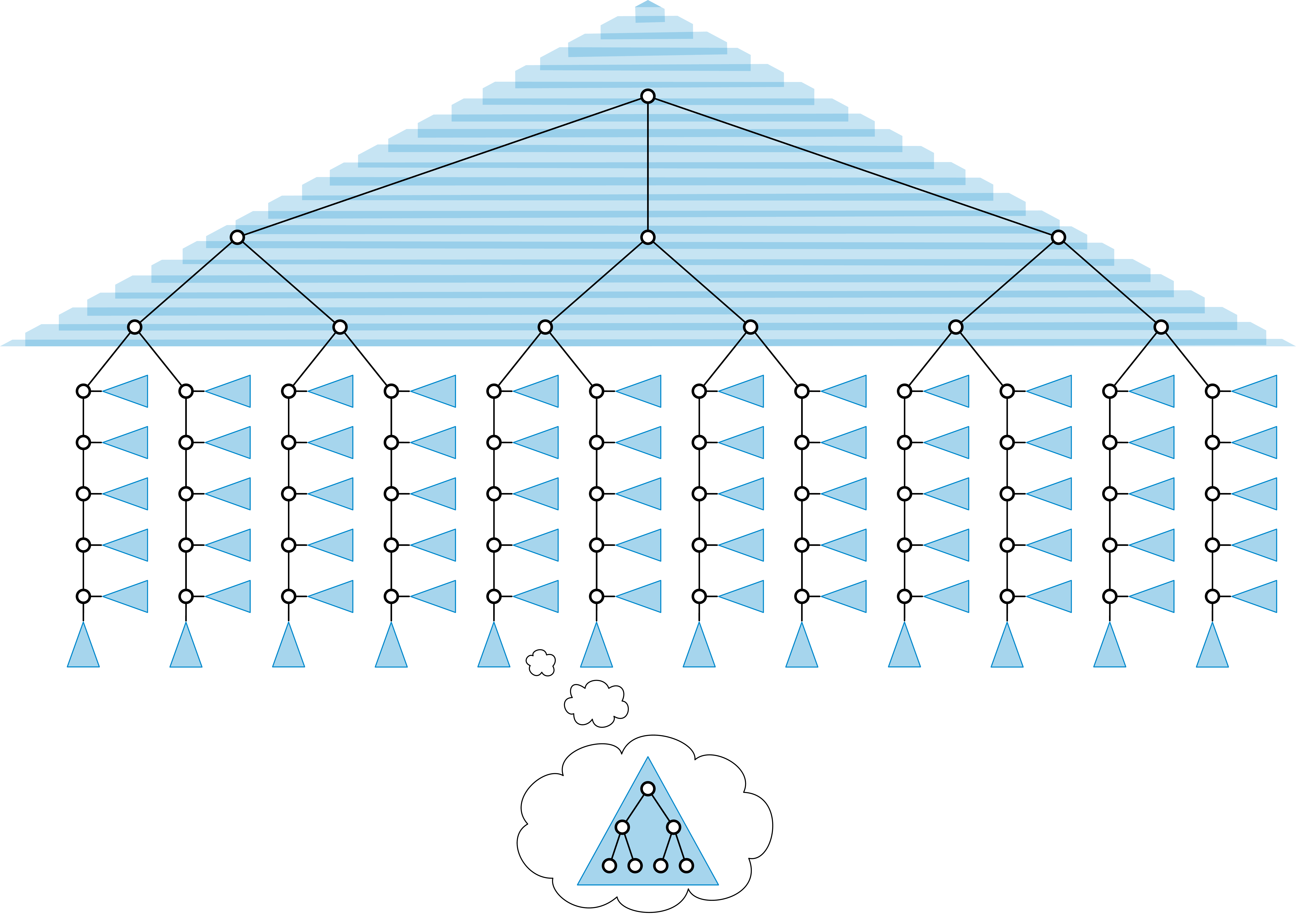}
	\caption{The graph $G_{\sR,2}^\ast$ with $\Delta = 3$, $s = 5$, and $\gamma = 2$.} 
	\label{fig:lb-graph}
\end{figure}

\begin{definition}[Main lower bound graph]\label{def-lb-graph-main}
Define $G=(V,E)$ as the tree $G_{\sR,k+1}^\ast$.
\end{definition}  

In the subsequent discussion, we focus on the tree $G = G_{\sR,k+1}^\ast$. It is clear from its construction that $G$ has $n = O(t^k)$ nodes, if we treat $\gamma$ as a constant independent of $t$. Indeed, $\gamma$ only depends on the underlying $\LCL$ problem $\Pi=(\Delta, \Sigma, \VV, \EE)$. More generally, the number of nodes in $G_{\sR,i}$ or $G_{\sR,i}^\ast$ is $O(t^{i-1})$, and the number of nodes in $G_{\sC,i}$  is $O(t^{i})$.
To show that $\Pi$ requires $\Omega(n^{1/k})$ rounds to solve, it suffices to show that $\Pi$ requires $t$ rounds to solve on $G$, for any positive integer $t$.

The nodes in $G = G_{\sR,k+1}^\ast$ are partitioned into layers $(\sR,1), (\sC,1), (\sR,2), (\sC,2), \ldots, (\sR,k+1)$ according to the rules in the above recursive construction. In the subsequent discussion, we order the layers by $(\sR,1) \prec (\sC,1)  \prec (\sR,2)  \prec (\sC,2)  \prec \cdots  \prec (\sR,k+1)$. For example, when we say layer $(\sR,i)$ or higher, we mean the set of all layers $(\sR,i), (\sC,i), \ldots, (\sR,k+1)$.

Intuitively, layers $(\sR,i)$ and $(\sC,i)$  resemble the parts $\VR{i}$ and $\VC{i}$ in a rake-and-compress decomposition. Except for some leaf nodes in layer $(\sR,1)$, all nodes in the graph have degree $\Delta$. Each connected component of the subgraph of $G$ induced by layer $(\sC,i)$ nodes is a path of $s$ nodes. 
Each connected component of the subgraph of $G$ induced by layer $(\sR,i)$ nodes is a  rooted tree $T_\gamma$ (if $1 \leq i \leq k$) or a rooted tree $T_\gamma^\ast$ (if $i = k+1$). We further classify the nodes in layer $(\sC,i)$ as follows. 

\begin{definition}[Classification of nodes in layer $(\sC,i)$]\label{def-layer-C}
The nodes in the $s$-node path  $(v_1, v_2, \ldots, v_s)$ in the construction of $G_{\sC,i}$ are classified as follows.
\begin{itemize}
    \item We say that $v_j$ is a front node if $1 \leq j \leq t$.
    \item We say that $v_j$ is a central node if $t+1 \leq j \leq s-t$.
    \item We say that $v_j$ is a rear node if $s-t+1 \leq j \leq s$.
\end{itemize}
\end{definition}

Based on \cref{def-layer-C}, we define the following subsets of nodes in $G = G_{\sR,k+1}^\ast$. We assume that $s = \Theta(t)$ is chosen to be sufficiently large so that central nodes exist.

\begin{definition}[Subsets of nodes in $G$]\label{def-subsets}
We define the following subsets of nodes in $G = G_{\sR,k+1}^\ast$.
\begin{itemize}
   \item Define $S_{\sR,1}$ as the set of nodes $v$ in $G$ such that the radius-$\gamma$ neighborhood of $v$ is isomorphic to $T_\gamma^\ast$.
    \item For $2 \leq i \leq k+1$, define $S_{\sR,i}$ as the set of nodes $v$ in $G$ such that the radius-$\gamma$ neighborhood of $v$ is isomorphic to $T_\gamma^\ast$
    and contains only nodes in $S_{\sC,i-1}$.
    \item For $1 \leq i \leq k$, define $S_{\sC,i}$ as the set of nodes $v$ in $G$  that are in layer $(\sR,i+1)$ or above or are central or front nodes in layer $(\sC,i)$.
\end{itemize}
\end{definition}

We prove some basic properties of the sets in \cref{def-subsets}.

\begin{lemma}[Subset containment]\label{lem-containment}
We have $S_{\sR,1} \supseteq S_{\sC,1} \supseteq  \cdots \supseteq S_{\sR,k+1} \neq \emptyset$.
\end{lemma}
\begin{proof}
We have $S_{\sC,i} \supseteq S_{\sR,i+1}$ since it follows from the definition of $S_{\sR,i+1}$ that $v \in S_{\sC,i}$ is a necessary condition for $v \in S_{\sR,i+1}$.
The claim that $S_{\sR,i} \supseteq S_{\sC,i}$ follows from the fact that each $v \in S_{\sC,i}$ is in layer $(\sC,i)$ or above: The radius-$\gamma$ neighborhood of any such node $v$ is  isomorphic to $T_\gamma^\ast$ and contains only nodes in layer $(\sR,i)$ or above, and we know that all nodes in layer $(\sR,i)$ or above are in $S_{\sC,i-1}$. Hence $v \in S_{\sC,i}$ implies $v \in S_{\sR,i}$. 

To see that $S_{\sR,k+1} \neq \emptyset$, consider the root $r$ of  $G = G_{\sR,k+1}^\ast$. 
The radius-$\gamma$ neighborhood of $r$ is  isomorphic to $T_\gamma^\ast$ and contains only nodes in layer $(\sR,k+1)$. We know that all nodes in layer $(\sR,k+1)$ are in $S_{\sC,k}$, so $r \in S_{\sR,k+1}$.
\end{proof}

\begin{lemma}[Property of $S_{\sC,i}$]\label{lem-layer-C-property}
For each node $v \in S_{\sC,i}$, one of the following holds.
\begin{itemize}
    \item $v$ is a central node in layer $(\sC,i)$.
    \item For each neighbor $u$ of $v$ such that $u \in S_{\sC,i}$, there exists a path $P = (v, u, \ldots, w)$ such that $w$ is a central node in layer $(\sC,i)$ and all nodes in $P$ are in $S_{\sC,i}$.
\end{itemize}
\end{lemma}
\begin{proof}
We assume that  $v \in S_{\sC,i}$ is not a central node in layer $(\sC,i)$. Consider any neighbor $u$ of $v$ such that $u \in S_{\sC,i}$. To prove the lemma, its suffices to find a path $P = (v, u, \ldots, w)$ such that $w$ is a central node in layer $(\sC,i)$ and all nodes in $P$ are in $S_{\sC,i}$. The existence of such a path $P$ follows from the simple observation that $S_{\sC,i}$ induces a connected subtree where all the leaf nodes are  central nodes in layer $(\sC,i)$.

More specifically, such a path $P$ can be constructed as follows.
If $u$ itself is a central node in layer $(\sC,i)$, then we can simply take $P=(v,u=w)$.

We first consider the case where $v$ is the parent of $u$ in the rooted tree  $G = G_{\sR,k+1}^\ast$. In this case, there must exist a descendant $w$ of $u$  such that $w$ is a central node in layer $(\sC,i)$. This gives us a desired path  $P = (v, u, \ldots, w)$. 

Next, we consider the case where $v$ is a child of $u$  in the rooted tree  $G = G_{\sR,k+1}^\ast$. In this case, it might be possible that all descendants $w$ of $u$  such that $w$ is a central node in layer $(\sC,i)$ are also descendants of $v$.
Hence we will consider a different approach.
Starting from $(v, u)$, we first follow the parent pointers to the root $r$ of  $G = G_{\sR,k+1}^\ast$. There are $\Delta$ children of $r$, and it is clear that each child of $r$ has a descendant $w$ that is a central node in layer $(\sC,i)$. Hence we can extend the current path $(v,u, \ldots, r)$ to a desired $P = (v, u, \ldots, r, \ldots, w)$ such that $w$ that is a central node in layer $(\sC,i)$.
\end{proof}

\subparagraph{Assumptions} We are given an $\LCL$ problem $\Pi=(\Delta, \Sigma, \VV, \EE)$ such that $d_\Pi = k$. Hence there does not exist a good sequence 
\[(\VV_1, \DDD_1, \VV_2, \DDD_2, \ldots, \VV_{k+1}).\]
Recall that the rules for a good sequence are as follows:
\begin{align*}
    \VV_i & =
    \begin{cases}
    \trim(\VV) & \text{if $i=1$},\\
    \trim(\VV_{i-1}\upharpoonright_{\DDD_{i-1}}) &  \text{if $i>1$},\\
    \end{cases}\\
    \DDD_i &\in \flexSCC(\DDD_{\VV_i}).
\end{align*}

The only nondeterminism in the above rules is the choice of $\DDD_i \in \flexSCC(\DDD_{\VV_i})$ for each $i$. The fact that $d_\Pi = k$ implies that for all possible choices of $\DDD_1, \DDD_2, \ldots, \DDD_k$, we always end up with $\VV_{k+1} = \emptyset$.

We assume that there is an algorithm $\A$ that solves $\Pi$ in $t$ rounds on  $G = G_{\sR,k+1}^\ast$. As the number $n$ of nodes in $G$ satisfies $t = \Omega(n^{1/k})$, to prove the desired $\Omega(n^{1/k})$ lower bound, it suffices to derive a contradiction. Specifically, we will prove that the existence of such an algorithm $\A$ forces the existence of a good sequence $(\VV_1, \DDD_1, \VV_2, \DDD_2, \ldots, \VV_{k+1})$, contradicting the fact that  $d_\Pi = k$. 

\subparagraph{Induction hypothesis} Our proof proceeds by an induction on the subsets $S_{\sR,1}$, $S_{\sC,1}$, $S_{\sR,2}$, $S_{\sC,2}$, $\ldots$, $S_{\sR,k+1}$. For each $1 \leq i \leq k$, the choice of $\DDD_i \in \flexSCC(\DDD_{\VV_i})$ is fixed in the induction hypothesis for $S_{\sC,i}$. The choice of $\VV_{i}$ is uniquely determined once $\DDD_1, \DDD_2, \ldots, \DDD_{i-1}$ have been fixed.

Before defining our induction hypothesis, we recall that the output labels of the half-edges surrounding a node $v$ are determined by the subgraph induced by the radius-$t$ neighborhood $U$ of $v$, together with the distinct IDs of the nodes in $U$. For each node $v$ in $G$, we define $\VV_v^\A$ as the set of all possible node configurations of $v$ that can possibly appear when we run $\A$ on $G$.  In other words, $C \in \VV_v^\A$ implies that there exists an assignment of distinct IDs to nodes in the radius-$t$ neighborhood of $v$ such that the output labels of the half-edges surrounding $v$ form the node configuration $C$. 

Similarly, for any two edges $e_1$ and $e_2$ incident to a node $v$, we define $\DDD_{v, e_1, e_2}^\A$ to be the set of all size-$2$ multisets $\{a,b\}$ of labels such that $\{a,b\}$ is a possible outcome of labeling the two half-edges $(v, e_1)$ and $(v, e_2)$ when we run the algorithm $\A$ on $G$.

\begin{definition}[Induction hypothesis for layer $S_{\sR,i}$]\label{def-IH-R}
For each $1 \leq i \leq k+1$, the induction hypothesis for $S_{\sR,i}$ specifies that each $v \in S_{\sR,i}$ satisfies  $\VV_v^\A \subseteq \VV_i$.
\end{definition}

\begin{definition}[Induction hypothesis for layer $S_{\sC,i}$]\label{def-IH-C}
For each $1 \leq i \leq k$, the induction hypothesis for $S_{\sC,i}$ specifies that for each $v \in S_{\sC,i}$ and any two incident edges $e_1 = \{v,u\}$ and $e_2 = \{v,w\}$ such that $u$ and $w$ are in layer  $(\sC,i)$ or higher, we have $\DDD_{v, e_1, e_2}^\A \subseteq \DDD_i$.
\end{definition}

Next, we prove that the induction hypotheses stated in \cref{def-IH-R,def-IH-C} hold.

\begin{lemma}[Base case: $S_{\sR, 1}$]\label{lem-base-case}
The  induction hypothesis for $S_{\sR, 1}$ holds.
\end{lemma}
 \begin{proof}
 Recall that the number $\gamma$ satisfies the following property. For each subset $\SSS \subseteq \VV$ and each $C \in \SSS \setminus \trim(\SSS)$, there exists no correct labeling of $T_\gamma^\ast$ where the node configuration of the root $r$ is $C$ and the  node configuration of remaining degree-$\Delta$ nodes is in $\SSS$.
 
 To prove the induction hypothesis for $S_{\sR, 1}$, consider any node $v \in S_{\sR, 1}$. By the definition of $S_{\sR, 1}$, the radius-$\gamma$ neighborhood of $v$ is isomorphic to $T_\gamma^\ast$ rooted at $v$, and, by setting $\SSS = \VV$, we infer that  there is no correct labeling of $G$ such that the node configuration of $v$ is in $\VV \setminus \trim(\VV) = \VV \setminus \VV_1$, so we must have $\VV_v^\A \subseteq \VV_1$, as $\A$ is correct.
 \end{proof}

\begin{lemma}[Inductive step: $S_{\sR,i}$]\label{lem-inductive-step-CR}
 Let $2 \leq i \leq k$. If the induction hypothesis for  $S_{\sR,i-1}$ and $S_{\sC,i-1}$ holds, then the  induction hypothesis for $S_{\sR,i}$ holds.
\end{lemma}
 \begin{proof}
  To prove the induction hypothesis for layer $S_{\sR,i}$, consider any node $v \in S_{\sR,i}$. By the definition of $S_{\sR,i}$, the radius-$\gamma$ neighborhood of $v$ is isomorphic to $T_\gamma^\ast$ rooted at $v$ and  contains only nodes from $S_{\sC,i-1}$. The goal is to prove that $\VV_v^\A \subseteq \VV_i = \trim(\VV_{i-1}\upharpoonright_{\DDD_{i-1}})$.
  
  The set of degree-$\Delta$ nodes in this $T_\gamma^\ast$  is precisely the set of nodes in the radius-$(\gamma-1)$ neighborhood of $v$ in $G$.  Consider any node $u$ in the radius-$(\gamma-1)$ neighborhood of $v$. 
  As $u$ is in  $S_{\sC,i-1} \subseteq S_{\sR,i-1}$, the induction hypothesis for  $S_{\sR,i-1}$ implies that 
   \[\VV_u^\A \subseteq \VV_{i-1}.\]
   Since all neighbors of $u$ are in  $S_{\sC,i-1}$, from the induction hypothesis for  $S_{\sC,i-1}$, we have \[\DDD_{u, e_1, e_2}^\A \subseteq \DDD_{i-1}\] for any two edges $e_1$ and $e_2$ incident to $u$. 
   Combining these two facts, we infer that 
   \[\VV_u^\A \subseteq \VV_{i-1}\upharpoonright_{\DDD_{i-1}}.\]
   
  Given that all degree-$\Delta$ nodes $u$ in this $T_\gamma^\ast$ satisfy that $\VV_u^\A \subseteq \VV_{i-1}\upharpoonright_{\DDD_{i-1}}$, the same argument as in the proof of \cref{lem-base-case} shows that $\VV_v^\A \subseteq \trim(\VV_{i-1}\upharpoonright_{\DDD_{i-1}}) = \VV_i$, as required. 
 \end{proof}
 
 \begin{lemma}[Inductive step: $S_{\sC,i}$]\label{lem-inductive-step-RC}
 Let $1 \leq i \leq k$. If the induction hypothesis for  $S_{\sR,i}$ holds, then the  induction hypothesis for  $S_{\sC,i}$ holds.
\end{lemma}
 \begin{proof}
To prove the induction hypothesis for  $S_{\sC,i}$, we show that there exists a choice $\DDD_i \in \flexSCC(\DDD_{\VV_i})$ such that the following statement holds. 
For any node $v \in S_{\sC,i}$ and any two incident edges $e_1 = \{v,u\}$ and $e_2 = \{v,w\}$ such that  $u$ and $w$ are in layer  $(\sC,i)$ or higher, we have $\DDD_{v, e_1, e_2}^\A \subseteq \DDD_i$.

We first make an observation about central nodes in layer $(\sC,i)$. Let $v$  be a central node in layer $(\sC,i)$, and let $u$ and $w$ be the two neighbors of $v$ in layer  $(\sC,i)$. Let $e_1 = \{v,u\}$ and $e_2 = \{v,w\}$.
Then it is clear that $\DDD_{v, e_1, e_2}^\A$ is the same for each choice of $v$ that is a central node in layer $(\sC,i)$, as the radius-$t$ neighborhoods of central nodes in the same layer are isomorphic, due to the definition of central nodes and the construction in \cref{def-lb-graphs}.
 For notational convenience, we write $\tilde{\DDD} = \DDD_{v, e_1, e_2}^\A$ to denote this set.

\subparagraph{Plan of the proof} Consider any node $v \in S_{\sC,i}$ and its two incident edges $e_1 = \{v,u\}$ and $e_2 = \{v,w\}$ such that  $u$  and $w$ are in layer  $(\sC,i)$ or higher. Due to the induction hypothesis for $S_{\sR,i}$, we already have $\VV_{v}^\A \subseteq \VV_i$ as $v \in S_{\sC,i} \subseteq S_{\sR,i}$, and so $\DDD_{v, e_1, e_2}^\A \subseteq \DDD_{\VV_i}$.

To prove that  there exists a choice $\DDD_i = \flexSCC(\DDD_{\VV_i})$ such that $\DDD_{v, e_1, e_2}^\A \subseteq \DDD_i$ for all such $v$, $e_1$, and $e_2$, we will first show that $\tilde{\DDD}$ must be a subset of a path-flexible strongly connected component of $\DDD_{\VV_i}$, and then we fix $\DDD_i \in \flexSCC(\DDD_{\VV_i})$ to be this path-flexible strongly connected component.

Next, we will argue that for each $\{a,b\} \in \DDD_{v, e_1, e_2}^\A$, there exist a walk in $\M_{\DDD_{\VV_i}}$ that starts from $(a,b)$ and ends in $\tilde{\DDD}$ and a walk in $\M_{\DDD_{\VV_i}}$ that starts from $\tilde{\DDD}$ and ends in $(a,b)$. This shows that $\{a,b\}$ is in the same strongly connected component as the members in $\tilde{\DDD}$, so we conclude that $\DDD_{v, e_1, e_2}^\A \subseteq \DDD_i$.

\subparagraph{Part 1: $\tilde{\DDD}$ is a subset of some $\DDD_i \in \flexSCC(\DDD_{\VV_i})$}
Consider any path  $(u_1, u_2, \ldots, u_s)$ of $s$ nodes in layer $(\sC,i)$ of  $G = G_{\sR,k+1}^\ast$. Note that any such a path must be the $s$-node path $(v_1, v_2, \ldots, v_s)$ in the construction of some $G_{\sC, i}$ in \cref{def-lb-graphs}.
 We choose $s = \Theta(t)$ to be sufficiently large to ensure that for each integer $0 \leq d \leq |\Sigma|^2$, there exist two nodes $u_j$ and $u_l$ in the path meeting the following conditions.
 \begin{itemize}
     \item $t+1 \leq j < l \leq s-t$, so $u_j$ and $u_l$ are central nodes.
     \item The distance $l-j$ between $u_j$ and $u_l$  equals $4t+3+d$. 
 \end{itemize}
 
 The choice of the number $4t+3$ is to ensure that   the union of the radius-$t$ neighborhoods of the endpoints of any edge in $G$  does not node-intersect the radius-$t$ neighborhood of both $u_j$ and $u_l$.
 This implies that after arbitrarily fixing distinct IDs in the radius-$t$ neighborhood of $u_j$ and $u_l$, it is possible to complete the ID assignment of the entire graph $G$ in such a way that the union of the radius-$t$ neighborhoods of the endpoints of each edge in $G$ does not contain repeated IDs. If we run $\A$ with such an ID assignment, it is guaranteed that the output is correct.
 
 Consider the directed graph $\M_{\DDD_{\VV_i}}$ and any two of its nodes $(a,b)$ and $(c,d)$ such that $\{a,b\} \in \tilde{D}$ and $\{c,d\} \in \tilde{D}$. Our choice of $\tilde{D}$ implies that there exists an assignment of distinct IDs to the radius-$t$ neighborhood of both $u_j$ and $u_l$ such that the output labels of $(u_j, \{u_j, u_{j-1}\})$, $(u_j, \{u_j, u_{j+1}\})$, $(u_l, \{u_l, u_{l-1}\})$, and $(u_l, \{u_l, u_{l+1}\})$ are $a$, $b$, $c$, and $d$, respectively.
 We complete the ID assignment of the entire graph $G$ in such a way that the radius-$t$ neighborhood of each edge in $G$ does not contain repeated IDs.
 For each node $u_y$ with $j < y < l$, the pair of the two output labels of $(u_y \{u_y, u_{y-1}\})$ and $(u_y, \{u_y, u_{y+1}\})$ resulting from $\A$ is also a node in $\M_{\DDD_{\VV_i}}$. Hence there exists a walk $(a,b) \leadsto (c,d)$ of length $4t+3+d$, for any $0 \leq d \leq |\Sigma|^2$. Since this holds for all choices of  $(a,b)$ and $(c,d)$  such that  $\{a,b\} \in \tilde{D}$ and $\{c,d\} \in \tilde{D}$, all members in $\tilde{D}$ are in the same strongly connected component.
 Furthermore, \cref{lem-inflex} implies that this strongly connected component is path-flexible, as the length of the walk can be any integer in between $4t+3$ and $4t+3+|\Sigma|^2$.

\subparagraph{Part 2:  $\DDD_{v, e_1, e_2}^\A \subseteq \DDD_i$} For this part, 
we use \cref{lem-layer-C-property}, which shows that in the graph  $G = G_{\sR,k+1}^\ast$ there are a path $P_1$ from $v$ to a central node in layer $(\sC,i)$ through $e_1$ and a path  $P_2$ from $v$ to a central node  in layer $(\sC,i)$ through $e_2$, and these paths use only nodes in $S_{\sC, i}$.

Consider the output labels resulting from running $\A$.
Let $a$ be the label of the half-edge $(v, e_1)$ and let $b$ be the label of the half-edge $(v, e_2)$. We have $\{a,b\} \in \DDD_{v, e_1, e_2}^\A \subseteq \DDD_{\VV_i}$. 
By taking the output labels in $P_1$ and $P_2$ resulting from running $\A$, we obtain two walks in the directed graph $\M_{\DDD_{\VV_i}}$: $(a,b) \leadsto (c,d)$ and $(b,a) \leadsto (e,f)$, where both $(c,d)$ and $(e,f)$ are nodes in $V(\M_{\DDD_{\VV_i}})$ such that $\{a,b\} \in \tilde{D}$ and $\{c,d\} \in \tilde{D}$. By taking the opposite directions, we also obtain two walks: $(d,c) \leadsto (b,a)$ and $(f,e) \leadsto (a,b)$. Hence $\{a,b\}$ is in the same strongly connected component of $\DDD_{\VV_i}$ as the members in $\tilde{D}$.

The same argument can be applied to all $\{a,b\} \in \DDD_{v, e_1, e_2}^\A$. The reason is that for each $\{a,b\} \in \DDD_{v, e_1, e_2}^\A$  there is an assignment of distinct IDs such that $\{a,b\}$ is the multiset of the two labels of  $(v, e_1)$ and $(v, e_2)$. Hence we conclude that all members in $\DDD_{v, e_1, e_2}^\A$ are within the same strongly connected component as that of members in $\tilde{\DDD}$, so  $\DDD_{v, e_1, e_2}^\A \subseteq \DDD_i$.
\end{proof}
 
 Applying \cref{lem-base-case,lem-inductive-step-CR,lem-inductive-step-CR,lem-inductive-step-RC} from $S_{\sR,1}$ all the way up to the last subset  $S_{\sR,k+1}$, we obtain the following result.
 
\begin{lemma}[Lower bound for the case $d_\Pi = k$]\label{lem-lower-finite}
If $d_\Pi = k$ for a finite integer $k$, then  $\Pi$ requires $\Omega(n^{1/k})$ rounds to solve on trees of maximum degree $\Delta$.
\end{lemma}
\begin{proof}
Assume that there is a $t$-round algorithm solving $\Pi$ on $G$.
By \cref{lem-base-case,lem-inductive-step-CR,lem-inductive-step-CR,lem-inductive-step-RC}, we infer that the induction hypothesis for the last subset  $S_{\sR,k+1}$ holds. By \cref{lem-containment}, $S_{\sR,k+1} \neq \emptyset$, so there is a node $v$ in $G$  such that  $\VV_v^\A \subseteq \VV_{k+1}$. Therefore, the correctness of $\A$ implies that $\VV_{k+1} \neq \emptyset$, which implies that $(\VV_1, \DDD_1, \VV_2, \DDD_2, \ldots, \VV_{k+1})$ chosen in the induction hypothesis is a good sequence, contradicting the assumption that $d_\Pi = k$.
Hence such a $t$-round algorithm $\A$ that solves $\Pi$ does not exist. As this argument holds for all integers $t$ and $t = \Omega(n^{1/k})$, where $n$ is the number of nodes in $G$, we conclude the proof.
\end{proof}

Now we are ready to prove \cref{thm-unrooted-characterization}.

\begin{proof}[Proof of \cref{thm-unrooted-characterization}]
The upper bound part of the theorem follows from \cref{lem-upper-finite,lem-upper-infinite}. The lower bound part of the theorem follows from \cref{lem-lower-0,lem-lower-finite}.
\end{proof}

\subsection{Complexity of the characterization}\label{sect-time}

In this section, we prove \cref{thm-unrooted-poly-time}. We are given a description of an $\LCL$ problem $\Pi=(\Delta, \Sigma,\VV,\EE)$ on $\Delta$-regular trees. We assume that the description is given in the form of listing the multisets in $\VV$ and $\EE$. Therefore, the description length of $\Pi$ is $\ell = O(\log |\Sigma|) \cdot (|\EE|\cdot 2 + |\VV| \cdot \Delta)$. Here we allow $\Delta$ to be a non-constant, as a function of $\ell$. We will design an algorithm that computes all possible good sequences $(\VV_1, \DDD_1, \VV_2, \DDD_2, \ldots, \VV_k)$  in time polynomial in $\ell$, and this allows us to compute $d_\Pi$. As, the main objective of this section is to show that the problem is polynomial-time solvable, we do not aim to optimize the time complexity of our algorithm.

For the case of $d_\Pi = \infty$, there are good sequences that are arbitrarily long. Recall that we have $\VV_1 \supseteq \VV_2 \supseteq \cdots \supseteq \VV_k$. Hence if $k > |\VV|$, there must exist some index $1 \leq i < k$ such that $\VV_i = \VV_{i+1}$. This immediately implies that $\VV_i = \VV_{i+1} = \VV_{i+2} = \cdots$, due to the following reasoning. The fact that $\VV_i = \VV_{i+1}$ implies that $\VV_{i+1} = \trim(\VV_{i}\upharpoonright_{\DDD_{i}}) = \VV_{i}\upharpoonright_{\DDD_{i}} = \VV_{i}$, so $\DDD_{i} = \DDD_{\VV_i}$. This means that $\DDD_{i} = \DDD_{\VV_i}$ itself is  the only element of $\flexSCC(\DDD_{\VV_i})$. Therefore, starting from $\VV_i$, the multisets $\DDD_{i} = \DDD_{\VV_i}$ and $\VV_{i+1} = \VV_i$ are uniquely determined. Similarly, we have  $\DDD_{j} = \DDD_{i}$ and $\VV_{j} = \VV_i$ for all $j \geq i$. We conclude that any good sequence with $k > |V|$ must stabilize at some point $i \leq |\VV|$, in the sense that $\DDD_{j} = \DDD_{i}$ and $\VV_{j} = \VV_i$ for all $j \geq i$.

\subparagraph{High-level plan}
Recall that the rules for a good sequence are as follows:
\begin{align*}
    \VV_i & =
    \begin{cases}
    \trim(\VV) & \text{if $i=1$},\\
    \trim(\VV_{i-1}\upharpoonright_{\DDD_{i-1}}) &  \text{if $i>1$},\\
    \end{cases}\\
    \DDD_i &\in \flexSCC(\DDD_{\VV_i}).
\end{align*}

To compute all  good sequences $(\VV_1, \DDD_1, \VV_2, \DDD_2, \ldots, \VV_k)$, we go through  all choices of $\DDD_i \in \flexSCC(\DDD_{\VV_i})$ and apply the rules recursively until we cannot proceed any further. The process stops when $\VV_i = \VV_{i-1}$ (the sequence stabilizes) or $\VV_i = \emptyset$ (the sequence ends). 

In time $O(|\Sigma|^2)$, we build a look-up table that allows us to check whether $\{\alpha, \beta\} \in \EE$  in $O(1)$ time, for any given size-2 multiset $\{\alpha, \beta\}$. In the subsequent discussion, we assume that such a look-up table is available.

We start with describing the algorithm for computing $\trim(\SSS)$ for a given $\SSS \subseteq \VV$.

\begin{lemma}[Algorithm for $\trim$]\label{lem-compute-trim}
The set $\trim(\SSS)$ can be computed in $O(\Delta |\SSS| |\Sigma|^2)$ time, for any given $\SSS \subseteq \VV$.
\end{lemma}
\begin{proof}
  We write $T_i'$ to denote the tree resulting from adding one extra edge $e^\ast$ incident to the root $r$ of $T_i$. We write $\Sigma_i$ to denote set of all possible  $\sigma \in \SSS$ such that there is a correct labeling of $T_i'$ where the node configuration of each degree-$\Delta$ node is in $\SSS$ and the half-edge label of $(r,e^\ast)$ is $\sigma$. The set $\Sigma_i$ can be computed recursively as follows.
\begin{itemize}
    \item For the base case,  $\Sigma_1$ is the set of all labels appearing in $\SSS$. 
    \item For the inductive step, each $\sigma \in \Sigma_{i-1}$ is added to $\Sigma_{i}$ if there exists $C \in \SSS$ such that $\sigma \in C$ and each of the $\Delta-1$ labels $\alpha$ in $C \setminus \{\sigma\}$ satisfies that $\{\alpha, \beta\} \in \EE$ for some $\beta \in \Sigma_{i-1}$.
\end{itemize}

 Using the above look-up table, given that  $\Sigma_{i-1}$ has been computed, the computation of  $\Sigma_{i}$ costs $O(\Delta |\SSS| |\Sigma_{i-1}|) = O(\Delta |\SSS| |\Sigma|)$ time. 
Clearly, we have $\Sigma_1 \supseteq \Sigma_2  \supseteq \cdots$, and whenever $\Sigma_i = \Sigma_{i+1}$, the sequence stabilizes:  $\Sigma_i = \Sigma_{i+1} = \Sigma_{i+2} = \cdots$.  We write $\Sigma^\ast$ to denote the fix point $\Sigma_i$ such that $\Sigma_i = \Sigma_{i+1} = \Sigma_{i+2} = \cdots$.
It is clear that $i \leq |\Sigma|$, so the fixed point $\Sigma^\ast$ can be computed in  $O(\Delta |\SSS| |\Sigma|^2)$ time.

Given the fix point $\Sigma^\ast$, the set  $\trim(\SSS)$ can be computed as follows. Observe that the tree $T_i^\ast$ is simply the result of merging  $\Delta$ trees $T_{i-1}'$ by merging the $\Delta$ degree-$1$ endpoints of $e^\ast$ into one node $r$. Therefore, $C \in \trim(\SSS)$ if and only if each of the $\Delta$ labels $\alpha \in C$ satisfies that $\{\alpha, \beta\} \in \EE$ for some $\beta \in \Sigma^\ast$. Using this characterization, given the fix point $\Sigma^\ast$, the set  $\trim(\SSS)$ can be similarly computed in $O(\Delta |\SSS| |\Sigma|)$ time. We conclude the following result.
\end{proof}

  Next, we give an algorithm that computes all path-flexible strongly connected components $\DDD' \in \flexSCC(\DDD_{\SSS})$ and their corresponding restriction $\SSS \upharpoonright_{\DDD'}$, for any given $\SSS \subseteq \VV$.

\begin{lemma}[Algorithm for $\flexSCC$]\label{lem-compute-flex-SCC}
The set of all  $\DDD' \in \flexSCC(\DDD_{\SSS})$ and their corresponding restrictions $\SSS \upharpoonright_{\DDD'}$ can be computed in $O(\Delta^8 |\SSS|^4)$ time, for any given $\SSS \subseteq \VV$.
\end{lemma}

\begin{proof}
Observe that $|\DDD_\SSS| = O(\Delta^2 |\SSS|)$, so the directed graph $\M_{\DDD_\SSS}$ has $O(\Delta^2 |\SSS|)$ nodes and $O(\Delta^4 |\SSS|^2)$ edges. 

Using the definition of \cref{def-scc}, testing whether $\{a,b\} \sim \{c,d\}$ for any $\{a,b\} \in \DDD_{\SSS}$ and $\{c,d\} \in \DDD_{\SSS}$ costs $O(|E(\M_{\DDD_\SSS})|) = O(\Delta^4 |\SSS|^2)$ time by doing four $s$-$t$  reachability computation. By going over all $\{a,b\} \in \DDD_{\SSS}$ and $\{c,d\} \in \DDD_{\SSS}$, 
 the set of all strongly connected components of ${\DDD_\SSS}$ can be computed in time $O(|E(\M_{\DDD_\SSS})|\cdot |V(\M_{\DDD_\SSS})|^2) = O(\Delta^8 |\SSS|^4)$.

For each strongly connected component $\DDD'$ of ${\DDD_\SSS}$, to decide whether $\DDD'$ is path-flexible, it suffices to pick one element $\{a,b\} \in \DDD'$ and check if $\{a,b\}$ is path-flexible. Recall that $\{a,b\}$ is path-flexible if  there exists an integer $K$ such that for each integer $k \geq K$, there exist length-$s$ walks  $(a,b) \leadsto (a,b)$,  $(a,b) \leadsto (b,a)$,  $(b,a) \leadsto (a,b)$, and  $(b,a) \leadsto (b,a)$ in  $\M_{\DDD}$.  

The fact that $(a,b) \in \DDD'$ and  $\DDD'$ is a strongly connected component of ${\DDD_\SSS}$ implies the existence of walks $(a,b) \leadsto (a,b)$,  $(a,b) \leadsto (b,a)$,  $(b,a) \leadsto (a,b)$, and  $(b,a) \leadsto (b,a)$ in  $\M_{\DDD}$.  Therefore, the task for deciding whether $\DDD'$ is path-flexible is reduced to the following task. Given a node $s=(a,b)$ in the directed graph $\M_{\DDD_\SSS}$, let $L$ be the set of possible lengths of an $s \leadsto s$ walk, check if there exists an integer $K$ such that
$L$ contains all integers that are at least $K$. Such an integer $K$ exists if and only if the greatest common divisor $\gcd(L)$ of $L$ is not one. Define $L'$ as the set of numbers in $L$  that are at most $2|V(\M_{\DDD_\SSS})|-1$. As shown in~\cite{lcls_on_paths_and_cycles}, we have $\gcd(L) = \gcd(L')$. 

The computation of $\gcd(L')$ can be done in $O(|V(\M_{\DDD_\SSS})|^3)$ time, as follows. From $i = 0$ up to $i = 2 |V(\M_{\DDD_\SSS})| - 1$, we compute a list of nodes $U_i \subseteq V(\M_{\DDD_\SSS})$ such that $v \in U_i$ if there is a walk $s \leadsto v$ of length $i$. Given $U_{i-1}$, it takes $O(|V(\M_{\DDD_\SSS})|^2)$ time to compute $U_{i}$, as we just need to go over all $O(|V(\M_{\DDD_\SSS})|^2)$ edges between the nodes in $V(\M_{\DDD_\SSS})$. The summation of the time complexity $O(|V(\M_{\DDD_\SSS})|^3)$, over all strongly connected component of $V(\M_{\DDD_\SSS})$, is $O(|V(\M_{\DDD_\SSS})|^4)=O(\Delta^8 |\SSS|^4)$.

To summarize, the  computation of  $\flexSCC(\DDD_{\SSS})$  costs $O(\Delta^8 |\SSS|^4)$ time. Given that we have computed all path-flexible strongly connected components $\DDD' \in \flexSCC(\DDD_{\SSS})$, the computation of the restriction $\SSS \upharpoonright_{\DDD'}$ for all path-flexible strongly connected components $\DDD' \in \flexSCC(\DDD_{\SSS})$ costs $O(|\DDD_{\SSS}| + \Delta^2 |\SSS|) = O(\Delta^2 |\SSS|)$ time, as we just need to check for each $C \in \SSS$ whether there is $\DDD' \in \flexSCC(\DDD_{\SSS})$ such that all $\Delta^2$ size-$2$ sub-multisets of $C$ belong to the $\DDD'$.
\end{proof}

 Combining \cref{lem-compute-trim,lem-compute-flex-SCC}, we obtain the following result.

\begin{lemma}[Computing all good sequences]\label{lem-compute-good-seq}
The set of all good sequences can be computed in $O(\Delta |\VV|^2 |\Sigma|^2 +\Delta^8 |\VV|^5)$ time, for any given $\LCL$ problem $\Pi=(\Delta, \Sigma, \VV,\EE)$.
\end{lemma}
\begin{proof}
Combining \cref{lem-compute-trim,lem-compute-flex-SCC}, we infer that given $\VV_i \subseteq \VV$, the cost of computing all possible $(\DDD_{i}, \VV_{i+1})$ is $O(\Delta |\VV_i| |\Sigma|^2 + \Delta^8 |\VV_i|^4)$ time, as the set of $\VV_i \upharpoonright_{\DDD'}$ over all $\DDD' \in \flexSCC(\DDD_{i})$ are disjoint subsets of $\VV_i$.

Since all the sets $\VV_i$  in the depth $i$ of the recursion are disjoint subsets of $\VV$, the total cost for the depth $i$ of the recursion is $O(\Delta |\VV| |\Sigma|^2 + \Delta^8 |\VV|^4)$. 

The recursion stops when $\VV_i = \VV_{i-1}$ or $\VV_i = \emptyset$, so the depth of the recursion is at most $|\VV|$. The reason is that we must have $|\VV_{i}| < |\VV_{i-1}|$ if  $\VV_i \neq \VV_{i-1}$ and $\VV_i \neq \emptyset$. Therefore, the total cost of computing all good sequences is $O(\Delta |\VV|^2 |\Sigma|^2 +\Delta^8 |\VV|^5)$.
\end{proof}

We are ready to prove  \cref{thm-unrooted-poly-time}.

\begin{proof}[Proof of \cref{thm-unrooted-poly-time}]
By \cref{lem-compute-good-seq}, the set of all good sequences can be computed in polynomial time, and we can compute $d_\Pi$ given the set of all good sequences.  If $d_\Pi = k$ is a positive integer, then from the discussion in \cref{sect-upper} we know how to turn a good sequence
$(\VV_1, \DDD_1, \VV_2, \DDD_2, \ldots, \VV_k)$ into a description of an  $O(n^{1/k})$-round algorithm for $\Pi$. If $d_\Pi = \infty$, then similarly a good sequence
$(\VV_1, \DDD_1, \VV_2, \DDD_2, \ldots, \VV_{O(\log n)})$ leads to a description of an $O(\log n)$-round algorithm for $\Pi$.
\end{proof}

\section{Rooted trees}\label{sec:rooted}

In this section, we give a polynomial-time-computable characterization of $\LCL$ problems for regular rooted trees with complexity $O(\log n)$ or $\Theta(n^{1/k})$  for any positive integer $k$.

\subsection{Locally checkable labeling for rooted trees} 
 A rooted tree is a tree where each edge is oriented in such a way that the outdegree of each node is at most $1$.
 A $\delta$-regular rooted tree is a rooted tree where the indegree of each node is either $0$ or  $\delta$.
 The root of a rooted tree is the unique node $v$ with $\outdeg(v) = 0$. Each node $v$ with $\indeg(v) = 0$ is called a leaf.
 For each directed edge $u \rightarrow v$, we say that $u$ is a child of $v$ and $v$ is the parent of $u$. 
An $\LCL$ problem for $\delta$-regular rooted trees is defined as follows.

\begin{definition}[$\LCL$ problems for regular rooted trees]\label{def-lcl-rooted}
For rooted trees, an $\LCL$ problem $\Pi=(\delta, \Sigma,\CC)$ is defined by the following  components.
\begin{itemize}
    \item $\delta$ is a positive integer specifying the maximum indegree. 
    \item $\Sigma$ is a finite set of labels.
    \item $\CC$ is a set of pairs $(\sigma, S)$ such that $\sigma \in \Sigma$ and $S$ is a size-$\delta$ multiset of labels in $\Sigma$.
\end{itemize}
\end{definition}

For notational simplicity, we also write $(\sigma \ : \ a_1 a_2 \cdots a_\delta)$ to denote $(\sigma, S)$, where $\sigma \in \Sigma$ and $S = \{a_1, a_2, \ldots, a_\delta\}$ is a size-$\delta$ multiset of elements in $\Sigma$. 
We call any such $(\sigma, S)$ a node configuration. A node configuration $(\sigma, S)$ is correct if  $(\sigma, S) \in \CC$. In this sense, the set $\CC$ specifies the node constraint.
We define the correctness criteria for a  labeling in \cref{def-correctness-rooted}. 
 
\begin{definition}[Correctness criteria]\label{def-correctness-rooted}
Let $G=(V,E)$ be a rooted tree whose maximum indegree is at most $\delta$.  A solution of $\Pi=(\delta,\Sigma,\CC)$ on $G$ is a labeling that assigns a label in $\Sigma$ to each node in $G$.
\begin{itemize}
    \item For each node $v\in V$ with $\indeg(v) = \delta$, we define its node configuration $C = (\sigma \ : \ a_1 a_2 \cdots a_\delta)$ by setting $\sigma$ as the  label of $v$ and setting $\{a_1, a_2, \ldots, a_\delta\}$ as the multiset of the labels of the $\delta$ children of $v$. We say that the labeling is locally-consistent on $v$ if   $C \in \CC$.
\end{itemize}
 The labeling  is a correct solution if it is locally-consistent on all $v\in V$ with $\indeg(v) = \delta$.
\end{definition}

  Similarly, although  $\Pi=(\delta,\Sigma,\CC)$ is defined for $\delta$-regular rooted trees, \cref{def-correctness-rooted} applies to all rooted trees whose maximum indegree is at most $\delta$.
  We may focus on $\delta$-regular rooted trees without loss of generality, as for any rooted tree $G$ whose maximum indegree is at most $\delta$,
  we may consider the rooted tree $G^\ast$ which is the result of appending leaf nodes to all nodes $v$ in $G$ with $1 < \indeg(v) < \delta$ to increase the indegree of $v$ to $\delta$. This only blows up the number of nodes by at most a $\delta$ factor. Any correct solution of $\Pi$ on $G^\ast$ restricted to $G$ is a correct solution of  $\Pi$ on $G$.

\cref{def-complete-trees-2} is the same as \cref{def-complete-trees} except that we change $\Delta-1$ to $\delta$.

\begin{definition}[Complete trees of height $i$]\label{def-complete-trees-2} We define the rooted trees $T_i$ recursively as follows.
 \begin{itemize}
     \item $T_0$ is the trivial tree with only one node.
     \item $T_i$ is the result of appending $\delta$ trees $T_{i-1}$ to the root $r$.
 \end{itemize}
\end{definition} 

Observe that $T_i$ is the unique maximum-size rooted tree of maximum indegree $\delta$ and radius $i$. All nodes within distance $i-1$ to the root $r$ in  $T_i$ have indegree $\delta$. All nodes whose distance to $r$ is exactly $i$ are leaf nodes.

\begin{definition}[Trimming]
Given a subset $\tilde{\Sigma} \subseteq \Sigma$ of labels, we define $\trim(\tilde{\Sigma})$ as a set of all labels $\sigma \in \tilde{\Sigma}$ such that for each $i \geq 1$ it is possible to find a correct labeling of $T_i$ such that the label of the root is $\sigma$ and the label of the remaining nodes are in $\tilde{\Sigma}$.
\end{definition}

Given any rooted tree $G$ of maximum indegree $\delta$, after labeling the a node $v$ with $\indeg(v)=\delta$ with a label $\sigma \in \trim(\tilde{\Sigma})$, it is always possible to extend this labeling to a complete correct labeling of the subtree rooted at $v$ using only node configurations in $\trim(\tilde{\Sigma})$. Such a labeling extension is possible due to \cref{lem-trimming-2}.

\begin{lemma}[Property of trimming]\label{lem-trimming-2}
Let $\Pi=(\delta,\Sigma,\CC)$ be an $\LCL$ problem.
Let   $\tilde{\Sigma} \subseteq \Sigma$ such that  $\trim(\tilde{\Sigma}) \neq \emptyset$. 
For each label $\sigma \in \trim(\tilde{\Sigma})$, there exists a node configuration  $(\sigma \ : \ a_1 a_2 \cdots a_\delta) \in \CC$ such that $a_i \in \trim(\tilde{\Sigma})$ for all $1 \leq i \leq \delta$.
\end{lemma}
\begin{proof}
Assuming that such a node configuration $(\sigma \ : \ a_1 a_2 \cdots a_\delta) \in \CC$ do not exist,  we derive a contradiction as follows.
We pick $s$ to be the smallest number such that there is no correct labeling of $T_s$ where the label of the root $r$ is in  $\tilde{\Sigma} \setminus \trim(\tilde{\Sigma})$ and the label of each remaining node of $T_s$ is in $\tilde{\Sigma}$.   Such a number $s$ exists due to the definition of  $\trim$.

Now consider a correct labeling of $T_{s+1}^\ast$ where the label of the root $r$ is $\sigma$ and the label of each remaining node is in $\tilde{\Sigma}$. Such a correct labeling exists due to the fact that $\sigma \in \trim(\tilde{\Sigma})$.  Our assumption on the non-existence of $(\sigma \ : \ a_1 a_2 \cdots a_\delta) \in \CC$ such that $a_i \in \trim(\tilde{\Sigma})$ for all $1 \leq i \leq \delta$ implies that the label $b_i$ of one child $u_i$ of the root $r$  of $T_{s+1}$ must be in $\tilde{\Sigma} \setminus \trim(\tilde{\Sigma})$. However, the subtree of $T_{s+1}$ rooted at $u_i$ is isomorphic to the rooted tree $T_{s}$. Since the label $b_i$ of $u_i$ is in  $\tilde{\Sigma} \setminus \trim(\tilde{\Sigma})$, our choice of $s$ implies that the labeling of the subtree of $T_{s+1}$ rooted at $u_i$ cannot be correct, which is a contradiction.
\end{proof}

\subparagraph{Restriction of an $\LCL$ problem}
Given a subset $\tilde{\Sigma} \subseteq \Sigma$ of labels, we define the restriction of $\Pi$ to $\tilde{\Sigma}$ as follows.
\begin{align*}
    \Pi \upharpoonright_{\tilde{\Sigma}} &= (\tilde{\Sigma}, \tilde{\CC}),\\
    \text{where} \ \ \tilde{\CC} &= \ \text{\parbox{0.5\textwidth}{the set of all $(\sigma \ : \ a_1 a_2 \cdots a_\delta) \in \CC$ such that $\sigma \in \tilde{\Sigma}$ and $a_i \in \tilde{\Sigma}$ for all $1 \leq i \leq \delta$.}}
\end{align*}

That is,
$\Pi \upharpoonright_{\tilde{\Sigma}}$ 
is simply the result of removing all labels and node configurations in $\Pi$ that involve labels not in $\tilde{\Sigma}$.

\subparagraph{Path-form of an LCL problem}
Given  an $\LCL$ problem $\Pi=(\delta, \Sigma, \CC)$, we define  its path-form $\PathPi=(1, \Sigma, \PathCC)$ as follows.
\[\PathCC = \text{\parbox{0.7\textwidth}{the set of all $(\sigma \ : \ a)$ such that there exists $(\sigma \ : \ a_1 a_2 \cdots a_\delta) \in \CC$ such that $a \in \{a_1, a_2, \ldots, a_\delta\}$.}}\]

\subparagraph{Automaton for the path-form of an LCL problem}
The path-form $\PathPi=(1, \Sigma, \PathCC)$  of an $\LCL$ problem $\Pi=(\delta, \Sigma, \CC)$ can be interpreted as a directed graph, where the node set is $\Sigma$ and the edge set is $\PathCC$, by viewing each $(\sigma \ : \ a) \in \PathCC$ as a directed edge $a \rightarrow \sigma$.

\subparagraph{Path-flexibility}   With respect to  $\PathPi=(1, \Sigma, \PathCC)$,  we say that $\sigma \in \Sigma$ is path-flexible if there exists an integer $K$ such that for each integer $s \geq K$, there exist a length-$s$ walk $\sigma \leadsto \sigma$ in the directed graph of $\PathPi=(1, \Sigma, \PathCC)$.  
\cref{lem-inflex-2} is useful in lower bound proofs.

\begin{lemma}[Property of path-inflexibility]\label{lem-inflex-2}
Suppose that $\sigma \in \Sigma$ is not path-flexible with respect to $\PathPi=(1, \Sigma, \PathCC)$. Then one of the following holds.
\begin{itemize}
    \item There is no walk $\sigma \leadsto \sigma$ in the directed graph of $\PathPi=(1, \Sigma, \PathCC)$.
    \item There exists an integer $2 \leq x \leq |\Sigma|$ such that for any positive integer $k$ that is not an integer multiple of $x$, there is no length-$k$ walk $\sigma \leadsto \sigma$ in the directed graph of $\PathPi=(1, \Sigma, \PathCC)$.
\end{itemize}
\end{lemma}
\begin{proof}
Since $\sigma$ is not path-flexible,  for any integer $K$ there is an integer $k \geq K$ such that there is no length-$k$ walk $s \leadsto t$. Let $U$ be the set of integers $k$ such that there is a length-$k$ walk $\sigma \leadsto \sigma$ in the directed graph of $\PathPi=(1, \Sigma, \PathCC)$. If $U = \emptyset$, then there is no walk $\sigma \leadsto \sigma$, so the lemma statement holds. For the rest of the proof, we assume $U \neq \emptyset$.
We choose $x=\gcd(U)$ to be the greatest common divisor of $U$. For any integer $k$ that is not an integer multiple of $x$, there is no length-$k$ walk $\sigma \leadsto \sigma$ in  the directed graph of $\PathPi=(1, \Sigma, \PathCC)$.
We must have $x \geq 2$ because there cannot be two co-prime numbers in $U$, since otherwise there exists an integer $K$ such that $U$ includes all integers that are at least $K$, implying that $\sigma$ is  path-flexible.
We also have $x \leq  |\Sigma|$, since the smallest number of $U$ is at most the number of nodes in the directed graph of $\PathPi=(1, \Sigma, \PathCC)$, which is $|\Sigma|$. 
\end{proof}

\subparagraph{Path-flexible  strongly connected components} 
For any strongly connected component $U \subseteq \Sigma$ of the directed graph of $\PathPi=(1, \Sigma, \PathCC)$, it is clear that either all $\sigma \in U$ are path-flexible or  all $\sigma \in U$ are not path-flexible. We say that a strongly connected component $U$ is path-flexible  if  all $\sigma \in U$ are path-flexible. We define
$\flexibility(U)$ as the minimum number $K$ such that  for each integer $k \geq K$ there is an $a \leadsto b$ walk of length $k$ for all choices of source   $a \in U$ and destination  $b \in U$.
It is clear that such a number $K$ exists given that $U$ is a path-flexible strongly connected component.

Given any $\LCL$ problem $\Pi=(\delta, \Sigma, \CC)$ and any $\tilde{\Sigma} \subseteq \Sigma$, we define $\flexSCC(\tilde{\Sigma})$ as the set of all subsets  $U \subseteq \tilde{\Sigma}$ that is a path-flexible strongly connected component of the directed graph of the path-form of $\Pi \upharpoonright_{\tilde{\Sigma}}$.
It is possible that $\flexSCC(\tilde{\Sigma})$ is an empty set, and this happens when all nodes in the directed graph of the path-form of $\Pi \upharpoonright_{\tilde{\Sigma}}$ are not path-flexible.

\cref{lem-flex-scc-2} shows that if we label the two endpoints $v_1$ and $v_{d+1}$ of a sufficiently long directed path $v_1 \leftarrow v_2 \leftarrow \cdots \leftarrow v_{d+1}$ using only labels in $U$, where $U \in \flexSCC(\tilde{\Sigma})$, then it is always possible to complete the labeling of the path using only labels in $U$ in such a way that the entire labeling is correct with respect to $\Pi \upharpoonright_{\tilde{\Sigma}}$.
Specifically, suppose the label of $v_1$ is $\alpha$ and the label of $v_{d+1}$ is $\beta$. Then \cref{lem-flex-scc-2} shows that it is possible to complete the labeling of $v_1 \leftarrow v_2 \leftarrow \cdots \leftarrow v_{d+1}$ by labeling $v_i$ with $\sigma_i \in U$.

\begin{lemma}[Property of path-flexible strongly connected components]\label{lem-flex-scc-2}
Consider any $\LCL$ problem $\Pi=(\delta, \Sigma, \CC)$ and any $\tilde{\Sigma} \subseteq \Sigma$.
Let $U \in \flexSCC(\tilde{\Sigma})$.
For any choices of $\alpha \in U$, $\beta \in U$, and a number $d \geq \flexibility(U)$, there exists a sequence
\[\sigma_1, \sigma_2, \ldots, \sigma_{d+1}\]
of labels in $U$ satisfying the following conditions.
\begin{itemize}
    \item First endpoint: $\sigma_1 = \alpha$.
    \item Last endpoint: $\sigma_{d+1} = \beta$.
    \item Node configurations: for $1 \leq i \leq d$, there is a node configuration $(\sigma \ : \ a_1 a_2 \cdots a_\delta) \in \CC$ meeting the following conditions.
    \begin{itemize}
        \item $\sigma_i = \sigma$.
        \item There is an index $j$ such that $\sigma_{i+1} = a_j$.
        \item $a_l \in \tilde{\Sigma}$ for all $1 \leq l \leq \delta$.
    \end{itemize}
\end{itemize}
\end{lemma}
\begin{proof}
By the path-flexibility of $U$, there exists a length-$d$ walk $\beta \leadsto \alpha$ in the directed graph of the path-form of $\Pi \upharpoonright_{\tilde{\Sigma}}$. We fix \[\sigma_1 \leftarrow \sigma_2 \leftarrow \cdots \leftarrow \sigma_{d+1}\] to be any such walk. For $1 \leq i \leq d$, $\sigma_i \leftarrow \sigma_{i+1}$ is a directed edge in the path-flexible strongly connected component $U$, meaning that $(\sigma_i \ : \ \sigma_{i+1})$ is a node configuration in the path-form of $\Pi \upharpoonright_{\tilde{\Sigma}}$, so there exists a node configuration $(\sigma_i \ : \ a_1 a_2 \cdots a_\delta)$ in $\Pi \upharpoonright_{\tilde{\Sigma}}$ such that $\sigma_{i+1} \in \{a_1, a_2, \ldots, a_\delta\}$.
\end{proof}

\subparagraph{Good sequences} Given an $\LCL$ problem $\Pi=(\delta, \Sigma,\CC)$ on $\delta$-regular rooted trees, we say that a sequence
\[(\SigmaR{1}, \SigmaC{1}, \SigmaR{2}, \SigmaC{2}, \ldots, \SigmaR{k})\]
is \emph{good} if it satisfies the following requirements.
\begin{itemize}
    \item $\SigmaR{1} = \trim(\Sigma)$. That is, we start the sequence from the result of trimming the set $\Sigma$ of all labels in the given $\LCL$ problem $\Pi=(\delta, \Sigma,\CC)$.
    \item  For each $1 \leq i \leq k-1$, $\SigmaC{i} \in \flexSCC(\SigmaR{i})$.  That is, $\SigmaC{i}$ is a path-flexible connected component of the automaton associated with the path-form of the $\LCL$ problem $\Pi \upharpoonright_{\SigmaR{i}}$.
    \item  For each $2 \leq i \leq k$, $\SigmaR{i} = \trim(\SigmaC{i-1})$. That is, $\SigmaR{i}$ is the result of trimming the set $\SigmaC{i-1}$.
    \item $\SigmaR{k} \neq \emptyset$. That is, we require that the last set of labels is non-empty.
\end{itemize}

It is straightforward to see that $\Sigma \supseteq \SigmaR{1} \supseteq \SigmaC{1} \supseteq \SigmaR{2} \supseteq \SigmaC{2} \supseteq \cdots \supseteq \SigmaR{k} \neq \emptyset$.

\subparagraph{Depth of an LCL problem}
We define the depth $d_\Pi$ of an $\LCL$ problem $\Pi=(\delta, \Sigma,\CC)$ on $\delta$-regular rooted trees as follows.
If there is no good sequence, then we set $d_\Pi = 0$.
If there is a good sequence $(\SigmaR{1}, \SigmaC{1}, \SigmaR{2}, \SigmaC{2}, \ldots, \SigmaR{k})$ for each positive integer $k$, then we set $d_\Pi = \infty$.
Otherwise, we set  $d_\Pi$ as  the largest integer $k$ such that there is a good sequence $(\SigmaR{1}, \SigmaC{1}, \SigmaR{2}, \SigmaC{2}, \ldots, \SigmaR{k})$.
We  prove the following results.

\begin{theorem}[Characterization of complexity classes]\label{thm-unrooted-characterization-2}
Let $\Pi=(\delta, \Sigma,\CC)$ be an $\LCL$ problem on $\delta$-regular rooted trees. We have the following.
\begin{itemize}
    \item If $d_\Pi = 0$, then $\Pi$ is unsolvable in the sense that there exists a rooted tree of maximum indegree $\delta$ such that there is no correct solution of $\Pi$ on this rooted tree.
    \item If $d_\Pi = k$ is a positive integer, then the optimal round complexity of $\Pi$ is $\Theta(n^{1/k})$.
    \item If $d_\Pi = \infty$, then $\Pi$ can be solved in $O(\log n)$ rounds.
\end{itemize}
\end{theorem}

In \cref{thm-unrooted-characterization-2}, all the upper bounds hold in the $\CONGEST$ model, and all the lower bounds hold in the $\LOCAL$ model. For example, if $d_\Pi = 5$, then $\Pi$ can be solved in $O(n^{1/5})$ rounds in the $\CONGEST$ model, and there is a matching lower bound $\Omega(n^{1/5})$ in the $\LOCAL$ model.

\begin{theorem}[Complexity of the characterization]\label{thm-unrooted-poly-time-2}
There is a polynomial-time algorithm $\A$ that computes $d_\Pi$  for any given  $\LCL$ problem $\Pi=(\delta, \Sigma,\CC)$ on $\delta$-regular rooted trees. If $d_\Pi = k$ is a positive integer, then  $\A$ also outputs a description of an  $O(n^{1/k})$-round algorithm for $\Pi$. If $d_\Pi = \infty$, then $\A$ also outputs a description of an  $O(\log n)$-round algorithm for $\Pi$.
\end{theorem}

The distributed algorithms returned by the polynomial-time algorithm $\A$ in \cref{thm-unrooted-poly-time-2} are also in the $\CONGEST$ model.

\subsection{Upper bounds}\label{sect-upper-2}

In this section, we prove the upper bound part of \cref{thm-unrooted-characterization-2}. If a good sequence $(\SigmaR{1}, \SigmaC{1}, \SigmaR{2}, \SigmaC{2}, \ldots, \SigmaR{k})$ exists for some positive integer $k$, we show that the $\LCL$ problem $\Pi=(\delta, \Sigma,\CC)$ can be solved in $O(n^{1/k})$ rounds. If a good sequence $(\SigmaR{1}, \SigmaC{1}, \SigmaR{2}, \SigmaC{2}, \ldots, \SigmaR{k})$ exists for all positive integers $k$, we show that $\Pi=(\delta, \Sigma,\CC)$ can be solved in $O(\log n)$ rounds. All these algorithms do not require sending large messages and can be implemented in the $\CONGEST$ model.

\subparagraph{Rake-and-compress decompositions} 
Similar to the case of unrooted trees, we will use a variant of rake-and-compress   decomposition for rooted trees. Our rake-and-compress decomposition for rooted trees is also parameterized by three positive integers $\gamma$, $\ell$, and $L$.
A $(\gamma, \ell, L)$ decomposition of a rooted tree $G=(V,E)$ is a partition of the node set 
\[V = \VR{1} \cup \VC{1} \cup \VR{2} \cup \VC{2}   \cup \cdots \cup \VR{L}\] 
satisfying the following requirements. The requirements will be  different from the ones in \cref{sect-upper} due to the difference between rooted trees and unrooted trees.

\subparagraph{Requirements for $\VR{i}$}
For each connected component $S$  of the subgraph  of $G$ induced by $\VR{i}$, it is required that $S$ is a rooted tree meeting the following conditions.
\begin{itemize}
\item All nodes in $S$ have no in-neighbor in $\VC{i} \cup \VR{i+1} \cup \cdots \cup  \VR{L}$.
\item All nodes in $S$ are within distance $\gamma-1$ to $z$.
\end{itemize}  

Here we only require that the nodes in $S$ do not  have in-neighbors from the higher layers of the decomposition. Only the root $z$ of $S$ can possibly have an out-neighbor outside of $S$, and we do not restrict anything about this out-neighbor.

\subparagraph{Requirements for $\VC{i}$}
For each connected component $S$  of the subgraph  of $G$ induced by $\VC{i}$,  $S$ is a directed path $v_1 \leftarrow  v_2 \leftarrow \cdots \leftarrow v_s$ of $s \in [\ell, 2\ell]$ nodes  meeting the following conditions.
\begin{itemize}
\item There is a node $u \in \VR{i+1} \cup \cdots \cup  \VR{L}$ such that $u \leftarrow v_1$.
\item There is a node $w \in \VR{i+1} \cup \cdots \cup  \VR{L}$ such that $v_s\leftarrow w$.
\item Other than $u$ and $w$, all the remaining neighbors of $S$ are not in $\VR{i+1} \cup \VC{i+1} \cup \cdots \cup  \VR{L}$.
\end{itemize}

The rest of the section is organized as follows. In \cref{sect-rake-and-compress-alg}, we will design distributed algorithms that efficiently compute a $(\gamma, \ell, L)$ decomposition, for certain choices of parameters. In \cref{sect-rake-and-compress-use}, we use our $(\gamma, \ell, L)$ decomposition algorithms to prove the upper bound part of \cref{thm-unrooted-characterization-2}.

\subsubsection{Algorithms for rake-and-compress decomposition}\label{sect-rake-and-compress-alg}

Given a rooted tree $G=(V,E)$, we define the two operations \rake\ and \compress\ as follows, where the operation \compress\ depends on a parameter $\ell$, which is a positive integer.

\begin{itemize}
    \item The operation \rake: Remove all nodes $v \in V$ with $\indeg(v) = 0$.
    \item The operation \compress: Remove all nodes $v \in V$ such that $v$ belongs to an $\ell$-node directed path $P = v_1 \leftarrow v_2 \leftarrow \cdots \leftarrow v_\ell$ such that $\indeg(v_i) = \outdeg(v_i) = 1$ for each $1 \leq i \leq s$.
\end{itemize}

Intuitively, the \rake\ operation removes the set of all leaf nodes, and the \compress\ operation removes the set of all nodes that belong to an $s$-node directed path consisting of only degree-$2$ nodes.

\subparagraph{The decomposition algorithm}
Recall that our goal is to find a $(\gamma, \ell, L)$ decomposition of a rooted tree $G=(V,E)$, which is a partition
$V = \VR{1} \cup \VC{1} \cup \VR{2} \cup \VC{2}   \cup \cdots \cup \VR{L}$ meeting all the requirements. We will first describe our algorithm, and then we analyze for which combinations of parameters  $(\gamma, \ell, L)$  our algorithm works.

Our algorithm for finding such a decomposition is as follows. For $i = 1, 2, \ldots$, perform $\gamma$ \rake\ operations and then perform one \compress\ operation. We initially set $\VR{i}$ to be the set of nodes removed during a \rake\ operation in the $i$th iteration. Similarly, we initially set  $\VC{i}$ to be  the set of nodes removed during the \compress\ operation in the $i$th iteration.

It is clear that these sets $\VR{i}$ and $\VC{i}$ already satisfy all the specified requirements, except that a connected component of the subgraph induced by $\VC{i}$ may be a path whose number of nodes exceeds $2 \ell$. Similar to existing rake-and-compress decomposition algorithms~\cite{balliu21rooted-trees,CP19timeHierarchy,Chang20}, to fix this, we will do a post-processing step which promotes some nodes from $\VC{i}$ to $\VR{i}$ to break long paths of $\VC{i}$ into small paths, for each $i$. We need the following definition.

\begin{definition}[Ruling set]
Let $P$ be a path.  A subset $I\subset V(P)$ is called an $(\alpha,\beta)$-independent set
if the following conditions are met:
(i) $I$ is an independent set that does not contain either endpoint of $P$, and
(ii) each connected component of the subgraph induced by $V(P) \setminus I$ has at least $\alpha$ node and at most $\beta$ node,
unless $|V(P)|<\alpha$, in which case $I=\emptyset$.
\end{definition}

It is a well-known~\cite{balliu19lcl-decidability,CP19timeHierarchy,ColeVishkin86} that an
$(\ell,2\ell)$-independent set of a path $P$ can be computed in $O(\log^\ast n)$ rounds deterministically in the $\CONGEST$ model when $\ell = O(1)$. In our post-processing step, we simply compute an
$(\ell,2\ell)$-independent set $I$   in each connected component induced by $\VC{i}$, in parallel for all $i$. Then we promote the nodes in $I$ from layer  $\VC{i}$ to layer  $\VR{i}$. After this promotion, it is clear that the  decomposition 
\[V = \VR{1} \cup \VC{1} \cup \VR{2} \cup \VC{2}   \cup \cdots \cup \VR{L}\]
is a $(\gamma, \ell, L)$ decomposition  meeting all the requirements, where  $L$ can be  any number such that no node remains after the $\gamma$ \rake\ operations in the $L$th iteration.

Assuming that $\ell = O(1)$, the round complexity of computing the decomposition is clearly $O(\gamma L) + O(\log^\ast n)$ in the $\CONGEST$ model, where $O(\gamma L)$ is the round complexity for executing  $L$ iterations of the rake-and-compress process, and $O(\log^\ast n)$ is the cost for the post-processing step.

\subparagraph{Number of layers}
Next, we consider the following question. Given two positive integers $\gamma$ and $\ell$, what is the smallest number $L$  such that no node remains after the $\gamma$ \rake\ operations in the $L$th iteration, for any given input $n$-node rooted tree $G=(V, E)$?

To answer this question, we consider the following notation. For each node $v \in V$,  we write $U_{\sR, i}^v$ to denote the set of nodes in the subtree rooted at $v$ right after we finish the $\gamma$ \rake\ operation in the $i$th iteration. Similarly, we write $U_{\sC, i}^v$ to denote the set of nodes in the subtree rooted at $v$ right after we finish the \compress\ operation in the $i$th iteration.  

In these definitions, the notion of subtree is with respect to the subgraph induced by the nodes that are not yet removed, not with respect to the original rooted tree. In particular, if $v$ is already removed before the $i$th iteration, then we have $U_{\sR, i}^v = U_{\sC, i}^v = \emptyset$. For notational simplicity, we write $U_{\sC, 0}^v$ to denote the set of nodes in the subtree rooted at $v$ in the original rooted tree $G$.

\begin{lemma}[Shrinkage rate]\label{lem-rake-and-compress-layers-aux}
For each $v \in V$ and $i \geq 1$, we have $|U_{\sC, i}^v| < |U_{\sC, i-1}^v| \cdot \frac{2\ell}{\gamma + 2\ell}$.
\end{lemma}
\begin{proof}
We consider the $i$th iteration of the rake-and-compress process and focus on the set of nodes $U_{\sR, i}^v$. We partition $U_{\sR, i}^v = A \cup B \cup C \cup D$ into four parts as follows.
\begin{itemize}
    \item $A$ is the set of nodes $u$ in $U_{\sR, i}^v$ such that $u$  belongs to an $\ell$-node directed path in $U_{\sR, i}^v$ consisting of only nodes whose indegree in  $U_{\sR, i}^v$ is exactly one.
    \item $B$ is the set of nodes $u$ in $U_{\sR, i}^v$ such that the indegree of $u$ in  $U_{\sR, i}^v$ is exactly one and $u \notin A$. 
    \item $C$ the set of nodes $u$ in $U_{\sR, i}^v$ whose indegree in  $U_{\sR, i}^v$ is greater than one.
    \item $D$ the set of nodes $u$ in $U_{\sR, i}^v$ whose indegree in  $U_{\sR, i}^v$ is zero. 
\end{itemize}

  We prove the following inequalities.

\begin{itemize}
\item We have $|A| \leq |U_{\sR, i}^v| - |U_{\sC, i}^v|$, since  $A$ is precisely the set of nodes in $U_{\sR, i}^v$ that will subsequently be removed during the \compress\ operation in the $i$th iteration. The reason that we have an inequality rather than an equality is that all the descendants of $A$ in $U_{\sR, i}^v$ are also not included in $U_{\sC, i}^v$.
    \item We have $|C|+1 \leq |D|$, since the number of leaf nodes in a rooted tree is at least one plus the number of nodes with more than one child.
    \item We have $|B| \leq (\ell-1)(|C|+|D|)$, since the number of connected components induced by indegree-1 nodes in a rooted tree is at most the number of nodes whose indegree is not one, and $B$ is the union of all these connected components of size at most $\ell-1$.
    \item We have $\gamma |D| \leq |U_{\sC, i-1}^v| - |U_{\sR, i}^v|$, since the fact that each leaf node of $U_{\sR, i}^v$ is not removed during the $\gamma$ \rake\ operations in the $i$th iteration implies that it has at least $\gamma$ descendants removed during these  $\gamma$ \rake\ operations. That is, the number $|U_{\sC, i-1}^v| - |U_{\sR, i}^v|$ of nodes in $U_{\sC, i-1}^v$ removed during the $\gamma$ \rake\ operations in the $i$th iteration is at least $\gamma$ times the number $|D|$ of leaf nodes of $U_{\sR, i}^v$.   
\end{itemize}

Combining these four inequalities, we have
\begin{align*}
 |U_{\sC, i}^v| & \leq |U_{\sR, i}^v| - |A|\\
 & = |B| + |C| + |D| \\
 & \leq \ell (|C| + |D|)\\
 &< 2 \ell |D| \\
 &\leq \frac{2 \ell}{ \gamma} (|U_{\sC, i-1}^v| - |U_{\sR, i}^v|)\\
 &\leq \frac{2 \ell}{ \gamma} (|U_{\sC, i-1}^v| - |U_{\sC, i}^v|).
\end{align*}
Hence $|U_{\sC, i}^v| < |U_{\sC, i-1}^v| \cdot \frac{2\ell}{\gamma + 2\ell}$.
\end{proof}

\begin{lemma}[Number of layers]\label{lem-rake-and-compress-layers}
If the inequality $n \cdot \left(\frac{2\ell}{\gamma + 2\ell}\right)^{L-1} \leq \gamma$ holds, then we have $V = \VR{1} \cup \VC{1} \cup \VR{2} \cup \VC{2}   \cup \cdots \cup \VR{L}$. In particular, we have the following.
\begin{itemize}
    \item If $\ell= O(1)$ and $\gamma = 1$, then we may set $L = O(\log n)$ to satisfy the inequality.
    \item If $\ell= O(1)$ and $L = k$,  then we may set $\gamma = O(n^{1/k})$ to satisfy the inequality.
\end{itemize}
\end{lemma}
\begin{proof}
 For any $v \in V$, we have $|U_{\sC, 0}^v| \leq n$, as $U_{\sC, 0}^v$ is the set of nodes in the subtree rooted at $v$ in the original rooted tree $G$. By
\cref{lem-rake-and-compress-layers-aux}, we have $|U_{\sC, L-1}^v| < n \cdot \left(\frac{2\ell}{\gamma + 2\ell}\right)^{L-1} \leq \gamma$, which implies that $v$ must be removed during the $\gamma$ \rake\ operations in the $L$th iteration, if $v$ has not been removed by the time the $L$th begins.  Hence  $V = \VR{1} \cup \VC{1} \cup \VR{2} \cup \VC{2}   \cup \cdots \cup \VR{L}$.
\end{proof}

We are ready to prove the main results of \cref{sect-rake-and-compress-alg}.

\begin{lemma}[$O(\log n)$-round rake-and-compress algorithm]
\label{lem-rake-and-compress-alg-log}
For any positive integer $\ell = O(1)$, a $(\gamma, \ell, L)$ decomposition of an $n$-node rooted tree with $\gamma = 1$ and $L = O(\log n)$  can be computed in $O(\log n)$ rounds in the $\CONGEST$ model.
\end{lemma}
\begin{proof}
By \cref{lem-rake-and-compress-layers}, we may set $\gamma = 1$ and $L = O(\log n)$  in such a way that we always have $V = \VR{1} \cup \VC{1} \cup \VR{2} \cup \VC{2}   \cup \cdots \cup \VR{L}$. The round complexity for computing the decomposition is $O(\gamma L) + O(\log^\ast n) = O(\log n)$.
\end{proof}

\begin{lemma}[$O(n^{1/k})$-round rake-and-compress algorithm]\label{lem-rake-and-compress-alg-poly}
For any positive integers $\ell = O(1)$ and $k = O(1)$, a $(\gamma, \ell, L)$ decomposition of an $n$-node rooted tree with $\gamma = O(n^{1/k})$ and $L = k$  can be computed in $O(\log n)$ rounds in the $\CONGEST$ model.
\end{lemma}
\begin{proof}
By \cref{lem-rake-and-compress-layers}, we may set $\gamma = O(n^{1/k})$ and $L = k$ in such a way that we always have $V = \VR{1} \cup \VC{1} \cup \VR{2} \cup \VC{2}   \cup \cdots \cup \VR{L}$. The round complexity for computing the decomposition is $O(\gamma L) + O(\log^\ast n) = O(n^{1/k})$.
\end{proof}

\subsubsection{Distributed algorithms via rake-and-compress decompositions}\label{sect-rake-and-compress-use}

In this section, we use \cref{lem-rake-and-compress-alg-log,lem-rake-and-compress-alg-poly} to design distributed algorithms solving a given  $\LCL$ problem $\Pi=(\delta, \Sigma,\CC)$ on $\delta$-regular rooted trees.

\begin{lemma}[Solving $\Pi$ using rake-and-compress decompositions]\label{lem-rake-and-compress-use-2}
 Suppose we are given an $\LCL$ problem $\Pi=(\delta, \Sigma,\CC)$ that admits a good sequence
\[(\SigmaR{1}, \SigmaC{1}, \SigmaR{2}, \SigmaC{2}, \ldots, \SigmaR{k}).\]
Suppose we are given a $(\gamma, \ell, L)$ decomposition of an $n$-node rooted tree $G=(V,E)$ of maximum indegree at most $\delta$
\[V = \VR{1} \cup \VC{1} \cup \VR{2} \cup \VC{2}   \cup \cdots \cup  \VR{L}\] 
with $L = k$ and $\ell = \max\{1, \flexibility( \SigmaC{1}), \ldots, \flexibility(\SigmaC{k-1})\}$. 
Then a correct solution of  $\Pi$ on $G$ can be computed in $O((\gamma+\ell)L)$ rounds in the $\CONGEST$ model.
\end{lemma} 
\begin{proof}
We present an $O((\gamma+\ell)L)$-round $\CONGEST$ algorithm finding a correct solution of  $\Pi$ on $G$.  
In this proof, whenever we say the algorithm labels a node $v$, we mean picking a node configuration  $(\sigma \ : \ a_1, a_2, \ldots, a_{\delta}) \in \CC$ and fixing the labels of $v$ and its $\delta$ children according to the chosen node configuration.

The algorithm processes the nodes of the graph in the order $\VR{L}, \VC{L-1}, \ldots, \VR{1}$. We require that when the algorithm processes $\VR{i}$, the algorithm only uses node configurations  $(\sigma \ : \ a_1, a_2, \ldots, a_{\delta}) \in \CC$ such that $\sigma$ and all of $a_1, a_2, \ldots, a_{\delta}$ are in $\SigmaR{i}$.  We also require that when the algorithm processes $\VC{i}$, the algorithm uses node configurations  $(\sigma \ : \ a_1, a_2, \ldots, a_{\delta}) \in \CC$ such that $\sigma \in \SigmaC{i}$ and all of  $a_1, a_2, \ldots, a_{\delta}$ are in $\SigmaR{i}$.

\subparagraph{Labeling $\VR{i}$} Suppose the algorithm has finished labeling the   nodes in $\VC{i} \cup \VR{i+1} \cup \cdots \cup  \VR{L}$.
The algorithm then labels each connected component $S$  of the subgraph  of $G$ induced by $\VR{i}$, in parallel and using $O(\gamma)$ rounds in the $\CONGEST$ model, as follows.

Recall that The set $S$ induces a rooted tree of height at most $\gamma-1$ such that no node in $S$ has an in-neighbor in  $\VC{i} \cup \VR{i+1} \cup \cdots \cup  \VR{L}$. The root $z$ of $S$ may have an out-neighbor $u$. 

If $u \notin \VC{i} \cup \VR{i+1} \cup \cdots \cup  \VR{L}$ of $u$ does not exist, then the label of $z$ is not fixed yet. In this case, we choose any node configuration  $(\sigma \ : \ a_1, a_2, \ldots, a_{\delta}) \in \CC$ such that $\sigma$ and all of $a_1, a_2, \ldots, a_{\delta}$ are in $\SigmaR{i}$ to label $z$ and its children.
Such a node configuration exists because of \cref{lem-trimming-2}, as we recall that $\SigmaR{i} = \trim(\SigmaC{i-1})$ (if $i>1$) and $\SigmaR{1} = \trim(\Sigma)$ (if $i=1$).

If  $u \in \VC{i} \cup \VR{i+1} \cup \cdots \cup  \VR{L}$, then the label of $u$ is fixed to be some label $\sigma \in \SigmaR{i}$, due to the above requirement of our algorithm for labeling $\VC{i} \cup \VR{i+1} \cup \cdots \cup  \VR{L}$, as we recall that  $\SigmaR{i} \supseteq \SigmaC{i} \supseteq \cdots \supseteq \SigmaR{k}$.
In this case, we choose any node configuration  $(\sigma \ : \ a_1, a_2, \ldots, a_{\delta}) \in \CC$ such that  all of $a_1, a_2, \ldots, a_{\delta}$ are in $\SigmaR{i}$ to label $z$ and its children. Similarly, the existence of such a node configuration is due to 
 \cref{lem-trimming-2} and the fact that $\sigma \in \SigmaR{i}$.

The node configuration for the remaining nodes in $S$ can be fixed similarly. We start processing a node $v \in S$ once the node configuration of its parent $u$ is fixed. Our requirement for labeling $\VR{i}$ ensures that the label of $v$ is fixed to be some $\sigma \in \SigmaR{i}$, so we can choose any node configuration  $(\sigma \ : \ a_1, a_2, \ldots, a_{\delta}) \in \CC$ where all of $a_1, a_2, \ldots, a_{\delta}$ are in $\SigmaR{i}$  and use this node configuration for $v$ to label its children.
The round complexity of labeling $S$ is $O(\gamma)$ because $S$ is a rooted tree of depth at most $\gamma -1$.

\subparagraph{Labeling $\VC{i}$}
 Suppose the algorithm has finished labeling the   nodes in $\VR{i+1} \cup \VC{i+1} \cup \cdots \cup  \VR{L}$.
The algorithm then labels each connected component $S$  of the subgraph  of $G$ induced by $\VC{i}$, in parallel and using $O(\ell)$ rounds in the $\CONGEST$ model, as follows.

The set $S$ has the property that there are two nodes $u$ and $v$ in $\VR{i+1} \cup \VC{i+1} \cup \cdots \cup  \VR{L}$ adjacent to $S$ such that the subgraph induced by  $S \cup \{u,v\}$ is a directed path $u \leftarrow v_1 \leftarrow v_2 \leftarrow \cdots, v_s \leftarrow v$, with $s \in [\ell, 2\ell]$.

Similarly, the requirement of the choice of  node configurations for  $\VR{i+1} \cup \VC{i+1} \cup \cdots \cup  \VR{L}$ ensures that 
the labels of $u$, $s_1$, and $v$ have been fixed to be some labels in $\SigmaC{i}$, as  we recall that  $\SigmaC{i} \supseteq \SigmaR{i+1} \supseteq \cdots \supseteq \SigmaR{k}$.

Now, our task is to assign node configurations to $v_1, v_2, \ldots, v_s$ in such a way that the labels used to label $v_1, v_2, \ldots, v_s$ are in 
$\SigmaC{i}$ and the labels used to label their children are in $\SigmaR{i}$.

To find such a labeling, we use  \cref{lem-flex-scc-2}. Specifically, recall that the length of the path  $v_1 \leftarrow v_2 \leftarrow \cdots, v_s \leftarrow v$ is $s  \geq \ell \geq \flexibility(\SigmaC{i})$ by our choice of $\ell$.  We let $\alpha$ be the existing label of $v_1$ and let $\beta$ be the existing label of $v$.
Recall that $\SigmaC{i} \in \flexSCC(\SigmaR{i})$, so we may apply \cref{lem-flex-scc-2} with $U = \SigmaC{i}$, $\tilde{\Sigma} = \SigmaR{i}$, $d = s$, and our choices of $\alpha \in U$ and $\beta \in U$.

\cref{lem-flex-scc-2}  returns a sequence of labels $\alpha = \sigma_1, \sigma_2, \ldots, \sigma_{s+1} = \beta$.
For each $1 \leq j \leq s$, we use $\sigma_j$ to label $v_j$. Moreover, for each $1 \leq j \leq s$, \cref{lem-flex-scc-2} guarantees that there is a node configuration $(\sigma_j \ : \ \sigma_{j,1} \sigma_{j,2} \cdots \sigma_{j,\delta}) \in \CC$ such that all of  $\sigma_{j,1} \sigma_{j,2} \cdots \sigma_{j,\delta}$ are in $\tilde{\Sigma} = \SigmaR{i}$ and there exists an index $l$ such that $\sigma_{j,l} = \sigma_{j+1}$. Therefore, we may use the labels in this size-$(\delta-1)$ multiset $\{\sigma_{j,1} \sigma_{j,2} \ldots \sigma_{j,\delta}\} \setminus \{\sigma_{j,l}\}$ to label the remaining $\delta-1$ children of $v_j$, so that the node configuration of $v_j$ is  $(\sigma_j \ : \ \sigma_{j,1} \sigma_{j,2} \cdots \sigma_{j,\delta}) \in \CC$.
The round complexity of labeling $S$ is $O(\ell)$ because $S$ is a path of at most $2 \ell$ nodes.

\subparagraph{Summary} The number rounds spent on labeling each part $\VR{i}$ is $O(\gamma)$, and the number rounds spent on labeling each part $\VC{i}$ is $O(\ell)$, so the overall round complexity for solving $\Pi$ given a  $(\gamma, \ell, L)$ decomposition is $O((\gamma+\ell)L)$ rounds in the $\CONGEST$ model.  
\end{proof}

Combining \cref{lem-rake-and-compress-use} with existing algorithms for computing $(\gamma, \ell, L)$ decompositions, we obtain the following results.

\begin{lemma}[Upper bound for the case $d_\Pi = k$]\label{lem-upper-finite-2}
If $d_\Pi = k$ for some positive integer $k$, then $\Pi$ can be solved in $O(n^{1/k})$ rounds in the $\CONGEST$ model.
\end{lemma}
\begin{proof}
In this case, a good sequence $(\SigmaR{1}, \SigmaC{1}, \SigmaR{2}, \SigmaC{2}, \ldots, \SigmaR{k})$ exists.
By \cref{lem-rake-and-compress-alg-poly}, a  $(\gamma, \ell, L)$ decomposition with $\gamma = O(n^{1/k})$ and $L = k$ can be computed in $O(n^{1/k})$ rounds, $\Pi$ can be solved in $O(n^{1/k}) + O((\gamma+\ell)L) = O(n^{1/k})$ rounds  using the algorithm of \cref{lem-rake-and-compress-use-2}. Here both $k$ and $\ell$ are $O(1)$, as they are independent of the number of nodes $n$.
\end{proof}

\begin{lemma}[Upper bound for the case $d_\Pi = \infty$]\label{lem-upper-infinite-2}
If $d_\Pi = \infty$, then $\Pi$ can be solved in $O(\log n)$ rounds in the $\CONGEST$ model.
\end{lemma}
\begin{proof}
In this case, a good sequence $(\SigmaR{1}, \SigmaC{1}, \SigmaR{2}, \SigmaC{2}, \ldots, \SigmaR{k})$ exists for all positive integers $k$. By \cref{lem-rake-and-compress-alg-log}, a  $(\gamma, \ell, L)$ decomposition with $\gamma = 1$ and $L = O(\log n)$ can be computed in $O(\log n)$ rounds. By choosing a good sequence $(\SigmaR{1}, \SigmaC{1}, \SigmaR{2}, \SigmaC{2}, \ldots, \SigmaR{k})$ with $k = L$, $\Pi$ can be solved in  $O(\log n)+ O((\gamma+\ell)L) = O(\log n)$ rounds using the algorithm of \cref{lem-rake-and-compress-use-2}. Similarly, here $\ell = O(1)$, as it is independent of the number of nodes $n$.
\end{proof}

\subsection{Lower bounds}\label{sect-lower-2}

In this section, we prove the lower bound part of \cref{thm-unrooted-characterization-2}. Similar to the case of unrooted trees, in our lower bound proofs, we pick $\gamma$ to be the smallest integer satisfying the following requirements. For each subset $\tilde{\Sigma}\subseteq \Sigma$ and each $\sigma \in \tilde{\Sigma} \setminus \trim(\tilde{\Sigma})$, there exists no correct labeling of $T_\gamma$ where the label of the root $r$ is $\sigma$ and the  label of remaining nodes is in $\tilde{\Sigma}$. Such a number $\gamma$ exists due to the definition of $\trim$.
Recall that  $T_\gamma$ is defined in \cref{def-complete-trees-2}.

\begin{lemma}[Unsolvability for the case $d_\Pi = 0$]\label{lem-lower-0-2}
If $d_\Pi = 0$, then $\Pi$ is unsolvable in the sense that there exists a rooted tree $G$ of maximum indegree $\delta$ such that there is no correct solution of $\Pi$ on $G$.
\end{lemma}
\begin{proof}
We take $G = T_\gamma$. Since $d_\Pi = 0$, we have $\trim(\Sigma) = \emptyset$. Our choice of $\gamma$ implies that there is no correct solution of $\Pi$ on $G = T_\gamma$.
\end{proof}

For the rest of this section, we focus on the case $d_\Pi = k$ is a positive integer. We will prove that $\Pi$ needs $\Omega(n^{1/k})$ rounds to solve in the $\LOCAL$ model.
\cref{def-lb-graphs-2} is similar to \cref{def-lb-graphs}.
The exact choice of   $s = \Theta(t)$ in \cref{def-lb-graphs-2} is to be determined.

\begin{definition}[Lower bound graphs]\label{def-lb-graphs-2}
We let $t$ be any positive integer, and we let $s = \Theta(t)$. 
\begin{itemize}
    \item Define $G_{\sR,1}$ as the rooted tree $T_\gamma$. All nodes in $G_{\sR,1}$  are said to be in layer $(\sR,1)$.
    \item  For each integer $i \geq 1$, define $G_{\sC,i}$ as the result of the following construction. Start with an $s$-node directed path $v_1 \leftarrow v_2 \leftarrow \cdots \leftarrow v_s$.  For each $1 \leq i \leq s -1 $, append $\delta-1$ copies of $G_{\sR,i-1}$ to $v_i$. Append $\delta$ copies of $G_{\sR,i-1}$ to $v_s$. The nodes $v_1, v_2, \ldots, v_s$ are said to be in layer $(\sC,i)$. We call $v_1 \leftarrow v_2 \leftarrow \cdots \leftarrow v_s$ the core path of $G_{\sC,i}$.
    \item  For each integer $i \geq 1$, define $G_{\sC,i}^\circ$ as the result of the construction of $G_{\sC,i}$, with a difference that we append $\delta-1$ copies of $G_{\sR,i-1}$ to $v_s$.
    \item For each integer $i \geq 2$, define $G_{\sR,i}$ as follows. Start with a rooted tree $T_\gamma$.  Append $\delta$ copies of $G_{\sC,i-1}$ to each leaf in $T_\gamma$.
    All nodes in  $T_\gamma$ are said to be in layer $(\sR,i)$.
\end{itemize}
\end{definition}  

The only difference between \cref{def-lb-graphs-2} and \cref{def-lb-graphs} is that here we define a new rooted tree $G_{\sC,i}^\circ$ where we append only $\delta-1$ copies of $G_{\sR,i-1}$ to the last node $v_s$ of the core path, so $\indeg(v_s) = \delta -1$. 
The purpose of this modification is to allow us to concatenate the rooted trees together without violating the maximum indegree bound $\delta$. Specifically, for our lower bound proof in \cref{sect-lower-2}, we define our main lower bound graph $G=(V,E)$ as follows.

\begin{definition}[Main lower bound graph]\label{def-lb-graph-main-2}
Define the rooted tree $G=(V,E)$ as the result of the following construction.
\begin{itemize}
    \item The construction starts with the rooted trees $G_{\sC,1}^\circ, G_{\sC, 2}^\circ, \ldots, G_{\sC, k}^\circ$ and  $G_{\sR,k+1}$.
    \item Let $P_i = v_1^i \leftarrow v_2^i \leftarrow \cdots \leftarrow v_s^i$ be the core path of $G_{\sC,i}^\circ$ and let  $r$ be the root of $G_{\sR,k+1}$.
    \item Add the  directed edges $v_s^1 \leftarrow v_1^2, v_s^2 \leftarrow v_1^3, \ldots, v_s^{k-1} \leftarrow v_1^k$, and $v_s^{k} \leftarrow r$.
\end{itemize}
\end{definition}

  It is clear that all nodes in the rooted tree $G$ have indegree either $0$ or $\delta$, and the total number of nodes in $G$ is $O(t^{k})$, if we treat $\gamma$ as a constant independent of $t$. 
Therefore, to show that $\Pi$ requires $\Omega(n^{1/k})$ rounds to solve, it suffices to show that there is no algorithm that solves $\Pi$ within $t$ rounds on $G$.

Similar to \cref{sect-lower}, 
the nodes in $G$ are partitioned into layers $(\sR,1)$, $(\sC,1)$, $(\sR,2)$, $(\sC,2)$, $\ldots$, $(\sR,k+1)$ according to the rules in the above recursive construction, and we order the layers by $(\sR,1) \prec (\sC,1)  \prec (\sR,2)  \prec (\sC,2)  \prec \cdots  \prec (\sR,k+1)$. Recall that in the graph $G$, the nodes in layer $(\sC,i)$ form directed paths of $s$ nodes. We consider the following classification of nodes in layer $(\sC,i)$. Again, we will choose $s = \Theta(t)$ to be sufficiently large to ensure that central nodes exist.

\begin{definition}[Classification of nodes in layer $(\sC,i)$]\label{def-layer-C-2}
The nodes in the $s$-node directed path  $v_1 \leftarrow v_2 \leftarrow \cdots \leftarrow v_s$ in the construction of $G_{\sC,i}$ and $G_{\sC,i}^\circ$ are classified as follows.
\begin{itemize}
    \item We say that $v_j$ is a front node if $1 \leq j \leq t$.
    \item We say that $v_j$ is a central node if $t+1 \leq j \leq s-t$.
    \item We say that $v_j$ is a rear node if $s-t+1 \leq j \leq s$.
\end{itemize}
\end{definition}

Based on \cref{def-layer-C-2}, we define the following subsets of nodes in $G$.  In \cref{def-subsets-2}, recall that $P_i$ is the core path of the rooted tree $G_{\sC,i}^\circ$ in the construction of $G$ in \cref{def-lb-graph-main-2}.

\begin{definition}[Subsets of nodes in $G$]\label{def-subsets-2}
We define the following subsets of nodes in $G$.
\begin{itemize}
   \item Define $S_{\sR,1}$ as the set of nodes $v$ in $G$ such that the subgraph induced by $v$ and its descendants within radius-$\gamma$ neighborhood of $v$ is isomorphic to $T_\gamma$.
    \item For $2 \leq i \leq k+1$, define $S_{\sR,i}$ as the set of nodes $v$ in $G$ such that the subgraph induced by $v$ and its descendants within radius-$\gamma$ neighborhood of $v$ is isomorphic to $T_\gamma$ and contains only nodes in $S_{\sC,i-1}$.
    \item For $1 \leq i \leq k$, define $S_{\sC,i}$ as the set of nodes $v$ in $G$ meeting one of the following conditions.
    \begin{itemize}
        \item $v$ is in layer $(\sR,i+1)$ or above.
        \item $v \in P_i$ is a central or rear node in layer $(\sC,i)$.
        \item $v \notin P_i$ is a central or front node in layer $(\sC,i)$.      
    \end{itemize}
\end{itemize}
\end{definition}

We prove some basic properties of the sets in \cref{def-subsets-2}. 

\begin{lemma}[Subset containment]\label{lem-containment-2}
We have $S_{\sR,1} \supseteq S_{\sC,1} \supseteq  \cdots \supseteq S_{\sR,k+1} \neq \emptyset$.
\end{lemma}
\begin{proof}
The claim that $S_{\sC,i} \supseteq S_{\sR,i+1}$ follows from the definition of $S_{\sR,i+1}$ that $v \in S_{\sC,i}$ is a necessary condition for $v \in S_{\sR,i+1}$.
To prove claim that $S_{\sR,i} \supseteq S_{\sC,i}$, we recall that $v \in S_{\sC,i}$ implies that  $v$ is in layer $(\sC,i)$ or above. By the construction of $G$, the subgraph induced by $v$ and its descendants within the radius-$\gamma$ neighborhood of $v$ is isomorphic to $T_\gamma$ and contains only nodes in  layer $(\sR,i)$ or above.  
Since all nodes in layer $(\sR,i)$ or above are in $S_{\sC,i-1}$. we infer that $v \in S_{\sC,i}$ implies $v \in S_{\sR,i}$. 

To see that $S_{\sR,k+1} \neq \emptyset$, consider the node $r$ in the construction of $G$.
Since $r$ is the root of $G_{\sR,k+1}$, the subgraph induced by $r$ and its descendants within radius-$\gamma$ neighborhood of $r$ is isomorphic to $T_\gamma$ and contains only nodes in  layer $(\sR,k+1)$. We know that all nodes in layer $(\sR,k+1)$ are in $S_{\sC,k}$, so $r \in S_{\sR,k+1}$.
\end{proof}

\begin{lemma}[Property of $S_{\sC,i}$]\label{lem-layer-C-property-2}
For each node $v \in S_{\sC,i}$, either one of the following holds.
\begin{itemize}
    \item $v$ is a central node in layer $(\sC,i)$.
    \item For each child $u$ of $v$ such that $u \in S_{\sC,i}$, there exists a directed path $P = w_1 \leftarrow \cdots \leftarrow v \leftarrow u \leftarrow \cdots, w_2$ such that $w_1 \in P_i$ and $w_2 \notin P_i$ are central nodes in layer $(\sC,i)$ and all nodes in $P$ are in $S_{\sC,i}$.
\end{itemize}
\end{lemma}
\begin{proof}
We assume that  $v \in S_{\sC,i}$ is not a central node in layer $(\sC,i)$. Consider any child $u$ of $v$ such that $u \in S_{\sC,i}$. The goal of the proof is to find a path $P = w_1 \leftarrow \cdots \leftarrow v \leftarrow u \leftarrow \cdots \leftarrow w_2$  such that all nodes of $P$ are in $S_{\sC,i}$, and $w_1 \in P_i$ and $w_2 \notin P_i$ are central nodes in layer $(\sC,i)$. The existence of such a directed path $P$ follows from a simple observation that $S_{\sC,i}$ induces a connected subtree where all the leaf nodes are central nodes in layer $(\sC,i)$ that are not in $P_i$ and the root node is a central node in layer $(\sC,i)$ that is in $P_i$.

Specifically, the directed path $P$ can be constructed as follows.
 To construct the part $w_1 \leftarrow \cdots \leftarrow v$, we simply start from $v$ and follow the parent pointers until we reach a node $w_1$ that is a central node in $P_i$.  The correctness of the construction of this part follows from the definition of $G$ and the fact that either $v$ is a rear node in $P_i$ or $v \notin P_i$ is in layer $(\sC,i)$ or above.

To construct the remaining part $v \leftarrow u \leftarrow \cdots \leftarrow w_2$, we simply observe that either $u$ itself is a central node in layer $(\sC,i)$ or there is a  descendant $w_2$ of $u$  such that $w_2$ is a central node in layer $(\sC,i)$. Hence we can always extend $v \leftarrow u$ to a desired path $v \leftarrow u \leftarrow \cdots \leftarrow w_2$.
\end{proof}

We note that the reason for attaching the rooted trees
$G_{\sC,1}^\circ, G_{\sC, 2}^\circ, \ldots, G_{\sC, k}^\circ$ to $G_{\sR,k+1}$ in the definition of $G$ is precisely that we want to define $S_{\sC,i}$ in such a way that allows us to have \cref{lem-layer-C-property-2}. That is, the design objective is to ensure that for each node $v \in S_{\sC,i}$ that is not a central node in layer $(\sC,i)$, there is a directed path in $S_{\sC,i}$ passing through $v$ and starting and ending at central nodes in layer $(\sC,i)$.

\subparagraph{Assumptions} We are given an $\LCL$ problem $\Pi=(\delta, \Sigma, \CC)$ such that $d_\Pi = k$. Hence there does not exist a good sequence 
\[(\SigmaR{1}, \SigmaC{1}, \SigmaR{2}, \SigmaC{2}, \ldots, \SigmaR{k}).\]
Recall that the rules for a good sequence are as follows:
\begin{align*}
    \SigmaR{i} & =
    \begin{cases}
    \trim(\Sigma) & \text{if $i=1$},\\
    \trim(\SigmaC{i-1}) &  \text{if $i>1$},\\
    \end{cases}\\
    \SigmaC{i} &\in \flexSCC(\SigmaR{i}).
\end{align*}

The only nondeterminism in the above rules is the choice of $\SigmaC{i} \in \flexSCC(\SigmaR{i})$ for each $i$. The fact that $d_\Pi = k$ implies that for all possible choices of $\SigmaC{1}, \SigmaC{2}, \ldots, \SigmaC{k}$, we always end up with $\SigmaR{k+1} = \emptyset$.

We assume that there is an algorithm $\A$ that solves $\Pi$ in $t = O(n^{1/k})$ rounds on  $G = G_{\sR,k+1}^\ast$, where $n$ is the number of nodes in $G$. To prove the desired $\Omega(n^{1/k})$ lower bound, it suffices to derive a contradiction. Specifically, we will prove that the existence of such an algorithm $\A$ forces the existence of a good sequence $(\SigmaR{1}, \SigmaC{1}, \SigmaR{2}, \SigmaC{2}, \ldots, \SigmaR{k})$, contradicting the fact that  $d_\Pi = k$. 

\subparagraph{Induction hypothesis} Our proof proceeds by an induction on the subsets $S_{\sR,1}$, $S_{\sC,1}$, $S_{\sR,2}$, $S_{\sC,2}$, $\ldots$, $S_{\sR,k+1}$. For each $1 \leq i \leq k$, the choice of $\SigmaC{i} \in \flexSCC(\SigmaR{i})$ is fixed in the induction hypothesis for $S_{\sC,i}$. The choice of $\SigmaR{i}$ is uniquely determined once $\SigmaC{1}, \SigmaC{2}, \ldots, \SigmaC{i-1}$ have been fixed.

Similar to \cref{sect-lower}, before defining our induction hypothesis, we recall that the output label of a node $v$ is determined by the subgraph induced by the radius-$t$ neighborhood $U$ of $v$, together with the distinct IDs of the nodes in $U$. For each node $v$ in $G$, we define $\Sigma_v^\A$ as the set of all possible output labels of $v$ that can possibly appear when we run $\A$ on $G$.  In other words, $\sigma \in \Sigma_v^\A$ implies that there exists an assignment of distinct IDs to nodes in the radius-$t$ neighborhood of $v$ such that the    output label of  $v$ is $\sigma$.

\begin{definition}[Induction hypothesis for layer $S_{\sR,i}$]\label{def-IH-R-2}
For each $1 \leq i \leq k+1$, the induction hypothesis for $S_{\sR,i}$ specifies that each $v \in S_{\sR,i}$ satisfies  $\Sigma_v^\A \subseteq \SigmaR{i}$.
\end{definition}

\begin{definition}[Induction hypothesis for layer $S_{\sC,i}$]\label{def-IH-C-2}
For each $1 \leq i \leq k$, the induction hypothesis for $S_{\sC,i}$ specifies that  there exists a choice $\SigmaC{i} \in \flexSCC(\SigmaR{i})$ such that each $v \in S_{\sC,i}$ satisfies  $\Sigma_v^\A \subseteq \SigmaC{i}$.
\end{definition}

Next, we prove that the induction hypotheses stated in \cref{def-IH-R-2,def-IH-C-2} hold.

\begin{lemma}[Base case: $S_{\sR, 1}$]\label{lem-base-case-2}
The  induction hypothesis for $S_{\sR, 1}$ holds.
\end{lemma}
 \begin{proof}
 Recall that the number $\gamma$ satisfies the following property. 
  For each subset $\tilde{\Sigma}\subseteq \Sigma$ and each $\sigma \in \tilde{\Sigma} \setminus \trim(\tilde{\Sigma})$, there exists no correct labeling of $T_\gamma$ where the label of the root $r$ is $\sigma$ and the  label of remaining nodes is in $\tilde{\Sigma}$.

 To prove the induction hypothesis for $S_{\sR, 1}$, consider any node $v \in S_{\sR, 1}$. By the definition of $S_{\sR, 1}$, the subgraph induced by $v$ and its descendants within radius-$\gamma$ of $v$ is isomorphic to $T_\gamma$ rooted at $v$. By setting $\tilde{\Sigma} = \Sigma$, we infer that  there is no correct labeling of $G$ such that the label of $v$ is in $\Sigma \setminus \trim(\Sigma) = \Sigma \setminus \SigmaR{1}$, so we must have $\VV_v^\A \subseteq \SigmaR{1}$, as $\A$ is correct.
 \end{proof}

\begin{lemma}[Inductive step: $S_{\sR,i}$]\label{lem-inductive-step-CR-2}
 Let $2 \leq i \leq k$. If the induction hypothesis for $S_{\sC,i-1}$ holds, then the  induction hypothesis for $S_{\sR,i}$ holds.
\end{lemma}
 \begin{proof}
  To prove the induction hypothesis for  $S_{\sR,i}$, consider any node $v \in S_{\sR,i}$. By the definition of $S_{\sR,i}$, the subgraph $S$ induced by $v$ and its descendants within radius-$\gamma$ of $v$ is isomorphic to $T_\gamma$ rooted at $v$ and contains only nodes in $S_{\sC,i-1}$. Our goal is to prove that $\VV_v^\A \subseteq \SigmaR{i} = \trim(\SigmaC{i-1})$.
  
  Consider any node $u$ in the subgraph $S$ induced by $v$ and its descendants within radius-$\gamma$ of $v$.
  As $u \in S_{\sC,i-1}$, the induction hypothesis for  $S_{\sC,i-1}$  implies that 
   \[\VV_u^\A \subseteq \SigmaC{i-1}.\]
  Hence the same argument in the proof of \cref{lem-base-case-2} shows that $\VV_v^\A \subseteq \trim(\SigmaC{i-1}) = \SigmaR{i}$, as required. 
 \end{proof}
 
 \begin{lemma}[Inductive step: $S_{\sC,i}$]\label{lem-inductive-step-RC-2}
 Let $1 \leq i \leq k$. If the induction hypothesis for layer $S_{\sR,i}$ holds, then the  induction hypothesis for layer $S_{\sC,i}$ holds.
\end{lemma}
 \begin{proof}
To prove the induction hypothesis for  $S_{\sC,i}$, we show that there exists a choice $\SigmaC{i} \in \flexSCC(\SigmaR{i})$  such that each $v \in S_{\sC,i}$ satisfies  $\Sigma_v^\A \subseteq \SigmaC{i}$.

We first consider the case that $v$ is a central node in layer $(\sC,i)$.  
Then it is clear that $\Sigma_v^\A$ is the same for all $v$ that is a central node in layer $(\sC,i)$, as the radius-$t$ neighborhood of these nodes $v$ are isomorphic, due to the definition of central nodes and the construction in \cref{def-lb-graphs-2}.
 For notational convenience, we write $\tilde{\Sigma}$ to denote the set  $\Sigma_v^\A$ for any central node $v$ in layer $(\sC,i)$.

\subparagraph{Plan of the proof} Consider any node $v \in S_{\sC,i}$, we claim that the node configuration $(\sigma \ : \ a_1 a_2 \cdots a_\delta)$ of $v$ resulting from running $\A$ uses only labels in  $\SigmaR{i}$. To see this, observe that  $v$ and all $\delta$ children of $v$ are in $S_{\sR,i}$, so the induction hypothesis for $S_{\sR,i}$ implies that their output labels must be in $\SigmaR{i}$. Hence all labels in $(\sigma \ : \ a_1 a_2 \cdots a_\delta)$ are in $\SigmaR{i}$.  

We write $\M$ to denote the directed graph representing the automaton associated with the path-form of the $\LCL$ problem $\Pi \upharpoonright_{\SigmaR{i}}$. The above claim implies the following. Consider any directed edge $u \leftarrow v$ such that both $u$ and $v$ are in $S_{\sC,i}$. Let $a$ be the output label of $u$ and let $b$ be the output label of $v$. Then $a \leftarrow b$ must be a directed edge in $\M$.

To prove that  there exists a choice $\SigmaC{i} \in \flexSCC(\SigmaR{i})$ such that $\Sigma_v^\A \subseteq \SigmaC{i}$ for all  $v\in S_{\sC,i}$. We will first show that $\tilde{\Sigma}$ must be a subset of a path-flexible strongly connected component of $\SigmaR{i}$, and then we fix $\SigmaC{i} \in \flexSCC(\SigmaR{i})$ to be this path-flexible strongly connected component.

Next, we will argue that for each $\sigma \in \Sigma_v^\A$, there exist a walk in $\M$ that starts from $\sigma$ and ends in $\tilde{\Sigma}$ and a walk in $\M$ that starts from $\tilde{\Sigma}$ and ends in $\sigma$. This shows that $\sigma$ is in the same strongly connected component as the members in $\tilde{\Sigma}$, so we conclude that $\Sigma_v^\A \subseteq \SigmaC{i}$.

\subparagraph{Part 1: $\tilde{\Sigma}$ is a subset of some $\SigmaC{i} \in \flexSCC(\SigmaR{i})$}
Consider a path $u_1 \leftarrow u_2  \leftarrow \cdots  \leftarrow u_s$ of $s$ nodes in layer $(\sC,i)$ of  $G$. 
 We choose $s = \Theta(t)$ to be sufficiently large to ensure that for each integer $0 \leq d \leq |\Sigma|$, there exist two nodes $u_j$ and $u_l$ in the path meeting the following conditions.
 \begin{itemize}
     \item $t+1 \leq j < l \leq s-t$, so $u_j$ and $u_l$ are central nodes.
     \item The distance $l-j$ between $u_j$ and $u_l$  equals $4t+3+d$. 
 \end{itemize}
 
 Similar to the proof of \cref{lem-inductive-step-RC}, the choice of the number $4t+3$ is to ensure that the union of the radius-$t$ neighborhood of any node $v$ and the radius-$t$ neighborhood of children of $v$ in $G$ does not node-intersect the radius-$t$ neighborhood of both $u_j$ and $u_l$.
 This implies that after arbitrarily fixing distinct IDs in the radius-$t$ neighborhood of $u_j$ and $u_l$, it is possible to complete the ID assignment of the entire graph $G$ in such a way that  the union of the radius-$t$ neighborhood of any node $v$ and the radius-$t$ neighborhood of children of $v$ in $G$ does not contain repeated IDs. If we run $\A$ with such an ID assignment, it is guaranteed that the output is correct.
 
 Consider the directed graph $\M$ and any of its two nodes $a$ and $b$ such that $a \in \tilde{\Sigma}$ and $b \in \tilde{\Sigma}$. Our choice of $\tilde{\Sigma}$ implies that there exists an assignment of distinct IDs to the radius-$t$ neighborhood of both $u_j$ and $u_l$ such that the output labels of $u_j$  and $u_l$ are $a$ and $b$, respectively.
 We complete the ID assignment of the entire graph $G$ in such a way that the union of the radius-$t$ neighborhood of any node $v$ and the radius-$t$ neighborhood of children of $v$  in $G$ does not contain repeated IDs.
 
 We write $\sigma_y$ to denote the output label of 
 $u_y$ resulting from running $\A$.  Hence  $a = \sigma_j \leftarrow \sigma_{j+1} \leftarrow \cdots \leftarrow \sigma_{l}=b$ is a walk $b \leadsto a$ in $\M$ of length $4t+3+d$. Therefore, for all choices of  $a \in \tilde{\Sigma}$ and $b \in \tilde{\Sigma}$ and  any $0 \leq d \leq |\Sigma|$, we may find a walk $b \leadsto a$ in $\M$ of length $4t+3+d$. This implies that all members in $\tilde{\Sigma}$ are in the same strongly connected component of $\M$.
 Furthermore, \cref{lem-inflex-2} implies that this strongly connected component is path-flexible, as the length of the walk can be any integer in between $4t+3$ and $4t+3+|\Sigma|$.

\subparagraph{Part 2:  $\Sigma_{v}^\A \subseteq \SigmaC{i}$ for each $v \in S_{\sC,i}$} For this part, 
we use \cref{lem-layer-C-property-2}, which shows that for each $v \in S_{\sC,i}$ in the graph  $G$,
there is a directed path $P = w_1 \leftarrow \cdots \leftarrow v \leftarrow \cdots \leftarrow w_2$ such that $w_1 \in P_i$ and $w_2 \notin P_i$ are central nodes in layer $(\sC,i)$ and all nodes in $P$ are in $S_{\sC,i}$.

Consider the output labels of the nodes in $P$ resulting from running $\A$. We write $\sigma_v$ to denote the output label of $v$.
Then $\sigma_{w_1} \leftarrow \cdots \leftarrow \sigma_v$ is a walk in $\M$ from $\sigma_v$ to a node in $\tilde{\Sigma}$ and  $\sigma_{v} \leftarrow \cdots \leftarrow \sigma_{w_2}$ is a walk in $\M$ from  a node in $\tilde{\Sigma}$ to $\sigma_v$. This shows that $\sigma_v$ is in the same strongly connected component of $\M$ as the members in $\tilde{\Sigma}$.

The same argument can be applied to all $\sigma_v \in \Sigma_{v}^\A$. The reason is that for each $\sigma_v \in \Sigma_{v}^\A$  there is an assignment of distinct IDs such that $\sigma_v$ is the output label of $v$. Hence we conclude that all members in $\Sigma_{v}^\A$ are in the same strongly connected component of $\M$ as the members in $\tilde{\Sigma}$, so  $\Sigma_{v}^\A \subseteq \SigmaC{i}$.
\end{proof}
 
 Applying \cref{lem-base-case-2,lem-inductive-step-CR-2,lem-inductive-step-CR-2,lem-inductive-step-RC-2} from $S_{\sR,1}$ all the way up to the last subset  $S_{\sR,k+1}$, we obtain the following result.
 
\begin{lemma}[Lower bound for the case $d_\Pi = k$]\label{lem-lower-finite-2}
If $d_\Pi = k$ for a finite integer $k$, then  $\Pi$ requires $\Omega(n^{1/k})$ rounds to solve on rooted trees of maximum indegree $\delta$.
\end{lemma}
\begin{proof}
Assume that there is a $t$-round algorithm solving $\Pi$ on $G$.
By \cref{lem-base-case-2,lem-inductive-step-CR-2,lem-inductive-step-CR-2,lem-inductive-step-RC-2}, we infer that the induction hypothesis for the last subset  $S_{\sR,k+1}$ holds. By \cref{lem-containment-2}, $S_{\sR,k+1} \neq \emptyset$, so there is a node $v$ in $G$  such that  $\Sigma_v^\A \subseteq \SigmaR{k+1}$. Therefore, the correctness of $\A$ implies that $\SigmaR{k+1} \neq \emptyset$, which implies that $(\SigmaR{1}, \SigmaC{1}, \SigmaR{2}, \SigmaC{2}, \ldots, \SigmaR{k+1})$ chosen in the induction hypothesis is a good sequence, contradicting the assumption that $d_\Pi = k$.
Hence such a $t$-round algorithm $\A$ that solves $\Pi$ does not exist. As $t$ can be any positive integer and $t = \Omega(n^{1/k})$, where $n$ is the number of nodes in $G$, we conclude the proof.
\end{proof}

Now we are ready to prove \cref{thm-unrooted-characterization-2}.

\begin{proof}[Proof of \cref{thm-unrooted-characterization-2}]
The upper bound part of the theorem follows from \cref{lem-upper-finite-2,lem-upper-infinite-2}. The lower bound part of the theorem follows from \cref{lem-lower-0-2,lem-lower-finite-2}.
\end{proof}

\subsection{Complexity of the characterization}\label{sect-time-2}

In this section, we prove \cref{thm-unrooted-poly-time-2}. We are given a description of an $\LCL$ problem $\Pi=(\delta, \Sigma,\CC)$ on $\delta$-regular rooted trees. We assume that the description is given in the form of listing all the node configurations in $\CC$. Therefore, the description length of $\Pi$ is $\ell = O(|\CC| \delta \log |\Sigma|)$. We allow $\delta$ to be a non-constant as a function of $\ell$. We will design an algorithm that computes all possible good sequences $(\SigmaR{1}, \SigmaC{1}, \SigmaR{2}, \SigmaC{2}, \ldots, \SigmaR{k})$  in time polynomial in $\ell$.

For the case of $d_\Pi = \infty$, there are good sequences that are arbitrarily long. Recall that we have $\SigmaR{1} \supseteq \SigmaC{1} \supseteq \cdots \supseteq \SigmaR{k}$. Hence if $k > |\Sigma|$, there must exist some index $1 \leq i < k$ such that $\SigmaR{i}  = \SigmaC{i}  = \SigmaR{i+1}$. This immediately implies that $\SigmaR{i}  = \SigmaC{i}  = \SigmaR{i+1} = \SigmaC{i+1} = \cdots$. The reason is that $\SigmaR{i}  = \SigmaC{i}$ implies that $\SigmaR{i}$ itself is  the only element of $\flexSCC(\SigmaR{i})$. We conclude that any good sequence with $k > |\Sigma|$ must stabilizes at some point $i \leq |\Sigma|$, in the sense that $\SigmaR{i}  = \SigmaC{i}  = \SigmaR{i+1} = \SigmaC{i+1} = \cdots$. 

\subparagraph{High-level plan}
Recall that the rules for a good sequence are as follows.
\begin{align*}
    \SigmaR{i} & =
    \begin{cases}
    \trim(\Sigma) & \text{if $i=1$},\\
    \trim(\SigmaC{i-1}) &  \text{if $i>1$},\\
    \end{cases}\\
    \SigmaC{i} &\in \flexSCC(\SigmaR{i}).
\end{align*}

To compute all  good sequences $(\SigmaR{1}, \SigmaC{1}, \SigmaR{2}, \SigmaC{2}, \ldots, \SigmaR{k})$, we go through  all choices of $\SigmaC{i} \in \flexSCC(\SigmaR{i})$ and apply the rules recursively until we cannot proceed any further. The process stops when $\SigmaC{i} = \SigmaR{i}$ (the sequence stabilizes) or $\SigmaC{i} = \emptyset$ or $\SigmaR{i} = \emptyset$ (the sequence ends).

We start with describing the algorithm for computing $\trim(\tilde{\Sigma})$ for a given $\tilde{\Sigma} \subseteq \Sigma$.

\begin{lemma}[Algorithm for $\trim$]\label{lem-compute-trim-2}
The set $\trim(\tilde{\Sigma})$ can be computed in $O(|\CC| \delta |\tilde{\Sigma}| + |\tilde{\Sigma}|^2)$ time, for any given $\tilde{\Sigma} \subseteq \Sigma$.
\end{lemma}
\begin{proof}
  We write $\Sigma_i$ to denote set of all possible  $\sigma \in \tilde{\Sigma}$ such that there is a correct labeling of the rooted tree $T_i$ where the label of each  node is in $\tilde{\Sigma}$ and the  label of the root $r$ is $\sigma$. The set $\Sigma_i$ can be computed recursively as follows.
\begin{itemize}
    \item For the base case,  $\Sigma_0$ is the set of all labels appearing in $\tilde{\Sigma}$. 
    \item For the inductive step,  each $\sigma \in \Sigma_{i-1}$ is added to $\Sigma_{i}$ if there exists a node configuration $(\sigma \ : \ a_1 a_2 \cdots a_\delta) \in \CC$ such that $a_j \in \Sigma_{i-1}$ for all $1 \leq j \leq \delta$.
\end{itemize}

The above recursive computation implies that given $\Sigma_{i-1}$ has been computed, the computation of  $\Sigma_{i}$ costs $O(|\CC| \delta + |\tilde{\Sigma}|)$ time. The algorithm simply goes over each $(\sigma \ : \ a_1 a_2 \cdots a_\delta) \in \CC$ and checks whether $a_j \in \Sigma_{i-1}$ for all $1 \leq j \leq \delta$. If we store $\Sigma_{i-1}$ as binary string of length $|\tilde{\Sigma}|$, testing whether $a_j \in \Sigma_{i-1}$ costs $O(1)$ time. Therefore, the process of going through all $(\sigma \ : \ a_1 a_2 \cdots a_\delta) \in \CC$ costs $O(|\CC| \delta)$ time. After that, using $O(|\CC| + |\tilde{\Sigma}|)$ time, we may calculate $\Sigma_{i}$ and store it as a binary string of length $|\tilde{\Sigma}|$.

Clearly, we have $\Sigma_1 \supseteq \Sigma_2  \supseteq \cdots$. Once $\Sigma_i = \Sigma_{i+1}$, the sequence stabilizes:  $\Sigma_i = \Sigma_{i+1} = \Sigma_{i+2} = \cdots$. It is clear that the sequence stabilizes at some $i \leq |\tilde{\Sigma}|$.  We write $\Sigma^\ast$ to denote the fix point $\Sigma_i$ such that $\Sigma_i = \Sigma_{i+1} = \Sigma_{i+2} = \cdots$.
It is clear that $\trim(\tilde{\Sigma}) = \Sigma^\ast$, and it can be computed in $O(|\CC| \delta + |\tilde{\Sigma}|) \cdot |\tilde{\Sigma}| =   O(|\CC| \delta |\tilde{\Sigma}| + |\tilde{\Sigma}|^2)$ time.
\end{proof}

  Next, we give an algorithm that computes all path-flexible strongly connected components $\Sigma' \in \flexSCC(\tilde{\Sigma})$, for any given $\tilde{\Sigma}\subseteq \Sigma$.

\begin{lemma}[Algorithm for $\flexSCC$]\label{lem-compute-flex-SCC-2}
The set of all $\Sigma' \in \flexSCC(\tilde{\Sigma})$ can be computed in $O(|\tilde{\Sigma}|^3 + |\CC|\delta)$  time, for any given $\tilde{\Sigma}\subseteq \Sigma$.
\end{lemma}

\begin{proof}
Let $\M$ be the directed graph representing the automaton associated with the path-form of the $\LCL$ problem $\Pi \upharpoonright_{\tilde{\Sigma}}$. 
The directed graph $\M$  has $|\tilde{\Sigma}|$ nodes and   $O(|\tilde{\Sigma}|^2)$ directed edges, and it can be constructed in $O(|\CC|\delta)$ time by going through all node configurations in $\CC$.
The strongly connected components of $\M$ can be computed in time linear in the number of nodes and edges of $\M$, which is $O(|\tilde{\Sigma}|^2)$.

By to the proof of \cref{lem-compute-flex-SCC},
for each strongly connected component $U$ of $\M$, deciding whether $U$ is path-flexible costs $O(|U|^3)$ time.  The summation of the time complexity $O(|U|^3)$, over all strongly connected component $U$ of $\M$, is $O(|\tilde{\Sigma}|^3)$.
To summarize, all $\Sigma' \in \flexSCC(\tilde{\Sigma})$ can be computed in $O(|\tilde{\Sigma}|^3 + |\CC|\delta)$ time.
\end{proof}

 Combining \cref{lem-compute-trim-2,lem-compute-flex-SCC-2}, we obtain the following result.

\begin{lemma}[Computing all good sequences]\label{lem-compute-good-seq-2}
The set of all good sequences can be computed in $O(|\Sigma|^4 + |\CC|\delta |\Sigma|^2)$ time, for any given $\LCL$ problem $\Pi=(\delta, \Sigma, \CC)$.
\end{lemma}
\begin{proof}
By \cref{lem-compute-trim-2}, given $\SigmaC{i} \subseteq \Sigma$, the cost of computing $\SigmaR{i+1}$ is $O(|\CC| \delta |\SigmaC{i}| + |\SigmaC{i}|^2)$ time. Since all  sets $\SigmaC{i}$ in the depth $i$ of the recursion are disjoint, the total cost for this step of the recursion is $O(|\CC| \delta |\Sigma| + |\Sigma|^2)$.

By \cref{lem-compute-flex-SCC-2}, given $\SigmaR{i} \subseteq \Sigma$, the cost of computing all possible $\SigmaC{i}$ is $O(|\SigmaR{i}|^3 + |\CC|\delta)$ time.  Since all  sets $\SigmaR{i}$ in the depth $i$ of the recursion are disjoint, the total cost for this step of the recursion is $O(|\Sigma|^3 + |\CC|\delta |\Sigma|)$. 

The depth of the recursion is at most $|\Sigma|$, so the total cost of computing all good sequences is $O(|\Sigma|^4 + |\CC|\delta |\Sigma|^2)$.
\end{proof}

We are ready to prove  \cref{thm-unrooted-poly-time-2}.

\begin{proof}[Proof of \cref{thm-unrooted-poly-time-2}]
By \cref{lem-compute-good-seq-2}, the set of all good sequences can be computed in polynomial time, and we can compute $d_\Pi$ given the set of all good sequences.  If $d_\Pi = k$ is a positive integer, then from the discussion in \cref{sect-upper-2} we know how to turn a good sequence
$(\SigmaR{1}, \SigmaC{1}, \SigmaR{2}, \SigmaC{2}, \ldots, \SigmaR{k})$ into a description of an  $O(n^{1/k})$-round algorithm for $\Pi$. If $d_\Pi = \infty$, then similarly a good sequence
$(\SigmaR{1}, \SigmaC{1}, \SigmaR{2}, \SigmaC{2}, \ldots, \SigmaR{O(\log n)})$ leads to a description of an $O(\log n)$-round algorithm for $\Pi$.
\end{proof}

\bibliography{references}

\end{document}